\documentclass[runningheads,11pt]{llncs}
\usepackage{array,xspace,multirow,hhline,graphicx,tikz,colortbl,tabularx,amsmath,amssymb,amsfonts}
\AtEndEnvironment{proof}{\qed}
\usepackage{thm-restate,thmtools}
\usepackage{xcolor}
\usepackage{booktabs,palatino}

\usepackage[sort,numbers]{natbib}
\usepackage{authblk}

\makeatletter 
\renewcommand\@biblabel[1]{#1} 
\makeatother

\usepackage{xfrac}
\usepackage{caption}
\usepackage{subcaption}
\usepackage{collcell} 

\usepackage{wrapfig} 
\usepackage{bbm}  
\usepackage{verbatim,ifthen} 
\usepackage{pifont} 

\usepackage{nicefrac}  
\usepackage[normalem]{ulem}
\usepackage{varioref}

\usepackage{calc}
\newsavebox\CBox
\newcommand\hcancel[2][0.5pt]{%
	\ifmmode\sbox\CBox{$#2$}\else\sbox\CBox{#2}\fi%
	\makebox[0pt][l]{\usebox\CBox}%
	\rule[0.5\ht\CBox-#1/2]{\wd\CBox}{#1}}

\tikzset{
	jumpdot/.style={mark=*,solid},
	excl/.append style={jumpdot,fill=white},
	incl/.append style={jumpdot,fill=black},
	rexcl/.append style={jumpdot,color=red,fill=white},
	rincl/.append style={jumpdot,fill=black,color=red},
} 
\renewcommand\small{\fontsize{10pt}{12pt}\selectfont}
\usepackage[ruled,vlined,linesnumbered]{algorithm2e} 
\SetKwComment{Comment}{/* }{ */}

\SetAlFnt{\small}
\SetAlCapFnt{\small}
\SetAlCapNameFnt{\small}
\SetAlCapHSkip{0pt}
\IncMargin{-\parindent}
\DeclareMathOperator*{\argmax}{arg\,max}
\DeclareMathOperator*{\argmin}{arg\,min}
\newcommand{\USC}[1]{\ifstrempty{#1}{\textrm{\textup{USC}}}{#1\textrm{\textup{-USC{}}}}}
\newcommand{\ESC}[1]{\ifstrempty{#1}{\textrm{\textup{ESC}}}{#1\textrm{\textup{-ESC{}}}}}
\newcommand{\USW}[1]{\ifstrempty{#1}{\textrm{\textup{USW}}}{#1\textrm{\textup{-USW{}}}}}
\newcommand{\ESW}[1]{\ifstrempty{#1}{\textrm{\textup{ESW}}}{#1\textrm{\textup{-ESW{}}}}}
\newcommand{\PO}{\textup{PO}}
\newcommand{\EFone}{\textrm{\textup{EF1}}}
 
\newcommand{\EFX}{\textrm{\textup{EFX}}\xspace}

\renewcommand{\emptyset}{\varnothing} 

\spnewtheorem{observation}{Observation}{\bfseries}{\itshape}

\usepackage{mathtools}

\definecolor{gray(x11gray)}{rgb}{0.75, 0.75, 0.75}

\usepackage{boxedminipage}
\usepackage{xspace}

\usepackage{xr}

\usepackage[capitalise,noabbrev]{cleveref}

\usepackage{geometry}
\Crefname{observation}{Observation}{Observations}   

\begin{document}

	\title{Maximum Welfare Allocations under Quantile Valuations}
	
	%
	
	\author{Haris Aziz, Shivika Narang, Mashbat Suzuki}
	
	%
	%
	%
	\institute{UNSW Sydney
		\email{\{haris.aziz, s.narang, mashbat.suzuki\} @unsw.edu.au}}

	\maketitle              
	%
	

	\begin{abstract}
		We propose a new model for aggregating preferences over a set of indivisible items based on a quantile value. In this model, each agent is endowed with a specific quantile, and the value of a given bundle is defined by the corresponding quantile of the individual values of the items within it.
		Our model captures the diverse ways in which agents may perceive a bundle, even when they agree on the values of individual items.  It enables richer behavioral modeling that cannot be captured by additive valuation functions. 
		We aim to maximize utilitarian and egalitarian welfare within the quantile-based valuation setting. For each of the welfare functions, we analyze the complexity and provide complementary approximation and exact algorithms. Interestingly, our results show that the complexity of both functions varies significantly, depending on whether the allocation is required to be balanced.  We provide near-optimal approximation algorithms for utilitarian welfare, and for egalitarian welfare, we present exact algorithms for many cases.
	\end{abstract}

	\section{Introduction} 
	The problem of allocating indivisible items in a fair and efficient manner has been well-studied in recent years \citep{Feig2006maximizing,feige2010submodular,budish2011combinatorial,CKM+2019unreasonable,AFSV2024best}. The overwhelming majority of this work focuses on settings where agents have monotone valuations for the items being assigned\footnote{A monotone valuation function is one where the marginal value of an item is always non-negative or always non-positive.} partially because of the underlying structure imposed by them. In practice however, agent preferences may be unreasonably non-monotone, and alternate models of preferences are needed. %
	Consider a setting where an incoming class of school students need to be divided into sections each with a different teacher. As students have little to no information to compare the teachers, we can assume they have no preferences over specific teachers. Meanwhile, teachers get satisfaction from their students' learning and growth throughout the year. Each teacher may have different levels of satisfaction with a given set of students assigned to them, even if the students were to perform similarly. One teacher may be upset if even one student performs poorly, whereas a different teacher may be satisfied if at least half their class does well. Some teachers may be delighted if they have even one exceptional child in their class, even if the others do not do particularly well. 
	
	These non-monotone ways of aggregating preferences cannot be modeled by existing valuation function classes such as additive, subadditive or even super additive valuations. We can capture these opinions by different quantile values for the set of students assigned. The pessimistic/critical teacher bases their satisfaction on the lowest quantile. The teacher whose satisfaction is based on at least half the class doing well can be captured by the median quantile.  The optimistic teacher bases their satisfaction on the highest quantile. As a result, finding allocations of students to teachers that maximize teacher welfare can be non-trivial. School teachers are some of the most overworked and underpaid professionals. Consequently, it is imperative to  ensure good allocations for them. 
	
	\paragraph{Quantile Valuations.} We introduce a novel valuation class, termed \textit{ quantile valuations}, which encompasses the aforementioned scenarios. In this framework, each agent is endowed with a specific quantile value  $\tau\in[0,1]$, and the value that she assigns for a bundle $S$ is the $\tau$-quantile of the distribution of item values in $S$. 
	Quantiles are widely used across data analysis and statistics because they provide a robust description of value distributions. Compared to measures like average or total/gross, most quantile based measures are significantly less susceptible to outliers. As a result, quantiles are commonly used in practical settings to aggregate population behavior. Different population bodies, such as cities, states and countries, are compared by quantile-based measures such as median household income, median age and median house price/rent\footnote{https://worldpopulationreview.com/country-rankings/median-income-by-country}. 
	Quantiles have also been used in decision theory to model agent preferences in settings where agents have preferences over stochastic outcomes. Specifically, quantiles have been considered in settings where an agent faces a choice of actions, each yielding a distribution over outcomes. Here, modeling the agent's choice as a quantile maximizer has been shown to provide a better approximation of human behavior than modeling them as an  expected utility maximizer \citep{DeGa2019dynamic,DeGa2022static}. 

	Within settings where an agent will likely be allocated a (nearly) fixed number of items, quantile offer a useful way to aggregate preferences. Such settings include allocating repair tasks to workers (where a fixed number of working hours may limit the number of tasks assigned) and assigning submitted conference papers to relevant reviewers.  Even typical school settings with in-person classes require that each teacher be assigned an equal sized set of students. In contrast, when it comes to online classes such as those for distance learning or preparatory classes for specific exams, there is no underlying restriction on class size. Analogously, journal editors may receive varied number of submissions which they further assign to reviewers to work on. Often due to conflicts of interests, some editors may not be able to be assigned a majority of the submissions, and a large disparity in their allocation size may be unavoidable. 
	Hence, we consider both the space of {\em balanced} allocations as well as that of all {\em unconstrained} allocations.
	
	\paragraph{Fairness and Efficiency under Quantile Valuations.} Welfare functions have been used to quantify economically efficient outcomes across a multitude of domains including fair division\cite{barman2021sublinear,CKM+2019unreasonable,barman2020uniform,barman2020tight, CaNa2024repeatedly},  social choice \cite{aziz2019fair,suzuki2024maximum}, combinatorial auctions \cite{bhawalkar2011welfare,badanidiyuru2012optimization,Feig2006maximizing} and many other \cite{barman2020uniform,murhekar2023incentives,blumrosen2006welfare}. All solutions that maximize welfare are also Pareto optimal and some, such as Nash welfare have also found to produce fair solutions in the item allocation setting. Egalitarian welfare, in particular, looks at maximizing the value attained by the worst-off agent. \ESW{} captures Rawlsian fairness \cite{rawls1991justice} or the ``santa claus problem''\cite{bansal2006santa,springer2022online,feige2008allocations}, a well-motivated and well-studied way to capture fairness. In contrast, utilitarian welfare, which maximizes the sum of the agents' values, is the most widely adapted welfare function and naturally captures (economic) efficiency. Consequently, these two objectives become natural choices to capture desirable outcomes under our proposed model of quantile-based valuations. 
	%
	%

	
	\subsection{Our Results}
	
	We study the problem of maximizing welfare for agents with quantile valuations. 
	Under quantile valuations, each agent $i$ specifies their value for individual items and a quantile value $\tau_i\in [0,1]$. Given a bundle $B$, agent $i$'s value for $B$ is the $\tau_i$th quantile of the values of the items in $B$. We provide comprehensive results on  {\em utilitarian social welfare }(\USW{}) (see for e.g. \citet{Hars1955cardinal}), which captures efficiency, and {\em egalitarian social welfare} (\ESW{}) (see for e.g. \citet{Moul2004fair}), which also captures fairness. \USW{} is essentially the sum of the agents' valuations for their assigned bundles, and \ESW{} is the minimum of these values.
	We study each objective for both balanced and unconstrained allocations, with our results for balanced allocations being given in \cref{sec:bal} and those for unconstrained allocation being given in \cref{sec:unbal}. 
	A balanced allocation is one where bundle sizes are as equitable as possible, i.e., each agent receives either $\lfloor m/n\rfloor$ or $\lceil m/n \rceil$ items, where  $m$ is the number of items and $n$ is the number of agents $n$.
	Our results are summarized in \cref{tab:contributions}. 

	{\small
		\begin{table*}[t]
			\centering
			\crefname{theorem}{Th.}{Th.}
			\crefname{proposition}{Prop.}{Prop.}
			\crefname{corollary}{Cor.}{Cor.}
			\begin{tabular}{lllllll}
				\toprule 
				&                                           &  \textbf{General}                                  &                                              & \textbf{Identical}    \\
				\midrule         
				\multirow{3}{*}{\textbf{Balanced}}  &\multirow{2}{*}{\textbf{USW}} \quad \quad  &   APX-h                            &(\cref{thm:hardnessUSW:Balanced})\quad\quad   & \multirow{2}{*}{in P} &\multirow{2}{*}{(\cref{thm:USWbalIdent})}  \\ 
				&                                           & $\USW{\min(\lceil\frac{m}{n}\rceil+1,n)}^{\dagger}$&(\cref{thm:balUSWgreedy}) 			&  \\
				&\textbf{ESW}                               & in  P$^{\dagger}$                                  &(\cref{thm:balESW})                           & in P                  &(\cref{thm:balESW}) \\ 
				\midrule
				\multirow{4}{*}{\textbf{Unconstrained}}\quad \quad \quad &\multirow{2}{*}{\textbf{USW}}&NP-h                                               & (\cref{thm:NP-USW})                          & \multirow{2}{*}{?} \\
				&                                        & $\USW{(\frac{n}{n-1})}^{\dagger}$                 &(\cref{thm:scapegoat})   \\
				&\multirow{2}{*}{\textbf{ESW}}               &   APX-h                                           &(\cref{thm:ESWunbalNP})                       &  \multirow{2}{*}{in P} &\multirow{2}{*}{(\cref{thm:ESWident})}\\
				&                                            &  in  P for $\tau \in \{0,\sfrac{1}{3},\frac{t}{t+1},1\}$, $t\in \mathbb{Z}_+$ \quad \quad &(\cref{thm:ESWunbal})               \\
				\bottomrule
				
			\end{tabular}
			\caption{Complexity of computing \USW{} and \ESW{} optimal allocations given $n$ agents and $m$ items. $\USW{\alpha}$ refers to an $\alpha$ approximation to the optimal \USW{}. Algorithmic results marked with $\dagger$ hold even when agents have heterogeneous quantiles. All other algorithmic results require agents to have the same (homogeneous) quantile value of $\tau$. All complexity results hold for both homogeneous and heterogenous quantiles. }
			\label{tab:contributions}
	\end{table*} }
	
	
	\paragraph{Utilitarian Welfare.}  
	We first show that  the problem of maximizing the Utilitarian Social Welfare is NP-hard for both balanced (\cref{thm:hardnessUSW:Balanced})  and unconstrained allocations (\cref{thm:NP-USW}). Over balanced allocations (where each agent receives $\lceil m/n\rceil$ or $\lfloor m/n\rfloor$ items), we prove that it is NP-hard to approximate the optimal \USW{}  within a factor $O(\frac{m}{n})$ for instances with $m\leq n^2$. We then present  a $\min(\lceil\frac{m}{n}\rceil+1,n)$-approximation algorithm, which matches the hardness of approximation bound up to a constant. 
	
	In the unconstrained setting,  we present a general reduction format that shows intractability for all $\tau \in [0,1)$. Complementing this, we present a $\left(\frac{n}{n-1}\right)$-approximation algorithm to the optimal $\USW{}$ (\cref{thm:scapegoat}). Our results thus demonstrate that the complexity of both problems differs significantly depending on whether the allocations are required to be balanced. 
	All our algorithms for utilitarian welfare are for instances with \textit{heterogeneous quantiles}, i.e., for settings where agents' quantile values need not be the same.

	\paragraph{Egalitarian Welfare.} 
	For egalitarian welfare, we begin with a result (\cref{lem:binRednESW}) showing that under quantile valuations, maximizing \ESW{} over any set of allocations can be reduced to maximizing \ESW{} over these allocations when agents have {\em binary valuations} (item values are either $0$ or $1$). That is, given a set $\Pi$ and an algorithm $\textsc{ALG}$ that finds a maximum \ESW{} allocation over $\Pi$ under binary values, by making at most $mn$ calls to $\textsc{ALG}$, we can find a maximum welfare allocation over $\Pi$ for an arbitrary instance. Here, $m$ and $n$ refer to the number of items and agents, respectively. Consequently, all our algorithms for \ESW{} are given for instances with binary valuations. 
	
	In the setting where allocations are constrained to be balanced, we prove that \ESW{} optimal allocations can be computed in polynomial time, even when agents have arbitrary \textit{heterogeneous quantiles}. This is in contrast to \USW{}, where we have not only intractability, but hardness of approximation. 
	When not restricted to balanced allocations, we show that the computational complexity of maximizing \ESW{} is highly dependent on the agents' quantile values. Specifically, we prove that when agents have homogeneous quantiles $\tau$, the problem is solvable in polynomial time for $\tau \in \{ 0, \sfrac{1}{3},1\} \cup\{  \sfrac{t}{t+1} \ | \  t \in \mathbb{Z}_+ \}$. In contrast,  for $\tau \in (0,\sfrac{1}{4}] \cup (\sfrac{3}{8}, \sfrac{2}{5}] \cup (\sfrac{5}{9}, \sfrac{3}{5}]$, the problem becomes APX-hard, with  no multiplicative approximation possible, unless P=NP. 
	
	\paragraph{Identical Valuations.} When all agents have the same valuation function, the strong intractability results for maximum \USW{} under balanced allocations and maximum \ESW{} over all allocations can be overcome. Specifically, finding maximum welfare balanced allocations can be done in polynomial time for either \USW{} or \ESW{}. Note that, the result for \ESW{} follows from the case of general valuations itself (\cref{thm:balESW}). For unconstrained allocations, maximum \ESW{} allocations can be found using a surprising method of ``balancing'' the number of low valued items across agents. These results are provided in \cref{sec:identical}. The problem of maximizing \USW{} in the unconstrained setting with identical valuations remains open. %

	%
	
	\paragraph{Near-optimal results.} Our results are in fact quite tight. For \USW{}, it is straightforward to see that the hardness bounds are match the guarantees provided by the approximation bounds, upto a constant. Specifically, in the balanced case, it is NP-hard to get an $\frac{m}{n(12+\epsilon)}$ approximation algorithm, for any $\epsilon>0$ and sufficiently large $m/n$. Our approximation algorithm matches this up to a constant by giving a $\min (\lceil\frac{m}{n}\rceil+1,n)$. For the unconstrained case, it is straightforward to see that an $\frac{n}{n-1}$ approximation is near-optimal for an NP-hard problem. Further, it is worth noting that for  both our intractability results for \USW{} we show a generalized reduction for {\em every} choice of $\tau \in [0,1)$. This is complemented by straightforward algorithms to maximize \USW{} when $\tau=1$. As a result, we present a comprehensive set of results for utilitarian welfare. 
	
	For egalitarian welfare, we show that even with binary values for the items, the problem of maximizing \ESW{} is NP-hard. As this holds for binary valuations, any $\alpha>0$ approximation on \ESW{} would imply exact \ESW{}. As a result, the problem of maximizing \ESW{} is APX-hard with no non-trivial approximation possible. We then identify a subset of quantiles where the problem of maximizing \ESW{} can be completed in polynomial time, while showing NP-hardness for a large sub-class of the remaining quantiles. The intractability holds when all agents have the same quantile, thus easily extends to  heterogeneous quantiles. These results contrast with those for  utilitarian welfare, where the cutoff between tractable and intractable instances is quite simple.

	\subsection{Related Work}
	The problem of allocating indivisible items fairly and/or efficiently is very well studied (see \citet{AAB+2022fair} for a survey). Existing literature almost exclusively assumes that aggregated preferences are monotone, very often, additive \citep{CKM+2019unreasonable,ACIW2022fair}. Our proposed valuations are non-monotone for most quantiles.  
	
	\paragraph{Allocating Indivisible Items.} The problem of allocating indivisible items fairly and/or efficiently is very well studied (See \citet{AAB+2022fair} for a survey). Existing literature almost exclusively assumes that preferences are aggregated in a monotone manner, often assuming additive valuations \citep{CKM+2019unreasonable,ACIW2022fair}, but also at times subadditive \citep{BKKN2024sublinear,BCIZ2020finding}  or superadditive valuations \citep{BVV2023fair,ViZi2023weighted}. A couple of papers also consider arbitrary valuations, with no underlying structure guaranteed, in addition to monotone valuations \citep{BBB+2024envy,BBPP2024nearly}. Our paper considers quantile preferences which may be monotone for the two extreme  quantiles, but non-monotone for all others.  
	
	\paragraph{Constrained Allocations.} While typical work on allocating indivisible items does not explicitly restrict the type of allocations studied there has been some work restricting the number of items that can be allocated. \citet{SHS2023efficient} and \citet{BiBa2018fair} consider a setting where items are partitioned into categories and there is a uniform constraint on how many items of each category can be allocated to a single agent. For this space, \citet{BiBa2018fair} focus on envy based fairness, while \citet{SHS2023efficient} consider efficiency via pareto optimality and envy-based fairness. 
	
	\citet{CaNa2024repeatedly} study a repeated matching setting where there are $T$ rounds and $n$ agents and $n$ items. In each round, each agent must receive exactly one item. Here, value for an item depends on how many times this agent has received the item in the past. For this space, \citet{CaNa2024repeatedly} pursue utilitarian social welfare and envy-based fairness.
	
	\paragraph{Quantile based preferences. } Quantile based preferences are well-established in mathematical economics and social choice theory. \cite{DeGa2019dynamic,DeGa2022static} show that quantile preferences are a more accurate model of real-life behavior of agents in random settings over expected utility. Recently, quantile valuations have been introduced to the setting of randomized social choice as well as one and two-sided matchings \cite{CaRo2025quantile}. 
	
	These preferences are a generalization of \textit{preference set extensions} that lift preferences over individual items to a set of items. The study of preference set extensions has a long-standing history in social choice theory~\citep{BBP04a} and has been applied to hedonic coalition formation games~\citep{CeHa2003computational,CeHa2004stable,Cech08a}, committee selection~\citep{AzMo20a} and multidimensional matchings \citep{HNR2025strategyproof}. Among them, one is called the \textit{best set extension} in which the sets are compared based on the best item in each set. One is called the \textit{worst set extension},   in which the sets are compared based on the best item in each set. The best and worst extension correspond to the $\tau=1$ and $\tau=0$ in our model.

	Quantile based set extensions have been explored in prior work through the lens of specific quantiles. The downward lexicographic (DL) and the upward lexicographic (UL) set extension are both natural refinements of the best and worst  set extensions, respectively. Both lexicographic extensions are also special cases of set extensions based on additive valuations~\citep{BBP04a}.  
	Lexicographic preferences  have been well studied within fair division \citep{AGMW15a,HMW2023fairly,HSVX2021fair,HSVX2023fairly,EPS2022fairly}.
	Other quantiles have also been considered previously. \citet{NiPa84a} characterize median quantile preferences that are a special case of $\tau=1/2$ in our model.

	\paragraph{Quantiles within Fair Division.} Recently, quantiles have been used in various ways to form and solve fair division problems. None of this work uses quantiles within the preferences, however. \citet{BHP2022dynamic} use quantiles to determine the allocation in dynamic settings. Specifically, in one of their algorithms, they allocate a new item to the agent for whom the ordinal rank of this new item is the highest compared to items already in their bundle.  Note that valuations here continue to be additive and the corresponding quantile values are only used to choose the agent who would value the item the most. 
	
	\citet{BFHN2023fair} consider settings where an agent assesses the fairness of a bundle by comparing it to her valuation in a random allocation. In this framework, a bundle is considered $q$-quantile fair, if it is at least as good as a bundle obtained in a uniformly random allocation with probability at least $q$. 
	In a similar vein, \citet{BFG+2024average} introduce the average value problem where the valuations are additive but they require that the average value of the bundles received by the agents meets a certain threshold.

	
	\section{Model}
	
	We shall use $[t]=\{1,\cdots, t\}$ for any $t\in \mathbb{Z}_+$.
	
	We consider a setting with  a set of agents $N$ s.t. $|N|=n$ and a set of items $M$, s.t. $|M|=m$. Each agent  $i\in N$ has a valuation function $v_i$ over $M$. We shall assume that for all $i\in N$ and $g\in M$ $v_i(g)\geq 0$, that is, all our items are {\em goods}. We shall use the terms ``items'' and ``goods'' interchangeably. Further, we refer to a subset of items as a {\em bundle}. 
	
	Informally, a valuation function is a $\tau$ \textit{quantile valuation}, for $\tau\in[0,1]$, if the value assigned to  a bundle $B\subseteq M$ is determined by the $\tau$ quantile of the distribution of item values in $B$.
	
	\begin{definition}[Quantile Valuations]
		Given a set of indivisible items $M$, we say that $v_i:2^M\rightarrow \mathbb{R}$ is a $\tau_i$ quantile for $\tau_i\in [0,1]$, if for any subset $B\subseteq M$, we have that \[v_i(S)= \min_{g\in S} \Bigg\{v_i(g): \frac{|\{g'\in B:v_i(g')\leq v_i(g)\}|}{|B|}\geq \tau_i\Bigg\}.  \]    
	\end{definition}
	
	An equivalent way of defining quantile valuations is to say that $v_i$ is a $\tau_i$ quantile for $\tau_i \in [0,1]$ if for any subset $B\subseteq M$ where $g_{i_{1}},\cdots,g_{i_{|B|}}$ are the items in $B$ s.t. $v_i(g_{i_1})\leq \cdots \leq v_i(g_{i_{|B|}})$ and $v_i(B)=v_i(g_{i_{\lceil \tau_i |B| \rceil} })$ if $\tau_i>0$, otherwise, $v_i(B)=v_i(g_{i_1})$. In particular, if $\tau_i$ is $0$, the agent values the given set as much as their least favorite item and if $\tau_i$ is $1$, they value it as much as their most favorite item.    We shall refer to the item $g_{i_{\lceil \tau_i B\rceil}}$ as the \textbf{quantile representative} of the bundle $B$ for agent $i$.
	
	We shall use $\tau_i$ to denote the quantile of agent $i$. When all agents have the same quantile, we shall simply use $\tau$. Whenever we assume agents need not have the same quantile, we shall mention that the instance has {\em heterogeneous quantiles}. Unless otherwise specified, we assume homogeneous quantiles, that is, all agents have the same quantile $\tau\in [0,1]$.  Consequently, an instance of our problem can be expressed by the tuple $I=\langle N, M, v, \tau \rangle$ where $v=(v_i)_{i\in N}$ and whenever we have heterogeneous quantiles, $\tau=(\tau_i)_{i\in N}$. 
	
	\begin{observation}
		Quantile valuations are non-monotone. Consider $g_1,g_2,g_3$ s.t. $v_i(g_1)=0$, $v_i(g_2)=1$ and $v_i(g_3)=2$ with $\tau_i=\sfrac{1}{2}$. Note that, $v_i(\{g_1,g_2,g_3)-v_i(\{g_1,g_3\})=1$ but $v_i(\{g_2,g_3\})-v_i(\{g_3\})=-1$.
	\end{observation}
	
	One benefit of this definition is that despite being non-monotone, the function does not require an oracle to specify the value on the subsets of items.  

	\paragraph{Allocations.} Each item must be allocated to some agent. Formally, an allocation $A=(A_1,\cdots, A_n)$ is an $n$-partition of $M$, with $A_i$ being the set of items assigned to agent $i\in N$. We shall use $\Pi(n,M)$ to denote the set of all allocations that divide the items in $M$ among $n$ agents. Our aim is to allocations with maximum welfare.
	
	\begin{definition}[Utilitarian Social Welfare (\USW{})]
		Given an instance $I=\langle N,M, v, \tau \rangle$ and an allocation $A=(A_1,\cdots, A_n)$, the utilitarian social welfare is the sum of the values received by the agents  $\USW{}(A)=\sum_{i\in N}v_i(A_i).$
	\end{definition}
	
	Given an instance $I=\langle N,M, v, \tau \rangle$, let $A^*$ be a maximum \USW{} allocation. We shall say that allocation $A$ is \USW{$\alpha$} for $\alpha \geq 1$, if $\USW{}(A)\geq  \frac{1}{\alpha}\USW{}(A^*)$. 
	
	\begin{definition}[Egalitarian Social Welfare (\ESW{})]
		Given an instance $I=\langle N,M,v, \tau \rangle$ and an allocation $A=(A_1,\cdots, A_n)$, the egalitarian social welfare is the minimum of the values incurred by the agents $\ESW{}(A)=\min_{i\in N}v_i(A_i).$    
	\end{definition}

	\paragraph{Balanced Allocations}
	
	Quantile valuations are very intuitive for settings where we insist on each agent getting an  equal number of items, whenever possible. This can be seen in the case of assigning papers to reviewers in conferences or assigning students to teachers. We shall consider both Utilitarian and Egalitarian Welfare with and without this requirement. 
	
	When considering balanced allocations, we shall consider only those allocations where each agent gets a bundle of size either $\lfloor m/n\rfloor$ or $\lceil m/n \rceil$. Note that this poses no restriction on the choice of agents, items or valuations. We shall use $\overline{\Pi}(n,M)$ to denote the set of all balanced allocations for instance $I$.
	It is important to note that when we consider maximizing \USW{} or \ESW{} over balanced allocations, we are in fact finding a maximum welfare allocation from $\overline{\Pi}(I)$ alone. That is, we are not holding the allocations to the standard of maximum welfare under unconstrained allocations. When not explicitly specified, we shall assume unconstrained allocations. 

	\begin{example}
		Consider an instance $I=\langle N,M,v,\tau\rangle$ where $n=4$ and $m=7$ s.t.  for each $i\in N$, item values are as follows:
		\begin{table*}
			\centering
			\begin{tabular}{cccccccc}
				&   $g_1$\quad &   $g_2$\quad &   $g_3$\quad &   $g_4$\quad &   $g_5$\quad &   $g_6$\quad &   $g_7$ \\
				\midrule
				$v_i$   &   1     &     1   &   1     &     1   &   0     &     0   &    0 
			\end{tabular}
		\end{table*}
		
		Observe that for any bundle $B\subseteq M$ and for any choice of $\tau_i$s, we have that $0\leq v_i(B)\leq 1$ for each $i\in N$. \\
		
		\noindent \underline{$\tau\leq \sfrac{1}{2}$.} First choose $\tau_i\leq \sfrac{1}{2}$ for each $i\in N$. Here if $\tau_i=0$, $v_i(B)=0$ if and only if $B\cap \{g_5,g_6,g_7\}\neq \emptyset$. Further, for each choice of $\tau\in (0,\sfrac{1}{2})$ for the bundle $B$ to have value $1$, $B$ must contain strictly more items of value $1$ than value $0$. 
		
		Consequently,  for any choice of $\tau_i\in [0,\sfrac{1}{2}]$, 
		the maximum \USW{} under any allocation is $3$ for example in the allocation $A$ where $A_1=\{g_1\}$, $A_2=\{g_2\}$, $A_3=\{g_3\}$ and $A_4=\{g_4,g_5,g_6,g_7\}$. In contrast, the maximum \USW{} under a balanced allocation is  $2$ as in the allocation $A'$ where $A'_1=\{g_1\}$, $A'_2=\{g_2,g_3\}$, $A'_3=\{g_4,g_5\}$ and $A'_6=\{g_6,g_7\}$. Observe that all allocations in these instance have the same \ESW{} of $0$.\\
		
		\noindent \underline{$\tau>\sfrac{1}{2}$.} In contrast, when $\tau_i>\sfrac{1}{2}$ for each $i\in N$, the maximum \USW{} under both balanced and unconstrained allocations is $4$ as in the allocation $A^*$ where $A^*_1=\{g_1,g_5\}$, $A^*_2=\{g_2,g_6\}$, $A^*_3=\{g_3,g_7\}$ and $A^*_4=\{g_4\}$. Further, in this case the maximum \ESW{} is now $1$ under both balanced and unconstrained allocations, as observed under $A^*$.
	\end{example}
	
	Note that when all agents have $\tau_i=1$, finding a maximum \USW{} allocation becomes straightforward. It suffices to find a maximum weight matching between agents and items, as the weight of this matching will be the maximum \USW{} possible. Allocating each agent their matched item and then allocating unmatched items arbitrarily ensures an allocation of \USW{} equal to the weight of the maximum matching.
	
	\begin{observation}
		Given an instance $I=\langle N,M,v,\tau\rangle$ where $\tau_i=1$ for each $i\in N$, a maximum \USW{} allocation can be found in polynomial time.  Further, the maximum \USW{} under balanced allocations will be the same as that under unconstrained allocations for $I$.
	\end{observation}
	
	Consequently, when considering \USW{}, our focus will be on instances where $\tau\in [0,1)$.  We now provide some other relevant definitions for our results.
	
	\begin{definition}[Identical Valuations]
		We say that a given instance has identical valuations if for each $i,j\in N$ and each $g\in M$, $v_i(g)=v_j(g)$ and $\tau_i=\tau_j$. Under instances with  identical valuations, we shall use $v$ to denote the valuation function of every agent. 
	\end{definition}
	
	\paragraph{Binary Valuations.} We will often focus on instances where each agent's valuation only takes values $0$ or $1$. That is, for each $i\in N$ and $g\in M$, $v_i(g)\in \{0,1\}$.  We shall refer to such valuations as \textbf{binary valuations}. In our model, the complexity of maximizing egalitarian welfare in over any given set of allocations $\Pi'\subseteq \Pi(n,M)$ is shown to be equivalent to maximizing \ESW{} over $\Pi'$ when all $v_i(g)\in \{0,1\}$, that is with binary valuations. When valuations are binary, maximum \ESW{} is $1$, if and only if there is an allocation where all agents receive a value of $1$. Consequently, under binary valuations, any allocation that gives $\alpha>0$ approximation to the maximum \ESW{} would simply be a maximum \ESW{} allocation. As a result, we do not pursue any multiplicative approximations to \ESW{}. 
	
	\subsubsection*{Relation between maximizing \USW{} and \ESW{}.}
	Observe that under binary valuations, there exists an allocation with an \ESW{} of $1$ if and only if there exists an allocation with a \USW{} of $n$. Consequently, whenever the problem of maximizing \ESW{} is shown to be hard for binary valuations, the same intractability would extend to maximizing \USW{}. However, it may be possible, under a specific setting with binary valuations, to maximize \ESW{} in polynomial time whereas maximizing \USW{} for that same setting can be intractable. What this implies is that while it may  be computationally tractable to check if all agents can simultaneously get a value of $1$, but finding a maximum sized subset of agents who can obtain a value of $1$ may be computationally hard. 
	Conversely, whenever  \USW{} can be maximized in polynomial time, a maximum \ESW{} allocation can also be found in polynomial time. Unfortunately, in the case of general quantile valuations, we either find it NP-hard or APX-hard to maximize \USW{}. While can indeed maximize \USW{} in polynomial time for balanced allocations under identical valuations.  
	
	

	
	\section{Balanced Allocations} \label{sec:bal}
	
	We first explore quantile valuations with the requirement that the allocations be balanced. Our results for \USW{} and \ESW{} lie in stark contrast with each other here. 
	

	\subsection{Utilitarian Social Welfare}\label{subsec:balUSW}
	We first show, when allocations are balanced,  that maximizing \USW{}  is NP-hard to approximate to better than a factor of $O(\lceil\frac{m}{n}\rceil)$. We then proceed to give a polynomial-time algorithm that matches hardness of approximation bound.
	

	\subsubsection{Hardness of Approximation.}
	We now give an approximation preserving reduction from the $k$-\textup{\textsc{DimensionalMatching}}(kDM) problem. In the kDM problem, we are given a k-partite hypergraph  $G=(X,H)$ where $X$ is the vertex set and the set of hyperedges $H\subseteq X^k$, that is, each hyperedge has size $k$ graph is the hyperedge set of $G$. The kDM problem requires finding a maximum collection of disjoint edges in $G$. Specifically, the decision version of the problem involves a hypergraph $G$ and a target $\ell$ and the given instance is a yes instance if and only if there is collection of vertex disjoint hyperedges (kDM) of size $\ell$ in $G$. \citet{LST2025asymptotically} showed that for any $\epsilon>0$, this problem is hard to approximate to a factor better than $\frac{k}{12+\epsilon}$ for a large $k$ unless \textsc{NP}$\subseteq$\textsc{BPP}.
	
	\begin{theorem}\label{thm:hardnessUSW:Balanced} 
		Given  instance $I=\langle N,M,v,\tau\rangle$ where $m\leq n^2$ and $\tau\in [0,1)$ , unless \textsc{NP}$\subseteq$\textsc{BPP}, for any constant $\epsilon>0$ and a sufficiently large value of $\frac{m}{n}$,  no polynomial time algorithm guarantees an $\USW{\frac{m}{n(12+\epsilon)}}$ balanced allocation.
	\end{theorem}
	\begin{proof}
		Given an instance of kDM, $\langle G=(X,H), \ell \rangle$,  we can assume, without loss of generality, that each vertex is contained in at least one hyper-edge. Thus, we have that $|X|\leq k |H|$. We create an instance of our problem with $n$ agents and $m=kn$ items 
		as follows: 
		\begin{itemize}
			\item For each edge $H_i\in H$, we create agent $i$. 
			\item For each vertex $x\in X$, we create item $g_x$.   
			\item To balance the item count, we introduce $k|H| - |X|$ dummy items $g'_1,\cdots, g'_{kn - |X|}$. 
		\end{itemize}
		
		Thus, we have $n=|H|$ agents and the number of items is $m=k|H|$. As a result, we have that $m = kn$. Recall that balanced allocations for this instance require $k=\frac{m}{n}$ items to be allocated to each agent. 
		
		For each agent $i\in N$, we set $\tau_i=0$ for all $i\in N$. Now for $i$ and each $g_x$, if $x\in H_i$, we set $v_i(g_x)=1$ else, we set $v_i(g_x)=0$. Finally, for each $t\in [k|H| - |X|]$, set $v_i(g'_t)=0$. 
		We now show that a matching of size $\ell$ in the kDM problem can be transformed into a balanced allocation whose \USW{} is at least $\ell$ in the reduced instance of our problem, and vice versa. Consider a matching $\mu$ of size $\ell$ in kDM. For each $H_i\in \mu$, allocate the items vertices in $H_i$. That is, $A_i=\{g_x|x\in H_i\}$. Arbitrarily allocate the remaining items, ensuring $|A_i|=k$. As each edge agent within the matching gets a value of $1$, we have that $\USW{}(A)\geq \ell$. 
		
		Now consider a balanced allocation $A$ in the reduced instance with a \USW{} of $\ell$. As the maximum value for any agent is $1$, this implies that $\ell$ agents receive a value of $1$ from $A$. By construction, $v_i(A_i)=1$ only if $A_i$ contains {\em all} the items corresponding to the vertices in $H_i$. Further, in the constructed instance, every balanced allocation must allocate each agent a bundle of size $k=\frac{m}{n}$. 
		Now, as $A$ is an allocation, we have that $\mu=\{H_i|v_i(A_i)=1\}$ must be a vertex disjoint collection of edges (a matching). Recall that, $\USW{}(A)\geq \ell$. Consequently, $|\mu|\geq \ell$.
		
		As a result, we have an approximation preserving reduction from the kDM problem. \\
		
		\noindent \textbf{Extending to other quantiles.} Observe that this reduction would also work the exact same way for all quantiles $\tau\in [0,\frac{1}{k}]$, as any bundle of size $k$ would be represented by its least valued item for these quantiles. Further, these ideas        
		can be extended to all $\tau\in (\frac{1}{k},1)$ by adding enough dummy items s.t. an agent gets a value of $1$ only if they get $k$ items of value $1$. Specifically, given any choice of $\tau\in (\frac{1}{k},1)$, choose $t$ s.t. $\lceil (t+k)\tau\rceil =t+1$. It remains to prove that such a choice of $t$ exists.
		
		Observe that $\lceil (t+k)\tau\rceil=t+1$ require $t<(t+k)\tau\leq t+1$. Consequently, given a fixed $\tau$ and $k$, the choice of $t$ must satisfy $t\in [\frac{\tau k-1}{1-\tau},\frac{\tau k}{1-\tau})$. Note that to show that an appropriate choice of $t$ exists, it suffices to show that this interval always contains at least one positive valued integer. First recall that $\tau>\frac{1}{k}$ and as a result, $\tau k>1$. Further as $1-\tau<1$, we have that $\frac{\tau k}{1-\tau}>1$. Further the length of the interval $[\frac{\tau k-1}{1-\tau},\frac{\tau k}{1-\tau})$ will be exactly $\frac{1}{\tau}>1$. Consequently, the interval $[\frac{\tau k-1}{1-\tau},\frac{\tau k}{1-\tau})$ must contain at least one choice of $t$ s.t.  $\lceil (t+k)\tau\rceil =t+1$.
		
		For such a choice of $t$, we can now do an analogous reduction with the same set of agents and a vertex item $g_x$ for each $x\in X$. Further, we add $k|H| - |X|+nt$ dummy items $g'_1,\cdots, g'_{n(t+k) - |X|}$ where agent valuations are set to $0$ for all dummy items as before. Note that in this instance we have $n=|H|$ agents and $m=n(t+k)$ and each agent must get a bundle of size $t+k$ under a balanced allocation. Using analogous arguments to the reduction for $\tau=0$, we obtain an approximation preserving reduction from the kDM problem for {\em every} choice of $\tau \in [0,1)$. 
		
		\citet{LST2025asymptotically} proved that for any $\epsilon>0$, there exists a class of instances with $k < |H|$ being sufficiently large, such that kDM is hard to approximate to a factor better than $\frac{k}{12+\epsilon}$ for a large $k$ unless \textsc{NP}$\subseteq$\textsc{BPP}. Thus, we have hardness of approximation for maximizing \USW{} over balanced allocations for instances where $m< n^2$. 
	\end{proof}
	

	\subsubsection{Near-Optimal Algorithm.}
	We now provide an approximation algorithm that  matches the lower bound placed by \cref{thm:hardnessUSW:Balanced} up to a constant. The greedy algorithm (\cref{alg:goods-USWbalanced}) proceeds by iteratively allowing unassigned agents to ``demand" their best possible set from the unassigned items. Specifically, this set contains the items that will be the quantile representative or higher valued. Note that the number of such items depends on each agent's quantile value and whether the agent is meant to get $\lceil m/n\rceil$ items or $\lfloor m/n \rfloor$ items. We then choose the agent whose value for their demanded set is highest. We repeat this  until all items are assigned. Throughout the algorithm we  ensure that exactly $m \mod n$ agents get bundles of size larger than $\lfloor \frac{m}{n}\rfloor$. 
	
	\begin{algorithm}[t]
		\KwIn{Instance with heterogeneous quantiles $\langle N,M, v, \tau \rangle$}
		\KwOut{A balanced allocation $A$}
		\DontPrintSemicolon
		Initialize set of unallocated goods $P\gets M$\;
		Initialize set of unassigned agents $N'\gets N$\;
		Let $t\gets (m \mod n) - n$  \;
		\quad \quad \quad \quad \quad  \quad \quad \quad \quad \Comment{Negative of the number of bundles that need $\lfloor \frac{m}{n}\rfloor $ items}
		
		\While{$N'\neq \emptyset$}{
			\eIf{$t>0$}{
				$t_+\gets 1$\;
			}
			{
				$t_+\gets 0$
			}
			$k\gets \lfloor m/n \rfloor + t_+$ \quad \quad \quad \quad \quad \quad \quad \quad \quad \quad \quad \quad \quad \quad \quad \quad \quad \quad  \Comment{Current bundle size}
			\For{each  $i\in N'$}{
				Let $k'\gets \min(k,k-\lceil \tau_i k\rceil+1) $\;
				Let $S_i\subseteq S$ be s.t. $|S_i|=k'$ and for all $g\in S_i$ and $g'\in S\setminus S_i$, $v_i(g)\geq v_i(g')$\;
			}
			Let $i^*\gets \argmax\limits_{i\in N'} \left( \min_{g\in S_i} v_i(g)\right)$\;
			Set $A_{i^*}\gets S_{i^*} $ and $k_{i^*}\gets k$\;        
			$P\gets P\setminus A_{i^*}$\;
			$N'\gets N' \setminus \{i^*\}$ \;
			Update $t\gets t+1$ \;
		}
		Allocate items in $P$ arbitrarily  s.t. $|A_i|=k_i$ for all $i\in N$\;
		\textbf{Return} $A$
		\caption{$\USW{\min(\lceil\frac{m}{n}\rceil+1,n)}$ Greedy Algorithm}\label{alg:goods-USWbalanced}
	\end{algorithm}
	
	\begin{restatable}{theorem}{USWgreedy}\label{thm:balUSWgreedy}
		Given an instance $I=\langle N,M,v,\tau \rangle$ and heterogeneous quantiles, \cref{alg:goods-USWbalanced} returns a balanced allocation which is $\USW{\min(\lceil\frac{m}{n}\rceil+1,n)}$ in polynomial time.
	\end{restatable}
	
	\begin{proof} In order to prove correctness, we first show balancedness then the approximation guarantee. Note that when the number of items $m$ is a multiple of $n$, all agents must get bundles of equal size. Some agents only need to get larger bundles when $m \mod n>0$. Specifically, under a balanced allocation, an agent must receive $\lfloor m/n\rfloor +1$ items if and only if $m \mod n>0$.\\
		
		\noindent \textbf{Balancedness.} First we establish that each agent gets either $\lceil m/n \rceil$ or $\lfloor m/n \rfloor$ items. Note that, in each iteration, agents demand a set of items based on the variable $k$ which only ever takes the value of $\lfloor m/n \rfloor$ or $\lfloor m/n \rfloor +1$. When an agent $i$ is first allocated their demanded set $S_i$, the current value of $k$ is stored in $k_i$. This agent is removed from $N'$ and the value of $k_i$ once fixed is never changed. Any items allocated to $i$ later are always such that $|A_i|=k_i$. 
		
		Further, note that an agent gets $\lfloor m/n\rfloor +1$ items only if the variable $t>0$. Recall that $t$ is initialized to $(m \mod n)-n$ and increased by $1$ in every iteration of the while loop. As a result, for only the first $n-(m \mod n)$ iterations of the while loop do agents get a value of $k_i=\lfloor m/n\rfloor $. Consequently, exactly $m \mod n$ agents receive bundles of size $\lfloor m/n\rfloor +1$, while the remaining agents receive bundles of size $\lfloor m/n\rfloor$.\\
		
		\noindent \textbf{Approximation Guarantee.} 	Recall the notation of $N',$ $P$ and $S_i$ as defined in \cref{alg:goods-USWbalanced}. Given $I$, let $A^*=(A^*_1,...,A^*_n)$ be a maximum \USW{} balanced allocation.   Without loss of generality, we assume that  $v_1(A^*_1) \geq v_2(A^*_2) \geq \dots \geq v_n(A^*_n)$. Let $i_t$ denote the agent who is allocated a bundle in the $t$-th iteration of the while loop, and let $A_{i_t}$ denote the corresponding bundle allocated to her under \cref{alg:goods-USWbalanced}. 
		
		Let $k_i'=\min(\lceil m/n\rceil,\lceil m/n \rceil-\lceil(\tau_i \lceil m/n \rceil )\rceil +1)$. That is, $k_i'$ is minimum number of items in any $\lceil m/n \rceil$ sized bundle $B$ s.t. $v_i(g)\geq v_i(B)$. Recall that under the greedy algorithm, agent $i$ ``demands" set $S_i$ containing $k_i'$ items. Specifically, these are the highest valued $k_i'$ items in $P$ for $i.$  Let $k^*$ be the maximum value of $k_i'$ encountered by \cref{alg:goods-USWbalanced} on its execution on the given instance. That is, $k^*$ is the largest set demanded by any agent. Observe that $1\leq k^*\leq \lceil m/n\rceil $.   We shall now show that the first $\lceil \frac{n}{k^*+1}\rceil$ agents to receive a bundle will have value comparable to the value under specific bundles under $A^*$. \\
		
		\noindent\textbf{Claim:} For each $t=1,\cdots, \lceil \frac{n}{k^*+1} \rceil$, we have that the value of agent $i_t$, $$v_{i_t}(A_{i_t}) \geq v_{(t-1)(k'+1)+1}(A^*_{(t-1)(k^*+1)+1}).$$
		
		\noindent{\em Proof of Claim.} 
		We shall prove this by induction. First consider agent $i_1$. Observe that at this point $N'=N$ and $P=M$. As a result, agent $1$ can demand their top $k'\gets \min(k,k-\lceil \tau_i k\rceil+1)$ items in $S_1$ and get a value of at least $v_1(A_1^*)$. 
		Consequently, the best possible bundle $A_{i_1}$ must be such that $v_{i_1}(A_{i_1})\geq v_1(A_1^*)$. 
		
		Suppose we have that for all $t\leq \bar{t}-1$, the claim holds. 
		Let $L=\bigcup_{\ell\in [\bar{t}-1]} A_{i_\ell}$ be the set of items that are allocated up to the $(\bar{t}-1)$th iteration of the while loop. In each iteration at most $k^*$ items are allocated. Consequently, we have that $|L|\leq k^*(\bar{t}-1)$. It follows that in the worst case, the number of bundles under $A^*$ for which some item has already be allocated in $L$ is   $| \{ j\in [(\bar{t}-1)(k^*+1)+1]  \ : \   A^*_j \cap L \neq  \emptyset \} | \leq |L| = k^*(\bar{t}-1) $.
		
		Consequently, we get that among the top $(\bar{t}-1)(k^*+1)+1$ bundles under $A^*$, at least $\bar{t}$ bundles do not intersect with $L$.  However, in order to be able to allocate one of these bundles, or one of equivalent value, the corresponding agent must be within the set $N'$. Thus far, $\bar{t}-1$ bundles have been allocated. Consequently, at least one bundle and agent pair among these $\bar{t}$ unallocated bundles must remain available for selection .
		Hence, we must have that $v_{i_{\bar{t}}}(A_{i_{\bar{t}}}) \geq v_{(t-1)(k^*+1)+1}(A^*_{(t-1)(k^*+1)+1})$.
		\qed
		\vspace{2mm}
		
		\noindent  We can now prove the approximation guarantee. Let $\alpha=\min(k^*+1,n)$. Observe that $\lceil \frac{n}{k^*+1}\rceil=\lceil\frac{n}{\alpha}\rceil$. The \USW{} of $A$ is lower bounded by $\sum_{t=1}^{\lceil \frac{n}{\alpha} \rceil}  v_{i_t}(A_{i_t})$. From the proof of the claim, we know that 
		\[\sum_{t=1}^{\lceil \frac{n}{\alpha} \rceil}  v_{i_t}(A_{i_t})\geq \sum_{t=1}^{\lceil \frac{n}{\alpha} \rceil}  v_{(t-1)(k^*+1)+1}(A_{(t-1)(k^*+1)+1}^*) \]
		Recall that agents are ordered according to $A^*$, that is, $v_1(A^*_1)\geq \cdots\geq v_n(A^*)$. As a result, we get that 
		\begin{align*}
			\USW{}(A^*) &=\sum_{t\in [n]} v_t(A^*_t)\\
			&\geq \alpha \sum_{t=1}^{\lceil\frac{n}{\alpha} \rceil} v_{(t-1)(k^*+1)+1}(A_{(t-1)(k^*+1)+1}^*).
		\end{align*}
		
		From the claim, we know that this is greater than or equal to $\alpha$ times the value obtained by the first $\lceil\frac{n}{k^*+1}\rceil$ agents under \cref{alg:goods-USWbalanced}. Hence, $$\USW{}(A)\geq \sum_{t=1}^{\lceil \frac{n}{\alpha} \rceil}  v_{i_t}(A_{i_t})\geq \sum_{t=1}^{\lceil \frac{n}{\alpha} \rceil}  v_{(t-1)(k^*+1)+1}(A_{(t-1)(k^*+1)+1}^*) \geq \frac{\USW{}(A^*)}{\alpha}=\frac{USW(A^*)}{\min (k^*+1,n)}.$$
		
		Observe that when each agent demands fewer than $n-1$ items, we are guaranteed $\USW{(k^*+1)}$ which may be even better than $\USW{(\lceil m/n\rceil+1)}$. However, when the number of items demanded is at least $n-1$, the greedy algorithm can only guarantee $\USW{n}$. Consequently, for an arbitrary $I$, \cref{alg:goods-USWbalanced} is $\USW{\min (\lceil\frac{m}{n}\rceil +1,n)}$.\\
		
		\textbf{Running Time.} Note that the while-loop proceeds as long as the set $N'$ is non-empty and each iteration of the while-loop removes one agent from $N'$. Consequently, the while-loop runs for exactly $n$ iterations. Within one iteration of the while-loop, the for-loop has at most $n$ iterations. Each iteration of this for-loop computes the set an agent in $N'$ will demand, which takes time $O(m)$. Thus the complete for-loop execution takes time $O(mn)$. Subsequently, finding the agent with the highest value for their demanded set can be done in time $O(n)$. All remaining steps execute in time $O(1)$. Consequently, the algorithm runs in time $O(mn^2)$.        
	\end{proof}

	
	\subsection{Egalitarian Social Welfare}
	We now move to maximizing egalitarian welfare. We begin with a very useful reduction, which facilitates all our algorithms for \ESW{}. We  show that whenever there is an algorithm to find an allocation with maximum \ESW{} under binary valuations, we can use it to find a maximum \ESW{} allocation under general non-negative valuations. 
	
	\begin{restatable}{lemma}{binreduction}\label{lem:binRednESW}
		Given a set of allocation $\Pi'\subseteq \Pi(n,M)$ and an algorithm $\textsc{ALG}$ which computes a maximum \ESW{} allocation over $\Pi'$ for instances with binary valuations.           
		The problem of maximizing \ESW{} over allocations in $\Pi'$ under heterogeneous quantiles and arbitrary non-negative item values can be solved by making at most $O(mn)$ calls to \textsc{ALG} . 
	\end{restatable}
	
	\begin{proof}
		Consider an arbitrary instance with (possibly) heterogeneous quantiles $I=\langle N,M,v,\tau \rangle$ and $\Pi'\subseteq \Pi(n,M)$. Each agent can receive at most $m$ distinct values, consequently, the \ESW{} of an allocation can take at most $mn$ different values. As a result, we can check for at most $mn$ distinct threshold values for $\nu$ s.t. we wish to find an allocation in $A\in \Pi'$ where $\ESW{}(A)\geq \nu$. 
		
		Fix a value for the threshold $\nu\in \{v_i(g)| i\in N, g\in M\}$. We can construct an alternate instance $I'=\langle N,M,v',\tau \rangle$ with binary valuations as follows: $v'_i(g)=1$ if and only if $v_i(g)\geq \nu$. We can now show that an allocation $A\in \Pi'$ has $\ESW{}(A)\geq \nu$ under $v$ if and only if $\ESW{}(A)=1$ under $v'$.
		
		Suppose we have an allocation $A\in \Pi'$ s.t. $A$ has $\ESW{}(A)\geq \nu$ under $v$. Thus, for each $i\in N$, $A_i$ must contain enough goods each with value at least $\nu$ so that $v_i(A_i)\geq \nu$. Thus, for the same quantile $\tau_i$, it must be that $v_i'(A_i)\geq 1$. Consequently, $\ESW{}(A)\geq 1$ under $v'$.
		
		Now, suppose we have an allocation $A\in \Pi'$ s.t. $A$ has $\ESW{}(A)\geq 1$ under $v'$. We can analogously see that $v_i'(A_i)\geq 1$ if and only if $v_i(A_i)\geq \nu$. As a result, it must be that $v_i(A_i)\geq \nu$ for each $i\in N$, and thus, $\ESW{}(A)\geq \nu$ under $v$.
		
		Hence, given an algorithm that finds a maximum \ESW{} allocation for $\Pi'$ under an instance with binary goods, we can make at most $mn$ calls to it to find a maximum \ESW{} allocation for $\Pi'$ under arbitrary goods.
	\end{proof}
	
	\cref{lem:binRednESW} enables us to maximize \ESW{} over balanced allocations, even if the quantile values are heterogeneous. 
	Henceforth, we need only consider a setting where $v_i(g) \in \{0,1\}$ for all $i\in N$ and all $g\in M$. Here, we shall try to see if an allocation with \ESW{} $1$ can exist. That is, all agents must get a value of $1$. In order to achieve this, we first make the following observation:
	
	\begin{observation}\label{obs:basicBinVal}
		Given a bundle $B\subseteq M$ and $i\in N$ s.t. $v_i(g)\in \{0,1\}$ for all $g\in M$, we have that $v_i(B)=1 \Leftrightarrow |\{g\in B|v_i(g)=0\}|\leq \lceil \tau_i |B|\rceil -1$.
	\end{observation}
	
	\begin{algorithm}[!ht]
		\KwIn{Instance with binary values and heterogeneous quantiles $\langle N,M, v, \tau \rangle$ }
		\KwOut{Balanced Allocation $A$}
		\DontPrintSemicolon
		\Comment{Set up graph for max flow }
		Let $N_1\gets \{i\in N| \lceil \tau_i \lfloor \frac{m}{n}\rfloor \rceil =\lceil (\tau_i\lceil \frac{m}{n}\rceil)\rceil \}$\;
		Let $N_0\gets N\setminus N_1$\;
		For each $i\in N$, set $m_i\gets \min (\lceil \frac{m}{n}\rceil, \lceil \frac{m}{n}\rceil -\lceil \tau_i\lceil\frac{m}{n}\rceil\rceil +1)$\;
		Create a demand flow graph $G=(V,E)$ with associated capacity function $c:E\rightarrow \mathbb{Z}_+$ and demand $d:E\rightarrow \mathbb{Z}_+$ where\\
		$V\gets N\cup M \cup \{s,t\}$ and if $|N_0|<m \mod n$ add vertex $t'$ to $V$\\
		For each $g \in M$ add edge $(s,g)$ to $E$ with capacity $c_{(s,g)}=1$\\
		For each $i\in N$ and each $g\in M$ s.t. $v_i(g)=1$, add edge $(g,i)$ to E with $c_{(g,i)}=1$\\
		For each $i\in N_0$ add  edge $(i,t)$ to $E$ with capacity $c_{(i,t)}=m_i$\;
		\eIf{$m \mod n >|N_0|$}{
			For each $i\in N_1$ add edge $(i,t')$ to $E$ with capacity $c_{(i,t')}\gets m_i$ and  demand $d_{(i,t')}\gets m_i-1$\\
			add edge $(t',t)$ to $E$ with capacity $c_{(t',t)}\gets \sum_{i\in N_1}(m_i-1) + (m \mod n) -|N_0|$\;
			
		}{
			For each $i\in N_1$, add edge $(i,t)$ to $E$ with capacity $c_{(i,t)}=m_i-1$\;
		}
		For each edge $e\in E$ whose demand has not defined, set $d_e\gets 0$.
		Let $f: E\rightarrow \mathbb{Z}_+$ be an integral maximum flow from $s$ to $t$ under $G$ satisfying $d_e\leq f_e\leq c_e$ for each $e\in E$, if it exists else set $f_e=0$ for all $e\in E$\;
		\Comment{Build allocation}
		
		\eIf{for some $x\in N\cup \{t'\}$ there exists $(x,t)\in E$ s.t. $f_{(x,t)}<c_{(x,t)}$}{
			Let $A$ be an arbitrary balanced allocation \label{step:no-instance-return}\;
		}{
			Initialize allocation $A=(A_1,\cdots A_n)$ where $A_i\gets \{g|f_{(g,i)}=1\}$ for each $i\in N$\; 
			\Comment{Set bundle sizes to allocate remaining items}
			\eIf{$m\mod n >|N_0|$}{
				Set $k_i\gets \lceil  \frac{m}{n}\rceil$ for all $i\in N_0$ and each $i\in N_1$ s.t. $f_{(i,t')}=c_{(i,t')}$\;
				Set $k_i\gets \lfloor  \frac{m}{n}\rfloor$ for each $i\in N_1$ s.t. $f_{(i,t')}<c_{(i,t')}$\;
			}{
				Initialize $S\gets \emptyset$\;
				\If{$m\mod n>0$}{
					Pick an arbitrary $S\subseteq N_0$ s.t.$|S|=m\mod n$\;
					
				}
				Set $k_i\gets\lceil \frac{m}{n}\rceil$ for all $i\in S$ 
				and $k_i\gets \lfloor \frac{m}{n}\rfloor$ for all $i\in N\setminus S$ \label{step:pick-within-Nzero}\;
			}
			Allocate items in $M\setminus (\bigcup_{i\in N} A_i)$ arbitrarily but ensuring $|A_i|=k_i$ for all $i\in N$ \label{step:yes-instance-return}\;
			
		}
		\textbf{Return} $A$\;
		\caption{Max \ESW{} over balanced allocations for binary goods.}\label{alg:balESW}
	\end{algorithm}

	This follows from the definition of quantile valuations. Thus, allocation $A$ satisfies $\ESW{}(A)=1$ if and only for each  $i\in N$, $A_i$ contains at least $k_i=\min (|A_i|,|A_i|-\lceil \tau_i |A_i|\rceil +1)$ items of value $1$. Note that the min argument only comes in when $\tau_i=0$. For the case of balanced allocations, we need $|A_i|\in \{\lfloor \frac{m}{n}\rfloor,\lceil \frac{m}{n}\rceil\}$. We use this observation to build our algorithm.
	
	\paragraph{Algorithm Overview.} Given an instance with binary valuations, our algorithm proceeds by first checking if it is possible to allocate each agent enough items of value 1 to build an allocation with an \ESW{} of $1$. If this is possible, we first allocate these items and then allocate the remaining items arbitrarily but respecting the balancedness requirement. In order to check if an \ESW{} of $1$ is possible, we first calculate the number of items of value $1$ an agent will need. Note that this number may be different for a bundle of size $\lfloor \frac{m}{n}\rfloor$ and one of size $\lceil \frac{m}{n}\rceil$. To this end, we separate agents into two sets: $N_1$ being the set of agents for whom $\lceil\tau_i \lfloor\frac{m}{n}\rfloor\rceil$ is the same as $\lceil (\tau_i\lceil \frac{m}{n}\rceil)\rceil $ and $N_0$ being the set for whom these value are different. 
	
	Observe that $N_0$ is non-empty only when $\lceil \frac{m}{n}\rceil\neq\lfloor \frac{m}{n}\rfloor$. As a result, $N_0$ is only non-empty when all agents do not have the same bundle size under a balanced allocation. Specifically, the agents in $N_0$ are those who need the same number of items of value $1$ under either bundle size. Without loss of generality, when deciding which agents are to receive a larger bundle size, we can first choose from the set $N_0$. We choose an agent in $N_1$ to have a larger bundle size only if there are not enough agents in $N_0$, i.e., $m\mod n >|N_0|$. Note that under a balanced allocation $A$  where $\ESW{}(A)=1$, for each $i\in N_1$ s.t. $|A_i|=\lfloor\frac{m}{n}\rfloor$, $A_i$ must contain at most $\lceil\tau_i \lfloor\frac{m}{n}\rfloor\rceil-1$ items of value $0$ for $i$. Analogously, when $|A_i|=\lceil \frac{m}{n}\rceil$ for $i\in N_1$, $A_i$ must contain at most $\lceil\tau_i \lfloor\frac{m}{n}\rfloor\rceil$ items of value $0$ for $i$. Further exactly $\max (0,(m \mod n) - |N_0|)$ agents in $N_1$ must receive a bundle of size $\lceil \frac{m}{n}\rceil$. 
	
	We set up a flow network to check if this is possible in \cref{alg:balESW}. The flow network is defined based on whether $|N_0|<m\mod n$ or not, but essentially it sends a flow from the source $s$ to items, and then from each item to agents who have value $1$ for the item, with capacities of each of these edges being set to $1$. Agents in turn send this flow ahead with exactly one outgoing edge, where the capacity of the edge being the number of items of value $1$ they will need under a balanced allocation with \ESW{} $1$. All agents in $N_0$ have this edge to the sink $t$. Whenever it holds that $|N_0|\geq m \mod n$, all agents in $N_1$ too have an edge to the sink $t$ with capacity being set according to them receiving a bundle of size $\lfloor m/n\rfloor$. 
	
	In the remaining case when $|N_0|< m\mod n$, we add an additional vertex $t'$ and all agents in $N_1$ have their outgoing edge to $t'$ and $t'$ has an outgoing edge to the sink $t$. The capacity of the edges from agents in $N_1$ to $t'$ are set according to these agents receiving a bundle of size $\lfloor m/n\rfloor +1$. Additionally, we place a minimum demand on the edges from agents in $N_1$ to $t'$ of receiving at least as much flow as the number of value $1$ items they would need to get value $1$ from a bundle of size $\lfloor m/n\rfloor$. The capacity of the edge from $t'$ to $t$ is set such that exactly $(m\mod n)-|N_0|$ agents receive a bundle of size $\lfloor m/n\rfloor+1$ while the remaining receive a bundle of size $\lfloor m/n\rfloor$. 
	
	Before proving the correctness of \cref{alg:balESW}, we show two examples of its execution in \cref{ex:one,ex:two}.

	\begin{example}\label{ex:one}
		Consider an instance with binary goods where $n=3$ and $m=7$. For agents $1$ and $2$, we have $v_1(g_1)=v_1(g_2)=v_2(g_1)=v_2(g_2)=1$ and $\tau_1=\tau_2=\sfrac{1}{2}$. In contrast, for agent $3$, we have $v_3(g_3)=v_3(g_4)=1$ and $\tau_3=\sfrac{2}{3}$. All remaining item values are $0$. Observe that on this instance $N_0=\{3\}$ and $N_1=\{1,2\}$ and $m \mod n=1=|N_0|$. Consequently, the flow graph that will be constructed for this instance is as depicted in \cref{fig:example-one}. We do not discuss the demands from the edges as on each edge, the demand is $0$.
		
		In this instance, we can construct the following flow $f$ where $f_{(s,g)}=1$ for each $g\in \{g_1,g_2,g_3,g_4\}$. Further, for each $i\in [3]$, we set $f_{(g_i,i)}=1$ and $f_{(g_4,3)}=1$. This in turn makes $f_{(1,t)}=f_{(2,t)}=1$ and $f_{(3,t)}=2$. Note that this flow is a max flow from $s$ to $t$ as each edge to $t$ is saturated. One resultant allocation from $f$ would be $A=(A_1,A_2,A_3)$ where $A_1=\{g_1,g_5\}$, $A_2=\{g_2,g_6\}$ and $A_3=\{g_3,g_4,g_7\}$ and $v_i(A_i)=1$ for each $i\in N$. 
	\end{example}

	\begin{figure}
		\centering
		
		\tikzset{every picture/.style={line width=0.9pt}}
		
		\begin{tikzpicture}[x=0.7pt,y=0.7pt,yscale=-1,xscale=1]
			
			\draw    (25.7,88.92) -- (73.96,20.95) ;
			\draw [shift={(75.7,18.5)}, rotate = 125.37] [fill={rgb, 255:red, 0; green, 0; blue, 0 }  ][line width=0.08]  [draw opacity=0] (5.36,-2.57) -- (0,0) -- (5.36,2.57) -- cycle    ;
			\draw   (5.9,88.92) .. controls (5.9,83.45) and (10.33,79.02) .. (15.8,79.02) .. controls (21.27,79.02) and (25.7,83.45) .. (25.7,88.92) .. controls (25.7,94.39) and (21.27,98.82) .. (15.8,98.82) .. controls (10.33,98.82) and (5.9,94.39) .. (5.9,88.92) -- cycle ;
			\draw   (75.7,18.5) .. controls (75.7,13.03) and (80.13,8.6) .. (85.6,8.6) .. controls (91.07,8.6) and (95.5,13.03) .. (95.5,18.5) .. controls (95.5,23.97) and (91.07,28.4) .. (85.6,28.4) .. controls (80.13,28.4) and (75.7,23.97) .. (75.7,18.5) -- cycle ;
			\draw   (75.7,63.5) .. controls (75.7,58.03) and (80.13,53.6) .. (85.6,53.6) .. controls (91.07,53.6) and (95.5,58.03) .. (95.5,63.5) .. controls (95.5,68.97) and (91.07,73.4) .. (85.6,73.4) .. controls (80.13,73.4) and (75.7,68.97) .. (75.7,63.5) -- cycle ;
			\draw   (75.7,104.1) .. controls (75.7,98.63) and (80.13,94.2) .. (85.6,94.2) .. controls (91.07,94.2) and (95.5,98.63) .. (95.5,104.1) .. controls (95.5,109.57) and (91.07,114) .. (85.6,114) .. controls (80.13,114) and (75.7,109.57) .. (75.7,104.1) -- cycle ;
			\draw   (75.5,148.9) .. controls (75.5,143.43) and (79.93,139) .. (85.4,139) .. controls (90.87,139) and (95.3,143.43) .. (95.3,148.9) .. controls (95.3,154.37) and (90.87,158.8) .. (85.4,158.8) .. controls (79.93,158.8) and (75.5,154.37) .. (75.5,148.9) -- cycle ;
			\draw   (305.99,87.93) .. controls (305.99,82.46) and (310.42,78.03) .. (315.89,78.03) .. controls (321.35,78.03) and (325.79,82.46) .. (325.79,87.93) .. controls (325.79,93.4) and (321.35,97.83) .. (315.89,97.83) .. controls (310.42,97.83) and (305.99,93.4) .. (305.99,87.93) -- cycle ;
			\draw   (201.13,123.36) .. controls (201.13,117.89) and (205.56,113.46) .. (211.03,113.46) .. controls (216.5,113.46) and (220.93,117.89) .. (220.93,123.36) .. controls (220.93,128.82) and (216.5,133.26) .. (211.03,133.26) .. controls (205.56,133.26) and (201.13,128.82) .. (201.13,123.36) -- cycle ;
			\draw   (200.56,63.64) .. controls (200.56,58.18) and (204.99,53.74) .. (210.46,53.74) .. controls (215.92,53.74) and (220.36,58.18) .. (220.36,63.64) .. controls (220.36,69.11) and (215.92,73.54) .. (210.46,73.54) .. controls (204.99,73.54) and (200.56,69.11) .. (200.56,63.64) -- cycle ;
			\draw   (200.56,18.79) .. controls (200.56,13.32) and (204.99,8.89) .. (210.46,8.89) .. controls (215.92,8.89) and (220.36,13.32) .. (220.36,18.79) .. controls (220.36,24.25) and (215.92,28.69) .. (210.46,28.69) .. controls (204.99,28.69) and (200.56,24.25) .. (200.56,18.79) -- cycle ;
			\draw   (95.5,18.5) -- (197.56,18.78) ;
			\draw [shift={(200.56,18.79)}, rotate = 180.16] [fill={rgb, 255:red, 0; green, 0; blue, 0 }  ][line width=0.08]  [draw opacity=0] (5.36,-2.57) -- (0,0) -- (5.36,2.57) -- cycle    ;
			\draw    (25.7,88.92) -- (73.58,146.59) ;
			\draw [shift={(75.5,148.9)}, rotate = 230.3] [fill={rgb, 255:red, 0; green, 0; blue, 0 }  ][line width=0.08]  [draw opacity=0] (5.36,-2.57) -- (0,0) -- (5.36,2.57) -- cycle    ;
			\draw   (25.7,88.92) -- (72.83,103.23) ;
			\draw [shift={(75.7,104.1)}, rotate = 196.89] [fill={rgb, 255:red, 0; green, 0; blue, 0 }  ][line width=0.08]  [draw opacity=0] (5.36,-2.57) -- (0,0) -- (5.36,2.57) -- cycle    ;
			\draw  (26.37,88.92) -- (74.55,61.79) ;
			\draw [shift={(77.17,60.32)}, rotate = 150.62] [fill={rgb, 255:red, 0; green, 0; blue, 0 }  ][line width=0.08]  [draw opacity=0] (5.36,-2.57) -- (0,0) -- (5.36,2.57) -- cycle    ;
			\draw   (95.5,63.5) -- (197.56,63.64) ;
			\draw [shift={(200.56,63.64)}, rotate = 180.08] [fill={rgb, 255:red, 0; green, 0; blue, 0 }  ][line width=0.08]  [draw opacity=0] (5.36,-2.57) -- (0,0) -- (5.36,2.57) -- cycle    ;
			\draw    (95.5,18.5) -- (197.8,62.46) ;
			\draw [shift={(200.56,63.64)}, rotate = 203.25] [fill={rgb, 255:red, 0; green, 0; blue, 0 }  ][line width=0.08]  [draw opacity=0] (5.36,-2.57) -- (0,0) -- (5.36,2.57) -- cycle    ;
			\draw    (95.5,63.5) -- (197.8,19.96) ;
			\draw [shift={(200.56,18.79)}, rotate = 156.94] [fill={rgb, 255:red, 0; green, 0; blue, 0 }  ][line width=0.08]  [draw opacity=0] (5.36,-2.57) -- (0,0) -- (5.36,2.57) -- cycle    ;
			\draw   (95.3,148.9) -- (199.57,126.45) ;
			\draw [shift={(202.5,125.82)}, rotate = 167.85] [fill={rgb, 255:red, 0; green, 0; blue, 0 }  ][line width=0.08]  [draw opacity=0] (5.36,-2.57) -- (0,0) -- (5.36,2.57) -- cycle    ;
			\draw    (95.5,104.1) -- (199.56,125.22) ;
			\draw [shift={(202.5,125.82)}, rotate = 191.47] [fill={rgb, 255:red, 0; green, 0; blue, 0 }  ][line width=0.08]  [draw opacity=0] (5.36,-2.57) -- (0,0) -- (5.36,2.57) -- cycle    ;
			\draw     (220.93,123.36) -- (303.22,89.08) ;
			\draw [shift={(305.99,87.93)}, rotate = 157.39] [fill={rgb, 255:red, 0; green, 0; blue, 0 }  ][line width=0.08]  [draw opacity=0] (5.36,-2.57) -- (0,0) -- (5.36,2.57) -- cycle    ;
			\draw   (220.36,63.64) -- (303.1,87.11) ;
			\draw [shift={(305.99,87.93)}, rotate = 195.83] [fill={rgb, 255:red, 0; green, 0; blue, 0 }  ][line width=0.08]  [draw opacity=0] (5.36,-2.57) -- (0,0) -- (5.36,2.57) -- cycle    ;
			\draw   (220.36,18.79) -- (303.65,86.04) ;
			\draw [shift={(305.99,87.93)}, rotate = 218.92] [fill={rgb, 255:red, 0; green, 0; blue, 0 }  ][line width=0.08]  [draw opacity=0] (5.36,-2.57) -- (0,0) -- (5.36,2.57) -- cycle    ;
			
			\draw (55,24.57) node [anchor=north west][inner sep=0.75pt]  [font=\footnotesize]  {$1$};
			\draw (55,55) node [anchor=north west][inner sep=0.75pt]  [font=\footnotesize]  {$1$};
			\draw (311.33,84) node [anchor=north west][inner sep=0.75pt]  [font=\footnotesize]  {$t$};
			\draw (9.83,84) node [anchor=north west][inner sep=0.75pt]  [font=\footnotesize]  {$s$};
			\draw (225.33,122) node [anchor=north west][inner sep=0.75pt]  [font=\footnotesize]  {$2$};
			\draw (225,70.9) node [anchor=north west][inner sep=0.75pt]  [font=\footnotesize]  {$1$};
			\draw (225,10.4) node [anchor=north west][inner sep=0.75pt]  [font=\footnotesize]  {$1$};
			\draw (99,150.4) node [anchor=north west][inner sep=0.75pt]  [font=\footnotesize]  {$1$};
			\draw (99,93) node [anchor=north west][inner sep=0.75pt]  [font=\footnotesize]  {$1$};
			\draw (99,65.4) node [anchor=north west][inner sep=0.75pt]  [font=\footnotesize]  {$1$};
			\draw (99.33,45.9) node [anchor=north west][inner sep=0.75pt]  [font=\footnotesize]  {$1$};
			\draw (99,25.4) node [anchor=north west][inner sep=0.75pt]  [font=\footnotesize]  {$1$};
			\draw (99,4.4) node [anchor=north west][inner sep=0.75pt]  [font=\footnotesize]  {$1$};
			\draw (205,12.9) node [anchor=north west][inner sep=0.75pt]  [font=\footnotesize]  {$1$};
			\draw (205,58.58) node [anchor=north west][inner sep=0.75pt]  [font=\footnotesize]  {$2$};
			\draw (205,117.08) node [anchor=north west][inner sep=0.75pt]  [font=\footnotesize]  {$3$};
			\draw (78,13) node [anchor=north west][inner sep=0.75pt]  [font=\footnotesize]  {$g_{1}$};
			\draw (78,145) node [anchor=north west][inner sep=0.75pt]  [font=\footnotesize]  {$g_{4}$};
			\draw (78,99) node [anchor=north west][inner sep=0.75pt]  [font=\footnotesize]  {$g_{3}$};
			\draw (78,58) node [anchor=north west][inner sep=0.75pt]  [font=\footnotesize]  {$g_{2}$};
			\draw (55,115) node [anchor=north west][inner sep=0.75pt]  [font=\footnotesize]  {$1$};
			\draw (55,86) node [anchor=north west][inner sep=0.75pt]  [font=\footnotesize]  {$1$};
			
		\end{tikzpicture}

		\caption{\cref{alg:balESW} on \cref{ex:one}. We omit all items that give value $0$ to all agents.}
		\label{fig:example-one}
	\end{figure}
	
	Note that \cref{ex:one} covers the case where $m \mod n\leq |N_0|$. We now show an example of \cref{alg:balESW} on an instance where $m\mod n > |N_0|$.
	\begin{example}\label{ex:two}
		We now extend the previous example by simply adding one item of no value to any agent. Here, we have $n=3$ and $m=8$. For agent $1$, we have $v_1(g_1)=v_1(g_2)=1$ and $\tau_1=\sfrac{1}{2}$. For agent $2$, we have $v_2(g_1)=v_2(g_2)=1$ and $\tau_2=\sfrac{1}{2}$. Finally for agent $3$, we have $v_3(g_3)=v_3(g_4)=1$ and $\tau_3=\sfrac{2}{3}$. All remaining item values are $0$. Observe that on this instance $N_0=\{3\}$ and $N_1=\{1,2\}$ and $m \mod n=2>|N_0|$. Consequently, the flow graph that will be constructed for this instance is as depicted in \cref{fig:example-two}. 
		
		In this instance, no flow can saturate the edge $(t',t)$ as the maximum flow that can come to $t'$ is $2$. Analogously, in the instance, at least one of agents $1$ and $2$ needs to receive two items of value one, while the other receives one item of value one. Clearly, this is not possible as they only get value one from $g_1$ and $g_2$. Consequently, all allocations in this instance have \ESW{} of zero. 
	\end{example}
	
	\begin{figure}
		\centering

		\tikzset{every picture/.style={line width=0.9pt}} 
		
		\begin{tikzpicture}[x=0.7pt,y=0.7pt,yscale=-1,xscale=1]
			
			\draw   (25.7,88.92) -- (73.96,20.95) ;
			\draw [shift={(75.7,18.5)}, rotate = 125.37] [fill={rgb, 255:red, 0; green, 0; blue, 0 }  ][line width=0.08]  [draw opacity=0] (5.36,-2.57) -- (0,0) -- (5.36,2.57) -- cycle    ;
			\draw   (5.9,88.92) .. controls (5.9,83.45) and (10.33,79.02) .. (15.8,79.02) .. controls (21.27,79.02) and (25.7,83.45) .. (25.7,88.92) .. controls (25.7,94.39) and (21.27,98.82) .. (15.8,98.82) .. controls (10.33,98.82) and (5.9,94.39) .. (5.9,88.92) -- cycle ;
			\draw   (75.7,18.5) .. controls (75.7,13.03) and (80.13,8.6) .. (85.6,8.6) .. controls (91.07,8.6) and (95.5,13.03) .. (95.5,18.5) .. controls (95.5,23.97) and (91.07,28.4) .. (85.6,28.4) .. controls (80.13,28.4) and (75.7,23.97) .. (75.7,18.5) -- cycle ;
			\draw   (75.7,63.5) .. controls (75.7,58.03) and (80.13,53.6) .. (85.6,53.6) .. controls (91.07,53.6) and (95.5,58.03) .. (95.5,63.5) .. controls (95.5,68.97) and (91.07,73.4) .. (85.6,73.4) .. controls (80.13,73.4) and (75.7,68.97) .. (75.7,63.5) -- cycle ;
			\draw   (75.7,104.1) .. controls (75.7,98.63) and (80.13,94.2) .. (85.6,94.2) .. controls (91.07,94.2) and (95.5,98.63) .. (95.5,104.1) .. controls (95.5,109.57) and (91.07,114) .. (85.6,114) .. controls (80.13,114) and (75.7,109.57) .. (75.7,104.1) -- cycle ;
			\draw   (75.5,148.9) .. controls (75.5,143.43) and (79.93,139) .. (85.4,139) .. controls (90.87,139) and (95.3,143.43) .. (95.3,148.9) .. controls (95.3,154.37) and (90.87,158.8) .. (85.4,158.8) .. controls (79.93,158.8) and (75.5,154.37) .. (75.5,148.9) -- cycle ;
			\draw   (272.99,38.43) .. controls (272.99,32.96) and (277.42,28.53) .. (282.89,28.53) .. controls (288.35,28.53) and (292.79,32.96) .. (292.79,38.43) .. controls (292.79,43.9) and (288.35,48.33) .. (282.89,48.33) .. controls (277.42,48.33) and (272.99,43.9) .. (272.99,38.43) -- cycle ;
			\draw   (201.13,123.36) .. controls (201.13,117.89) and (205.56,113.46) .. (211.03,113.46) .. controls (216.5,113.46) and (220.93,117.89) .. (220.93,123.36) .. controls (220.93,128.82) and (216.5,133.26) .. (211.03,133.26) .. controls (205.56,133.26) and (201.13,128.82) .. (201.13,123.36) -- cycle ;
			\draw   (200.56,63.64) .. controls (200.56,58.18) and (204.99,53.74) .. (210.46,53.74) .. controls (215.92,53.74) and (220.36,58.18) .. (220.36,63.64) .. controls (220.36,69.11) and (215.92,73.54) .. (210.46,73.54) .. controls (204.99,73.54) and (200.56,69.11) .. (200.56,63.64) -- cycle ;
			\draw   (200.56,18.79) .. controls (200.56,13.32) and (204.99,8.89) .. (210.46,8.89) .. controls (215.92,8.89) and (220.36,13.32) .. (220.36,18.79) .. controls (220.36,24.25) and (215.92,28.69) .. (210.46,28.69) .. controls (204.99,28.69) and (200.56,24.25) .. (200.56,18.79) -- cycle ;
			\draw    (95.5,18.5) -- (197.56,18.78) ;
			\draw [shift={(200.56,18.79)}, rotate = 180.16] [fill={rgb, 255:red, 0; green, 0; blue, 0 }  ][line width=0.08]  [draw opacity=0] (5.36,-2.57) -- (0,0) -- (5.36,2.57) -- cycle    ;
			\draw   (25.7,88.92) -- (73.58,146.59) ;
			\draw [shift={(75.5,148.9)}, rotate = 230.3] [fill={rgb, 255:red, 0; green, 0; blue, 0 }  ][line width=0.08]  [draw opacity=0] (5.36,-2.57) -- (0,0) -- (5.36,2.57) -- cycle    ;
			\draw  (25.7,88.92) -- (72.83,103.23) ;
			\draw [shift={(75.7,104.1)}, rotate = 196.89] [fill={rgb, 255:red, 0; green, 0; blue, 0 }  ][line width=0.08]  [draw opacity=0] (5.36,-2.57) -- (0,0) -- (5.36,2.57) -- cycle    ;
			\draw   (26.37,88.92) -- (74.55,61.79) ;
			\draw [shift={(77.17,60.32)}, rotate = 150.62] [fill={rgb, 255:red, 0; green, 0; blue, 0 }  ][line width=0.08]  [draw opacity=0] (5.36,-2.57) -- (0,0) -- (5.36,2.57) -- cycle    ;
			\draw     (95.5,63.5) -- (197.56,63.64) ;
			\draw [shift={(200.56,63.64)}, rotate = 180.08] [fill={rgb, 255:red, 0; green, 0; blue, 0 }  ][line width=0.08]  [draw opacity=0] (5.36,-2.57) -- (0,0) -- (5.36,2.57) -- cycle    ;
			\draw     (95.5,18.5) -- (197.8,62.46) ;
			\draw [shift={(200.56,63.64)}, rotate = 203.25] [fill={rgb, 255:red, 0; green, 0; blue, 0 }  ][line width=0.08]  [draw opacity=0] (5.36,-2.57) -- (0,0) -- (5.36,2.57) -- cycle    ;
			\draw    (95.5,63.5) -- (197.8,19.96) ;
			\draw [shift={(200.56,18.79)}, rotate = 156.94] [fill={rgb, 255:red, 0; green, 0; blue, 0 }  ][line width=0.08]  [draw opacity=0] (5.36,-2.57) -- (0,0) -- (5.36,2.57) -- cycle    ;
			\draw    (95.3,148.9) -- (199.57,126.45) ;
			\draw [shift={(202.5,125.82)}, rotate = 167.85] [fill={rgb, 255:red, 0; green, 0; blue, 0 }  ][line width=0.08]  [draw opacity=0] (5.36,-2.57) -- (0,0) -- (5.36,2.57) -- cycle    ;
			\draw   (95.5,104.1) -- (199.56,125.22) ;
			\draw [shift={(202.5,125.82)}, rotate = 191.47] [fill={rgb, 255:red, 0; green, 0; blue, 0 }  ][line width=0.08]  [draw opacity=0] (5.36,-2.57) -- (0,0) -- (5.36,2.57) -- cycle    ;
			\draw    (220.93,123.36) -- (344.18,76.49) ;
			\draw [shift={(346.99,75.43)}, rotate = 159.18] [fill={rgb, 255:red, 0; green, 0; blue, 0 }  ][line width=0.08]  [draw opacity=0] (5.36,-2.57) -- (0,0) -- (5.36,2.57) -- cycle    ;
			\draw    (220.36,63.64) -- (270.28,39.72) ;
			\draw [shift={(272.99,38.43)}, rotate = 154.4] [fill={rgb, 255:red, 0; green, 0; blue, 0 }  ][line width=0.08]  [draw opacity=0] (5.36,-2.57) -- (0,0) -- (5.36,2.57) -- cycle    ;
			\draw    (220.36,18.79) -- (270.18,37.38) ;
			\draw [shift={(272.99,38.43)}, rotate = 200.47] [fill={rgb, 255:red, 0; green, 0; blue, 0 }  ][line width=0.08]  [draw opacity=0] (5.36,-2.57) -- (0,0) -- (5.36,2.57) -- cycle    ;
			\draw   (346.99,75.43) .. controls (346.99,69.96) and (351.42,65.53) .. (356.89,65.53) .. controls (362.35,65.53) and (366.79,69.96) .. (366.79,75.43) .. controls (366.79,80.9) and (362.35,85.33) .. (356.89,85.33) .. controls (351.42,85.33) and (346.99,80.9) .. (346.99,75.43) -- cycle ;
			\draw    (292.79,38.43) -- (344.51,73.74) ;
			\draw [shift={(346.99,75.43)}, rotate = 214.32] [fill={rgb, 255:red, 0; green, 0; blue, 0 }  ][line width=0.08]  [draw opacity=0] (5.36,-2.57) -- (0,0) -- (5.36,2.57) -- cycle    ;
			
			\draw (55,24.57) node [anchor=north west][inner sep=0.75pt]  [font=\footnotesize]  {$1$};
			\draw (55,52.9) node [anchor=north west][inner sep=0.75pt]  [font=\footnotesize]  {$1$};
			\draw (278.33,32.4) node [anchor=north west][inner sep=0.75pt]  [font=\footnotesize]  {$t'$};
			\draw (9.83,84) node [anchor=north west][inner sep=0.75pt]  [font=\footnotesize]  {$s$};
			\draw (221,109) node [anchor=north west][inner sep=0.75pt]  [font=\footnotesize]  {$2$};
			\draw (221,63) node [anchor=north west][inner sep=0.75pt]  [font=\footnotesize]  {$2$};
			\draw (221,8.9) node [anchor=north west][inner sep=0.75pt]  [font=\footnotesize]  {$2$};
			\draw (95,150.4) node [anchor=north west][inner sep=0.75pt]  [font=\footnotesize]  {$1$};
			\draw (99,91.4) node [anchor=north west][inner sep=0.75pt]  [font=\footnotesize]  {$1$};
			\draw (99,65.4) node [anchor=north west][inner sep=0.75pt]  [font=\footnotesize]  {$1$};
			\draw (99,45.9) node [anchor=north west][inner sep=0.75pt]  [font=\footnotesize]  {$1$};
			\draw (99,25.4) node [anchor=north west][inner sep=0.75pt]  [font=\footnotesize]  {$1$};
			\draw (99,4.4) node [anchor=north west][inner sep=0.75pt]  [font=\footnotesize]  {$1$};
			\draw (205.83,12.9) node [anchor=north west][inner sep=0.75pt]  [font=\footnotesize]  {$1$};
			\draw (204.83,58.58) node [anchor=north west][inner sep=0.75pt]  [font=\footnotesize]  {$2$};
			\draw (205.33,117.08) node [anchor=north west][inner sep=0.75pt]  [font=\footnotesize]  {$3$};
			\draw (79.83,12) node [anchor=north west][inner sep=0.75pt]  [font=\footnotesize]  {$g_{1}$};
			\draw (78.33,145) node [anchor=north west][inner sep=0.75pt]  [font=\footnotesize]  {$g_{4}$};
			\draw (78.83,99) node [anchor=north west][inner sep=0.75pt]  [font=\footnotesize]  {$g_{3}$};
			\draw (78.33,59) node [anchor=north west][inner sep=0.75pt]  [font=\footnotesize]  {$g_{2}$};
			\draw (55,116) node [anchor=north west][inner sep=0.75pt]  [font=\footnotesize]  {$1$};
			\draw (55,87) node [anchor=north west][inner sep=0.75pt]  [font=\footnotesize]  {$1$};
			\draw (352.33,70) node [anchor=north west][inner sep=0.75pt]  [font=\footnotesize]  {$t$};
			\draw (315.33,41.9) node [anchor=north west][inner sep=0.75pt]  [font=\footnotesize]  {$3$};
			\draw (245,14.4) node [anchor=north west][inner sep=0.75pt]  [font=\footnotesize]  {$d=1$};
			\draw (245,50.9) node [anchor=north west][inner sep=0.75pt]  [font=\footnotesize]  {$d=1$};

		\end{tikzpicture}
		
		\caption{\cref{alg:balESW} on \cref{ex:two}. We omit all items that give value $0$ to all agents. The demand on the edge is denoted by $d=x$ and we omit mentioning the demand on edges with demand $0$.}
		\label{fig:example-two}
	\end{figure}

	\begin{restatable}{proposition}{ESWbalanced}\label{prop:balancedESW}
		Given $I=\langle N,M,v,\tau\rangle$ with binary goods and heterogeneous quantiles, \cref{alg:balESW} finds a max \ESW{} balanced allocation in polynomial time.  
	\end{restatable}
	
	\begin{proof} Before we show that \cref{alg:balESW} maximizes \ESW{}, we first establish that it always returns a balanced allocation. Recall that under a balanced allocation $m \mod n$ agents must receive $\lceil \frac{m}{n}\rceil$ items while the remaining agents receive $\lfloor \frac{m}{n}\rfloor$ items.\\

		\noindent \textbf{Balancedness.} We first prove that \cref{alg:balESW} does indeed return a balanced allocation. Given $I$, let $G$ and $c$ be the flow graph and its capacities, as defined in \cref{alg:balESW} and $f$ be an integral maximum flow from $s$ to $t$ on $G$. If some an edge incident on $t$ has less flow through it than its capacity, then clearly the allocation returned must be balanced as picked in Step \ref{step:no-instance-return}. '
		
		Suppose not, that is,  each incoming edge to $t$ is saturated. Observe that the allocation returned by \cref{alg:balESW} will have $|A_i|=k_i$ for each $i\in N$, as chosen in  step \ref{step:yes-instance-return}. Consequently, it is sufficient to show that exactly $m \mod n$ agents have $k_i=\lfloor  \frac{m}{n}\rfloor+1$ and the remaining have $k_i=\lfloor  \frac{m}{n}\rfloor$, to ensure that the allocation is balanced. Whenever $m \mod N\leq |N_0|$, we ensure this is precisely the case in step \ref{step:pick-within-Nzero}, where we pick an arbitrary set of agents within $N_0$ of size $m\mod n$ for whom we set $k_i=\lfloor  \frac{m}{n}\rfloor+1$, while all other agents in $N_0$ and all agents in $N_1$ satisfy $k_i=\lfloor  \frac{m}{n}\rfloor$.
		
		Finally, when $m\mod n> |N_0|$, it must hold that $|N_1|>0$ and that $t$ only has one incoming edge, which is from $t'$. Further, as this edge is saturated, $f_{(t',t)}=c_{(t',t)}=\sum_{i\in N_1}(m_i-1) + (m \mod n) -|N_0|$. Note that $t'$ only has incoming edges from agents in $N_1$ and for all $i\in N_1$, $c_{(i,t)}=m_i$. These agents in $N_1$ each only have one outgoing edge, which is to $t'$ and has a demand of $m_i-1$.  Consequently, it must hold that $f_{(i,t)}=m_i=c_{(i,t)}$ for exactly $(m \mod n)-|N_0|$ agents in $N_1$ and for all remaining agents it must hold that $f_{(i,t)}=m_i-1<c_{(i,t)}$. As a result, the number of agents for whom $k_i=\lceil \frac{m}{n}\rceil$ is $|N_0|+(m \mod n)-|N_0|=m \mod n$, while the remaining agents all satisfy $k_i=\lfloor \frac{m}{n}\rfloor$.\\
		
		\noindent \textbf{Maximum ESW.} We now prove that \cref{alg:balESW} returns a maximum \ESW{} allocation. Recall that the maximum \ESW{} possible under binary valuations is $1$ and in order for an allocation $A$ to satisfy $\ESW{}(A)=1$, for each agent $i\in N$, $A_i$ must contain at least $\min (|A_i|, |A_i|-\lceil \tau_i|A_i|\rceil +1)$ items of value $1$. For each $i\in N$, let $m_i=\min (\lceil \frac{m}{n}\rceil, \lceil \frac{m}{n}\rceil -\lceil \tau_i\lceil\frac{m}{n}\rceil\rceil +1)$, as defined in \cref{alg:balESW}. That is, $m_i$ denotes the number of items of value $1$ that $i$ needs in a bundle of size $\lceil \frac{m}{n} \rceil$ to receive a value of $1$. 
		
		By definition of the  $N_0$, whenever $\lceil \frac{m}{n}\rceil\neq \lfloor \frac{m}{n}\rfloor$, for any bundle $B$ s.t. $|B|=\lfloor \frac{m}{n}\rfloor$, $v_i(B)=1$ if and only if $B$ contains at least $m_i$ items of value $1$ for every $i\in N_0$. Analogously, by definition of $N_1$, whenever $\lceil \frac{m}{n}\rceil\neq \lfloor \frac{m}{n}\rfloor$, for any bundle $B$ s.t. $|B|=\lfloor \frac{m}{n}\rfloor$, $v_i(B)=1$ if and only if $B$ contains at least $m_i-1$ items of value $1$ for every $i\in N_1$. Consequently, for a balanced allocation to satisfy $\ESW{}(A)=1$, for every agent $i\in N_0$, $A_i$ must contain $m_i$ items of value $1$ and for every $i\in N_1$, $A_i$ must contain $m_i - |A_i|+\lceil \frac{m}{n}\rceil$ items.

		Recall that \cref{alg:balESW} constructs a flow graph based on whether the size of $N_0$ is less than $m \mod n$ or not. Further, whenever the max flow on this graph does not saturate every incoming edge to $t$, the algorithm returns an arbitrary allocation. Consequently, we now show that an allocation with \ESW{} $1$ exists if and only if the max flow of the constructed graph does indeed saturate each incoming edge to $t$. As the structure of the flow graph depends on if the value of $|N_0|$ is less than  $m \mod n$, we argue the two cases separately.
		\vspace{2mm}
		
		\noindent \underline{Case 1: $m \mod n \leq |N_0|$.} Note that in this case, $t$ has an incoming edge from every agent $i\in N$, where $c_{(i,t)}=m_i$ for $i\in N_0$ and $c_{(i,t)}=m_i-1$ for $i\in N_1$. Consequently, it suffices to show that an allocation with \ESW{} of $1$ exists if and only if there exists a flow $f$ on the graph $G$ as defined in \cref{alg:balESW} s.t. $f_{(i,t)}=c_{(i,t)}$ for each $i\in N$. Note that in this case, the demand on every edge is $0$, so every valid flow will satisfy the demand. 
		
		Let there exist an allocation $A$ s.t. $\ESW{}(A)=1$. In order to  construct a maximum flow on this graph,  we first define bundles $B_1,\cdots, B_n$ as follows: for each $i\in N_0$, choose $B_i\subseteq A_i$ s.t. $|B_i|=m_i$ and $v_i(g)=1$ for each $g\in B_i$. Similarly, for each $i\in N_1$ choose $B_i\subseteq A_i$ s.t. $|B_i|=m_i-1$ and $v_i(g)=1$ for each $g\in B_i$. Note that such bundles must exist for each $i$ as $\ESW{}(A)=1$. Further, depending on $A_i$, there may be multiple ways to select $B_i$, and we are indifferent across all choices which have the necessary size while only containing items of value $1$ to $i$. As $B_i\subseteq A_i$ for each $i\in N$, the sets $B_1,\cdots, B_n$ must be pairwise disjoint. 
		
		Define integral flow $f:E\rightarrow \mathbb{Z}_+$ as follows:
		\begin{itemize}
			\item For each $g\in M\setminus (\cup_{i\in N}B_i)$ set $f_{(s,g)}=0$ and for all $(g,i)\in E$, set $f_{(g,i)}=0$.
			\item For each $g$ s.t. there exist $i$ for which $g\in B_i$, set $f_{(s,g)}=1$ and $f_{(g,i)}=1$.
			\item For each $i\in N$, set $f_{(i,t)}=c_{(i,t)}$
		\end{itemize}
		
		Note that, by choice of $B_i$s, it must hold that for each $i\in N$, $\sum_{g:v_i(g)=1}f_{(g,i)}=|B_i|=c_{(i,t)}$. Consequently, $f$ is a valid flow and saturates every incoming edge to $t$.
		
		Conversely, let a flow $f$ exist s.t. for each $i\in N$ $\sum_{g:v_i(g)=1}f_{(g,i)}=c_{(i,t)}$. Without loss of generality, let $f$ be an integral flow. Further, define set $B_i=\{g|f_{(g,i)=1}$. As the each item $g$ can only send out at most $1$ unit of flow, the sets $B_i$ for $i\in N$ must be pairwise disjoint. Further, for each $i\in N_0$, $|B_i|=c_{(i,t)}=m_i$ and for each $i\in N_1$, $|B_i|=c_{(i,t)}=m_i-1$.  Let $S$ and $A$ be as chosen in \cref{alg:balESW}. Observe that $B_i\subseteq A_i$, for each $i$, and as $|B_i|=m_i$ for each $i\in N_0$ and $|B_i|=m_i-1$ for each $i\in N_1$, it must hold that $\ESW{}(A)=1$.
		\vspace{2mm}
		
		\noindent \underline{Case 2: $m \mod n > |N_0|$.} In this case, we have that the flow graph constructed contains vertex $t'$ and each agent in $N_1$ has only one outgoing edge to $t'$ with $c_{(i,t')}=m_i$ and $d_{(i,t')}=m_i -1$ for each $i\in N_1$. In turn, $t'$ has one outgoing to $t$ with $c_{(t',t)}=\sum_{i\in N_1} (m_i-1)+(m\mod n) -|N_0|$. As a result, under any valid flow $f$, at most $(m \mod n)-|N_0|$ agents in $N_1$ can satisfy $f_{(i,t')}=m_i$. 
		
		Let there exist an allocation $A$ s.t. $\ESW{}(A)=1$. Recall that whenever $m \mod n=0$, $N_0$ must be empty. It is straightforward to see that least $(m \mod n)- |N_0|$ agents in $N_1$ satisfy $|A_i|=\lfloor \frac{m}{n}\rfloor +1$. If $m\mod n>0$, select $S\subseteq N_1$ s.t. $|S|=(m \mod n)- |N_0|$ and $|A_i|\lfloor \frac{m}{n}\rfloor +1$. Else if $m \mod n =0$, let $S=\emptyset$. As in the previous case,  define bundles $B_1,\cdots, B_n$ as follows: for each $i\in N_0\cup S$, choose $B_i\subseteq A_i$ s.t. $|B_i|=m_i$ and $v_i(g)=1$ for each $g\in B_i$. Similarly, for each $i\in N_1\setminus S$ choose $B_i\subseteq A_i$ s.t. $|B_i|=m_i-1$ and $v_i(g)=1$ for each $g\in B_i$. As before, such bundles must exist for each $i$ as $\ESW{}(A)=1$. As $B_i\subseteq A_i$ for each $i\in N$, the sets $B_1,\cdots, B_n$ must be pairwise disjoint. 
		
		Define integral flow $f:E\rightarrow \mathbb{Z}_+$ as follows:
		\begin{itemize}
			\item For each $g\in M\setminus (\cup_{i\in N}B_i)$ set $f_{(s,g)}=0$ and for all $(g,i)\in E$, set $f_{(g,i)}=0$.
			\item For each $g$ s.t. there exist $i$ for which $g\in B_i$, set $f_{(s,g)}=1$ and $f_{(g,i)}=1$.
			\item For each $i\in N_0\cup S$, set $f_{(i,t')}=m_i$.
			\item For each $i\in N_1 \setminus S$, set $f_{(i,t')}=m_i-1$ and set $f_{(t',t)}=c_{(t',t)}$.
		\end{itemize}
		
		Note that, by choice of $B_i$s, it must hold that for each $i\in N$, $\sum_{g:v_i(g)=1}f_{(g,i)}=|B_i|$. For each $i\in N_0\cup S$, we have that $|B_i|=m_i$ and for each $i\in N_1\setminus S$, $|B_i|=m_i-1$. Consequently, $f$ is a valid flow that satisfies $d_e\leq f_e\le c_e$ for each $e\in E$ and saturates every incoming edge to $t$.
		
		Conversely, let a flow $f$ exist s.t. for each $e\in E$, $d_e\leq f_e \leq c_e$ and every incoming edge to $t$ is saturated.  Let $S$ and $A$ be as chosen in \cref{alg:balESW}. Further, define the sets $B_i=\{g|f_{(g,i)=1}$ for each $i\in N$. As the each item $g$ can only send out at most $1$ unit of flow, the sets $B_1,\cdots, B_n$ must be pairwise disjoint. Further, for each $i\in N_0\cup S$, $|B_i|=c_{(i,t)}=m_i$ and for each $i\in N_1\setminus S$, $|B_i|=c_{(i,t)}=m_i-1$.  Observe that $B_i\subseteq A_i$, each $i\in N_0\cup S$ receives at least $m_i$ items of value $1$ and each $i\in N_1\setminus S$ receives at least $m_i-1$ items of value $1$. Consequently, it must hold that $\ESW{}(A)=1$.\\
		
		\noindent \textbf{Running Time.} The flow graph constructed has at most $n+m+3$ vertices and at most $nm+m+n+1$ edges. A maximum flow can be found in polynomial time on this graph (See for example section 7.7 of \cite{kleinberg2006algorithm} on computing maximum flows with demands on edges) and the allocation can accordingly be built in polynomial time. As a result, \cref{alg:balESW} runs in time polynomial in the given instance. 
	\end{proof}
	
	Combining \cref{lem:binRednESW} and \cref{prop:balancedESW}, we get the following theorem. 
	\begin{theorem}\label{thm:balESW}
		Given $I=\langle N,M,v,\tau\rangle$ with heterogeneous quantiles  where $m=kn$, a balanced allocation with maximum \ESW{} can be found in polynomial time. 
	\end{theorem}
	

	\section{Unconstrained Allocations}\label{sec:unbal}
	
	Typical work on allocating indivisible items does not require allocations to be balanced. Further, there are many practical settings, such as online classrooms in distance education, where allocations need not be balanced. Thus, for completeness, we now turn to unconstrained allocation. We find that it is possible to give significantly better guarantees on \USW{}, compared to the balanced setting. In contrast, maximum \ESW{} now becomes intractable for a large subset of the quantiles. Omitted proofs are deferred to \cref{app:unbalanced}.

	
	\subsection{Utilitarian Social Welfare}
	
	We find that while maximizing \USW{} still remains intractable, we are able to circumvent hardness of approximation and achieve a near exact approximation algorithm that even works for heterogeneous quantiles. 
	
	
	\subsubsection{Intractability.}
	
	We first show that for non-identical agents with quantile $\tau=0$ for all agents, the problem of maximizing social welfare proves to be NP-hard for goods.  We give a reduction from the \textup{\textsc{Exact3Cover}} (X3C) problem, which is known to be NP-hard \citep{garey1979computers}.
	\begin{restatable}{theorem}{USWreductionNP}\label{thm:NP-USW}
		Given instance $I=\langle N, M, v,\tau \rangle$  where $\max_{i\in N} \tau_i<1$, finding a maximum \USW{} allocation is NP-complete. 
	\end{restatable} 
	
	\begin{proof}
		It is straightforward to see that this problem is in NP. Given instance $I$ and a target \USW{} value $\alpha$,  an allocation $A$ with $\USW{}(A)\geq\alpha$, serves as a polynomial sized certificate. It can be checked in polynomial time that $A$ has \USW{} at least $\alpha$.
		
		We now provide a single reduction format that will work for all choice of homogeneous quantiles $\tau\in [0,1)$. While the exact construction is quantile-dependent, the ideas and arguments behind them will work largely analogously.         
		We give our reduction  from the well known \textup{\textsc{Exact3Cover}} (X3C) problem, which is known to be NP-hard \citep{garey1979computers}. Under the \textup{\textsc{Exact3Cover}} problem we are given a universe of $3t$ elements $\mathcal{U}$ and a family of sets $\mathcal{S}=\{S_1,\cdots, S_{\ell}\}$ s.t. for each $j\in [t]$, $S_j\subset \mathcal{U}$ and $|S_j|=3$. The aim is to find $t$ mutually disjoint sets $S_{j_1},\cdots, S_{j_t}$ that cover the given set of elements, i.e., $\bigcup_{p\in [\ell]}S_{j_p}=\mathcal{U}.$ \\
		
		\noindent \textbf{Reduction Overview} Our reduction will construct agents corresponding to sets and items corresponding to elements of the given X3C instance. Agents will only have value $1$ for element items that are contained in their corresponding set and value $0$ for other element items. Additionally there will be $k_1t$ "buffer" items that give all agents value $1$ and $\ell -t$ ``prized'' items that give all agents a value of $2$. Additionally, there will be $k_0t+k_0'(\ell-t)$ ``dummy items'' whenever $\tau>0$ which give each agent a value of $0$. The aim is to ensure that for a fixed quantile value, a \USW{} of $2\ell -t$ is possible if and only if an X3C exists in the given instance. Specifically, we choose the values of $k_0,$ $k_0'$ and $k_1$ s.t. an X3C exists if and only if it is possible to give $t$ distinct set agents their three corresponding element items, $k_0$ dummy items and $k_1$ buffer items and the remaining $\ell -t$ agents get one prized item and $k_0'$ dummy items each. \\
		
		\noindent \textbf{Choosing reduction parameters.} Given a fixed $\tau\in [0,1)$, we choose the values of $k_0$, $k_0'$ and $k_1$ as follows:
		
		\begin{itemize}
			\item If $\tau=0$, set $k_0=k_0'=k_1=0$,
			\item If $\tau\in (\frac{1}{p+1},\frac{1}{p}]$, for a positive integer $p>2$, set $k_0'=0$, $k_0=1$ and $k_1=p-3$,
			\item If $\tau \in (\sfrac{1}{3},\sfrac{2}{5}]$, set $k_0'=0$, $k_0=2$ and $k_1=1$,
			\item If $\tau\in (\sfrac{1}{2},\sfrac{1}{2}]$, set $k_0'=0$, $k_0=2$ and $k_1=0$, and finally,
			\item If $\tau \in (\sfrac{1}{2},1)$, set positive integers $k_0'\in [\frac{2\tau-1}{1-\tau},\frac{\tau}{1-\tau})$, $k_0\in [\frac{4\tau-1}{1-\tau},\frac{3\tau}{1-\tau})$ and $k_1=0$.
		\end{itemize}
		
		Note that it remains to prove that there do exist values for the choice of $k_0'$ and $k_0$ when $\tau>\sfrac{1}{2}.$ Here, it suffices to show that the respective intervals contain at least one positive integer. To this end, first consider $k_0'\in [\frac{2\tau-1}{1-\tau},\frac{\tau}{1-\tau})$. It is non-trivial to see if this is indeed a non-empty interval. We first observe that as $\tau<1$, $2\tau-1<2tau-\tau=\tau$. As a result, this interval is indeed well defined. Further, as $\tau>\sfrac{1}{2}$, we have that $\frac{2\tau-1}{1-\tau}> 0$. Further, the length of this interval is $\frac{\tau}{1-\tau}-\frac{2\tau-1}{1-\tau}=\frac{\tau-2\tau+1}{1-\tau}=1$. Consequently, this interval will contain exactly one positive integer for each choice of $\tau\in (\sfrac{1}{2},1)$.
		
		For the choice of $k_0\in[\frac{4\tau-1}{1-\tau},\frac{3\tau}{1-\tau})$, we have that $3\tau=4\tau-\tau>4\tau-1$ and thus the interval is non-empty. As $\tau>\sfrac{1}{2}$, we have that $4\tau-1>1$ and as a result $\frac{4\tau-1}{1-\tau}>1$. Further, the length of this interval is $\frac{3\tau-4\tau+1}{1-\tau}=\frac{1-\tau}{1-\tau}=1$ . Consequently, the interval must contain exactly one positive integer. As a result, for each $\tau\in [0,1)$, our choices of $k_0$, $k_0'$ and $k_1$ are well defined. \\
		
		\noindent \textbf{Reduction Setup.} Formally, given an instance of \textup{\textsc{Exact3Cover}} $\langle \mathcal{U}, \mathcal{S} \rangle$,  we construct an instance of our problem  with $\ell$ agents and $3t +t(k_0+k_1)t+ (k_0'+1)(\ell+t)$ items as follows: 
		For each set $S_j$, create an agent $i_j$. For each element $e\in \mathcal{U}$, create an item $g_e$. Create an additional set of prized $\ell-t$ items $G'$, a set of $k_1t$ buffer items $B$ and a set of $k_0t+k_0'(\ell-t)$ dummy items $D$.
		
		For each agent $i$, set the quantile value $\tau_i=\tau$ as given. The agent values are chosen as follows: for any $i\in N$, we set $i$'s value for a prized item $g'\in G$ as $v_i(g')=2$, for a buffer item $b\in B$ as $v_i(b)=1$ and for a dummy item $d\in D$ as $v_i(d)=0$. Further, for each $j\in [t]$ we set $v_{i_j}(g_e)=1$ if $e\in S_j$ and $v_{i_j}(g_e)=0$ otherwise.
		
		Before showing the correctness of our reduction, we first prove a series of claims regarding the choice of $k_0$, $k_0'$ and $k_1$ and how they impact the quantile representative of bundles of relevant sizes.\\
		
		\noindent \textbf{Claim 1.} For each $\tau\in (0,1)$, we have that $\lceil \tau(k_0+3+k_1)\rceil =k_0+1$ and for any $q\in \mathbb{Z}_+$, $\lceil \tau(k_0+k_1+3+q)\rceil \leq k_0+q$.
		\vspace{2mm}
		
		\noindent {\em Proof of Claim 1.} We prove this claim for each relevant interval separately.\\
		
		\noindent \underline{$\tau\in (\frac{1}{p+1},\frac{1}{p}]$ for positive integer $p\geq 3$}. Here we have that $k_0=1$ and $k_1=p-3$. Consequently $\lceil \tau(k_0+k_1+3)\rceil=\lceil \tau (p+1)\rceil$. As $\tau\in (\frac{1}{p+1},\frac{1}{p}]$ we have that  $\lceil \tau(k_0+k_1+3)\rceil=2=k_0+1$. 
		
		Further, for any choice of $q\geq 1$, we have that $\lceil \tau(k_0+k_1+3+q)\rceil=\lceil \tau(p+1+q)\rceil$. Before discussing a general choice of $q\in \mathbb{Z}_+,$ first consider $q=1$.         
		Here, as $\tau\in (\frac{1}{p+1},\frac{1}{p}]$ and $p\geq 3$, we have that $\tau(p+2)\leq 2$. Now consider an arbitrary integer $q>1$. We have that 
		$$\lceil \tau (p+1+q)\rceil=\lceil \tau (p+2+q-1)\rceil\leq 2+q-1=k_0+q.$$
		
		\noindent \underline{$\tau \in (\sfrac{1}{3},\sfrac{2}{5}]$.} Here, we have that $k_0=2$ and $k_1=1$. Clearly $\lceil \tau(k_0+k_1+3)\rceil=\lceil 6\tau\rceil= 3=k_0+1$ for each $\tau\in (\sfrac{1}{3},\sfrac{2}{5}]$. Further, for any $q\in \mathbb{Z}_+$, we have that $\lceil \tau(k_0+k_1+3+q)\rceil=\lceil \tau (6+q)\rceil \leq \lceil (6+q)\frac{2}{5}\rceil=\lceil 1.2+0.4q\rceil \leq q+2=k_0+q$.\\
		
		\noindent \underline{$\tau \in (\sfrac{2}{5},\sfrac{1}{2}].$} Here, we have that $k_0=2$ and $k_1=0$. Clearly, $\lceil \tau(k_0+k_1+3)\rceil=\lceil 5\tau\rceil=3=k_0+1$ for each $\tau\in (\sfrac{2}{5},\sfrac{1}{2}]$. Further, for any $q\in \mathbb{Z}_+$, we have that 
		$$\lceil \tau (k_0+k_1+3+q)\rceil=\lceil \tau (5+q)\rceil \leq \lceil \frac{5+q}{2}\rceil=2+\lceil\frac{q+1}{2}\rceil \leq 2+q$$ as for any $q>0$, we have that $\frac{q+1}{2}\leq q$. \\
		
		\noindent \underline{$\tau \in (\sfrac{1}{2},1]$.} Here, we have that $k_0=\in [\frac{4\tau-1}{1-\tau},\frac{3\tau}{1-\tau})$ and $k_1=0$. Consequently, we have that $\lceil \tau(k_0+k_1+3)\rceil=\lceil \tau(k_0+3)\rceil$. As $k_0\in [\frac{4\tau-1}{1-\tau},\frac{3\tau}{1-\tau})$, we have that $k_0\geq \frac{3\tau -1}{1-\tau}$. This implies that $k_0(1-\tau)\geq3\tau-1$ and consequently, $k_0+1\geq \tau(k_0+3)$. Further as $k_0<\frac{3\tau}{1-\tau}$, we have that $k_0<\tau(k_0+3)$. Therefore, it holds that $\lceil \tau (k_0+3)\rceil=k_0+1$. 
		
		We now prove that $\lceil \tau (k_0+3+q)\rceil\leq k_0+q$.  As $k_0\geq \frac{4\tau-1}{1-\tau}$, we have that $k_0(1-\tau)\geq 4\tau-1$ and thus $\tau(k_0+4)\leq k_0+1$. Further, as $k_0<\frac{3\tau}{1-\tau}< \frac{4\tau}{1-\tau}$, we have that $\tau(k_0+4)>k_0$. Thus, it holds that $\lceil \tau(k_0+4)\rceil=k_0+1$. Finally, for any fixed $q\in \mathbb{Z}_+$, we have that 
		
		\begin{align*}
			\lceil \tau(k_0+q+3)\rceil  &=\lceil \tau(k_0+4+q-1)\rceil\\
			&\leq \lceil \tau(k_0+4)\rceil +\lceil \tau(q-1)\rceil\\
			&< \lceil \tau(k_0+4)\rceil +q-1 \tag{As $\tau<1$}\\
			&=k_0+1+q-1=k_0+q
		\end{align*}

		Hence we have that for each $\tau\in (0,1)$, we have that $\lceil \tau(k_0+3+k_1)\rceil =k_0+1$ and for any $q\in \mathbb{Z}_+$, $\lceil \tau(k_0+k_1+3+q)\rceil \leq k_0+q$. \qed
		\vspace{4mm}
		
		\noindent \textbf{Claim 2.} For each $\tau\in (0,1),$ we have that for any $q\leq k_0+k_1+2$, $\lceil\tau q\rceil\leq k_0$.
		\vspace{2mm}
		
		\noindent {\em Proof of Claim 2.} As in the previous claim, we give a separate proof for each interval. In each case, it suffices to establish the claim for $q^*=k_0+k_1+2$ as for any $q<q^*$, $\lceil \tau q\rceil\leq \lceil \tau q^*\rceil$.\\
		
		\noindent \underline{$\tau\in (\frac{1}{p+1},\frac{1}{p}]$ for positive integer $p\geq 3$.} In this case, we have $k_0=1$ and $k_1=p-3$. Consequently, we need to check for $q^*=k_0+k_1+2=1+p-3+2=p$. As $\tau\leq \frac{1}{p}$, we have that $\lceil \tau q^*\rceil\leq 1=k_0.$\\
		
		\noindent \underline{$\tau \in (\sfrac{1}{3},\sfrac{2}{5}]$.} In this case, we have $k_0=2$ and $k_1=1$. Consequently, we need to check for $q^*=5$. As $\tau\leq \sfrac{2}{5}$, we have that $\lceil \tau q^*\rceil\leq 2=k_0$. \\
		
		\noindent \underline{$\tau \in (\sfrac{2}{5},\sfrac{1}{2}]$.} In this case, we have $k_0=2$ and $k_1=0$. Consequently, we need to check for $q^*=4$. As $\tau \leq \sfrac{1}{2}$, we have that $\lceil \tau q^*\rceil \leq 2=k_0$.\\
		
		\noindent \underline{$\tau \in (\sfrac{1}{2},1)$}. In this case, we have $k_0\in [\frac{4\tau-1}{1-\tau},\frac{3\tau}{1-\tau}) $ and $k_1=0$. Consequently, we need to establish that for all $\tau>\sfrac{1}{2}$, $\lceil \tau(k_0+2)\rceil \leq k_0$. In fact, it is sufficient to show that $\tau (k_0+2)\leq k_0$. We have that  
		
		\begin{align*}
			k_0  &\geq \frac{4\tau-1}{1-\tau}\\
			&=\frac{2\tau+2\tau-1}{1-\tau}\\
			&>\frac{2\tau}{1-\tau} \tag{As $\tau>\sfrac{1}{2}$.} 
		\end{align*}
		
		As a result, we have that $k_0(1-\tau)\geq 2\tau $ and therefore  $ \tau(k_0+2)\leq k_0$. 
		
		Hence, we have that for each $\tau \in (0,1)$, we have that for any any $q\leq k_0+k_1+2$, $\lceil\tau q\rceil\leq k_0$. \qed
		\vspace{4mm}

		\noindent \textbf{Claim 3.} For each $\tau\in (0,1)$, we have that $\lceil \tau (k_0'+1)\rceil=k_0'+1$ and for any $q\in \mathbb{Z}_+$, $\lceil \tau (k_0'+q+1)\rceil\leq k_0'+q$.
		\vspace{3mm}
		
		\noindent {\em Proof of Claim 3.} Recall that for all $\tau\leq \sfrac{1}{2}$, we have that $k_0'=0$. We thus argue for the cases when either $\tau\leq \sfrac{1}{2}$ or when $\tau>\sfrac{1}{2}$.\\
		
		\noindent \underline{$\tau\leq \sfrac{1}{2}.$} In this case $\lceil \tau (k_0'+1)\rceil=\lceil\tau\rceil=1=k_0'+1$ as $\tau>0$. Now fix $q\in \mathbb{Z}_+$. We have that $\lceil\tau(k_0'+q+1)\rceil=\lceil\tau(q+1\rceil)\leq \lceil\frac{q+1}{2}\rceil\leq q$ as $\tau \leq \sfrac{1}{2}$. \\
		
		\noindent \underline{$\tau\in (\sfrac{1}{2},1)$.} In this case we have that $k_0'\in [\frac{2\tau-1}{1-\tau},\frac{\tau}{1-\tau})$. Now as $\tau<1$, we have that  $\lceil \tau(k_0'+1)\rceil\leq k_0'+1$. Further as $k_0'<\frac{\tau}{1-\tau}$, we have that $k_0'(1-\tau)<\tau$ and thus $k_0'<\tau (k_0'+1).$ As a result, we have that $\lceil\tau (k_0'+1)\rceil=k_0'+1$.
		
		We now prove the remaining bound. To this end, recall that $k_0'\geq \frac{2\tau -1}{1-\tau}$. Therefore, we have that $k_0'(1-\tau)\geq 2\tau-1$ and thus, $\tau (k_0'+2)\leq k_0'+1$. As we have just shown $\tau (k_0'+1)>k_0'$. Consequently, we have that $k_0'<\tau(k_0'+2)\leq k_0'+1$ and as a result, $\lceil \tau (k_0'+2)\rceil=k_0'+1$. 
		
		Fix an arbitrary choice of $q\in \mathbb{Z}_+$, we obtain
		
		\begin{align*}
			\lceil \tau(k_0'+q+1)\rceil  &=\lceil \tau(k_0'+1+q-1)\rceil\\
			&\leq \lceil \tau(k_0'+1)\rceil +\lceil \tau(q-1)\rceil\\
			&< \lceil \tau(k_0'+1)\rceil +q-1 \tag{As $\tau<1$}\\
			&=k_0'+1+q-1=k_0'+q
		\end{align*}
		Hence, we have that for each $\tau\in (0,1)$, we have that $\lceil \tau (k_0'+1)\rceil=k_0'+1$ and for any $q\in \mathbb{Z}_+$, $\lceil \tau (k_0'+q+1)\rceil\leq k_0'+q$. \qed
		\vspace{3mm}
		
		We can now show that $\langle \mathcal{U},\mathcal{S}\rangle$ has an \textup{\textsc{Exact3Cover}} if and only if there exists an allocation $A=(A_1,\cdots, A_n)$ under the constructed instance $\langle N,M,v\rangle$ with \USW{} at least $2\ell-t$. \\
		
		\noindent {\bf X3C to maximum \USW{} allocation.} Firstly, assume that an Exact 3-Cover does exist, and that it is, without loss of generality, $S_1,\cdots, S_t$. Consider an allocation $A$ where for each $j\in [t]$, $A_{i_j}$ contains the element items corresponding to those contained in $S_j$, $k_0$ distinct dummy items and $k_1$ distinct buffer items. For $j>t$, $A_{i_j}$ contains $1$ prized item and $k_0'$ distinct dummy items. We now prove that  $\USW{}(A)=2\ell -t$. 
		
		Consider an agent $i_j$ where $j\in [t]$. We have that for agent $i_j$, $A_{i_j}$ contains $k_0$ items of value $0$, $3$ items of value $1$ and $k_1$ items of value $2$ and $|A_{i_j}|=k_0+3+k_1$. For the case of $\tau<0$, we have that $k_0=k_1=0$ and clearly $v_{i_j}(A_{i_j})=1$ as $A_{i_j}$ contains only the element items corresponding to the set $S_j$. In case $\tau\in (0,1)$, we have from Claim 1 that $\lceil \tau(k_0+3+k_1)\rceil=k_0+1$ and thus $v_{i_j}(A_{i_j})=1$.
		
		For an agent $i_j$ where $j>t$, we have that $A_{i_j}$ contains $k_0'$ items of value $0$ and exactly $1$ item of value $1$. When $\tau=0$, we have that $k_0'=0$ and clearly $v_{i_j}(A_{i_j})=2$. Instead, when $\tau\in (0,1)$, we have from Claim 3 that $\lceil \tau(k_0'+1)\rceil=k_0'+1$ and thus $v_{i_j}(A_{i_j})=2$. As a result, we have that $\USW{}(A)=t+2(\ell-t)=2\ell -t$.\\
		
		\noindent {\bf Maximum \USW{} allocation to X3C.} For the converse,  let an allocation $A$ exist s.t. $\USW{}(A)=2\ell -t$. We shall show that an Exact 3-Cover must exist.  
		Observe that there are exactly $\ell -t$ distinct (prized) items of value $2$, thus, under any allocation, at most $\ell -t$ agents can obtain a value of $2$. Further, as items values are all $0$, $1$ or $2$, if fewer than $\ell-t$ agents obtain a value of $2$, it is not possible to obtain  $\USW{}(A)=2\ell -t$. Consequently, there must exist $\ell-t$ agents who obtain a value of $2$.  Without loss of generality, let these be agents $i_{t+1},\cdots, i_{\ell}$. 
		
		Fix $j\in [\ell]\setminus [t]$. From the pigeonhole principle, $A_{i_j}$ must contain exactly one prized item, let this be $g'$. Note that when $\tau=0$, it must hold that $A_{i_j}=\{g'\}$, else $v_{i_j}(A_{i_j})\leq 1$.  Combining the fact that $A_{i_j}$ contains only one prized item with Claim 3,  we get that for $\tau>0$ if $|A_{i_j}|>k_0'+1$, it must hold that $v_{i_j}(A_{i_j})\leq 1$. As a result, $|A_{i_j}|\leq k_0'+1$ for all $\tau\in [0,1)$ and all $j>t$. As long as this holds, $v_{i_j}(A_{i_j})=2$ irrespective of whether the items in the set, other than $g'$ are element items, buffer items or dummy items. Recall that dummy items give value $0$ to all agents, but buffer and element items give value $1$ to at least one agent.  As a result, we can assume, without loss of generality, that $A_{i_j}\setminus D=\{g\}$.   Thus, the $\ell-t$ agents getting a value of $2$ can collectively receive at most $(\ell-t)k_0'$ dummy items from $D$.
		
		Therefore, for all choices of $\tau$, the remaining $k_0t$ dummy items, $3t$ element items and $k_1t$ buffer items must be allocated among the agents $i_1,\cdots, i_t$. As a result, at least one agent $i_j$ for $j\in [t]$ must satisfy $|A_{i_j}|\geq k_0+k_1+3$. Note that as $\USW{}(A)=2\ell-t$, it must hold that for each $j\in [t]$, $v_{i_j}(A_{i_j})=1$. From Claim 1, we have that  for any bundle $\beta \subseteq M$ s.t. $|\beta|\geq k_0+k_1+3$, for any agent $i\in N$, $v_i(\beta)=1$ if and only if $\beta$ contains at most $k_0$ items of value $0$. Recall that agents $i_1,\cdots, i_t$ must collectively receive $k_0t$ dummy items, which give each agent a value of $0$. From the pigeonhole principle, we get that for each $j\in [t]$, it must hold that $|A_{i_j}\cap D|=k_0$. 
		
		Combining this with Claim 2, we have that $|A_{i_j}|\geq k_0+k_1+q$. Again, applying pigeonhole principle, as there are exactly $3t$ element items and  $k_1t$ buffer items, it must hold that $|A_{i_j}|=k_0+k_1+3$ for each $j\in [t]$. Further, as for each $j\in [t]$,s $v_{i_j}(A_{i_j})=1$  and as $|A_{i_j}\cap D|=k_0$, $i_j$ cannot be assigned an element item not contained within $S_j$. Therefore, $A_{i_j}$ must contain at most $3$ element items, corresponding to those in $S_j$, and at least $k_1$ buffer items. As there are exactly $k_1t$ buffer items, $A_{i_j}$ must contain all three element items corresponding to $S_j$. As a result, the sets $S_1,\cdots, S_t$ must form an exact 3-cover.  
	\end{proof}
	
	We now complement this intractability result by providing an intriguing way of getting near optimal utilitarian welfare. 
	

	\subsubsection{Near Exact Algorithm.}
	In contrast to the balanced case, we find an asymptotically exact approximation for \USW{}. We call this the scapegoat algorithm (\cref{alg:unbalUSW-scapegoat}). It proceeds by considering allocations where one agent is the ``scapegoat" and receives $m-n+1$ items, while the remaining agents get one item each of high value. Exactly $n$ such allocations are considered, one for each agent as the scapegoat. For a fixed scapegoat, the corresponding allocation is built by a maximum weight one-one matching between the other agents and the items. 
	The algorithm chooses the allocation with the highest \USW{}.

	\begin{algorithm}[t]
		\KwIn{Instance with heterogeneous quantiles $\langle N,M, v, \tau \rangle$ }
		\KwOut{Allocation $A$}
		\For{each $i\in N$}{
			Create weighted bipartite graph $G^i=(X,Y,E,w)$ where\\
			$X$ contains $x_g$ for each $g\in M$,\\
			$Y$ contains $y_{j}$ for each $j\in N\setminus {i}$ and \\
			$w({x_g,y_{j}})= v_{j}(g)$\;
			Let $\mu$ be a maximum weight matching in $G^i$\;
			Set $A^i_j=\{g|x_g=\mu(y_j)\}$ for all $j\neq i$\;
			Set $A^i_i=M\setminus \cup_{j\neq i}A^i_j$\;
		}
		Let $A\gets \argmax \{\USW{}(A^i)|i\in N\}$\;
		\textbf{Return} $A$\;
		\caption{Scapegoat Algorithm for $\USW{\frac{n}{n-1}}$}\label{alg:unbalUSW-scapegoat}
	\end{algorithm}
	
	\begin{restatable}{theorem}{scapegoat}\label{thm:scapegoat}	
		Given instance  $I=\langle N,M,v,\tau \rangle$ with heterogeneous quantiles, scapegoat algorithm (\cref{alg:unbalUSW-scapegoat}) returns an $\USW{(\frac{n}{n-1})}$ allocation in polynomial time. 
	\end{restatable}
	
	\begin{proof} 
		Given $I$, let $A^i$ and $A$ be as in \cref{alg:unbalUSW-scapegoat} when run on $I$. Let $A^*$ be a maximum \USW{} allocation for $I$. Further, let $i^*\in N$ be such that $A=A^{i}$. 
		
		By definition of quantile valuations, for each $j\in N$, there is some $g_j\in A_j^*$ s.t. $v_j(A_j^*)=v_j(g_j)$. Without loss of generality we can assume that $v_1(A_1^*)\geq v_2(A_2^*)\geq \cdots \geq v_n(A_n^*)$. As a result, $\frac{n-1}{n}\USW{}(A^*)\leq \sum_{j=1}^{n-1} v_j(A^*_j)$.
		
		Now, as $A^{i}$ has maximum \USW{} over all the allocations constructed, it is straightforward to see that its \USW{} must be at least the weight of the best max weight matching constructed under \cref{alg:unbalUSW-scapegoat}. This matching in turn must have weight exactly equal to $\sum_{j=1}^{n-1} v_j(g_j)=\sum_{j=1}^{n-1}v_j(A^*_j)$. Thus, we get that
		\[\USW{}(A)\geq \sum_{j=1}^{n-1}v_j(g_j)\geq \frac{n-1}{n}\USW{}(A^*).\]
		
		Hence, $A$ is $\USW{(1+\frac{1}{n-1})}$.
		
		Note that since the maximum weight matching can be computed in $O(nm)$ time,  and such a matching is computed  $n$ times, our algorithm terminates in $O(mn^2)$ time.
	\end{proof}

	Building on this approach, we now show that when even one agent has $\tau_i=1$, we can now maximize \USW{} in poly time. Essentially this agent can be treated as the scapegoat, and we can simply use a maximum weight one-one matching as in \cref{alg:unbalUSW-scapegoat} and allocate all remaining items to the scapegoat. 
	
	\begin{restatable}{proposition}{optimistic}\label{prop:optimistUSWunbal}
		Given instance $I=\langle N, M, v,\tau \rangle$ with heterogeneous quantiles and an agent $i^*$ such that $\tau_{i^*}=1$, a maximum \USW{} allocation can be found in polynomial time.
	\end{restatable}
	
	\begin{proof}
		Given $I$, let $A^*$ be a maximum \USW{} allocation. For each $j\in N$, let $g_j\in A^*_j$ be such that $v_j(A^*_j)=v_j(g_j)$. 
		
		Consider a maximum weight matching $\mu$ on the bipartite graph with {\em all} agents and all items. Observe that the weight of $\mu$ is at least $\sum_j v_j(A^*_j)=\USW{}(A)$. Now define allocation $A$ to be such that for all $j\neq i^*$, they are allocated only their matched item under $\mu$. Agent $i^*$ is allocated the matched item under $\mu$ along with all remaining items. 
		
		Clearly $\USW{}(A)$ is the weight of $\mu$. Hence, $A$ has maximum \USW{}.
	\end{proof}

	
	\subsection{Egalitarian Social Welfare}
	Rather surprisingly, we  find that maximizing \ESW{} over all allocations is intractable for some quantiles and tractable for others. Here, we assume all agents have the same quantiles. Clearly, the intractability results would extend to settings with arbitrary heterogeneous valuations.  We illustrate the spectrum of quantiles for which the problem is tractable vs intractable in \cref{fig:ESWpvsnp}. 
	When presenting algorithms, we shall again assume binary valuations. From \cref{lem:binRednESW}, a polynomial time algorithm for the binary case is sufficient to get a general algorithm.
	
	\subsubsection{Exact Algorithms. } We are able to find polynomial time algorithms for maximizing \ESW{} under a class of quantiles which includes many natural quantiles like $\tau = 0,\sfrac{1}{2},\sfrac{2}{3},\sfrac{3}{4},\sfrac{9}{10}$. 
	We begin with an observation which is true for all quantiles: for maximum \ESW{} to be $1$, each agent must get at least one item of value $1$ simultaneously.

	\begin{observation}\label{obs:egalbasic}
		Under an instance with binary goods, for allocation $A$, $\ESW{}(A)=1$ if and only if for each $i\in N$, there exists $g\in A_i$, s.t. $v_i(g)=1$.
	\end{observation}
	
	This gives a necessary condition for an allocation with \ESW{} of 1 to exist.  We now specifically consider quantiles of the form $\tau=\frac{t}{t+1}$ for $k\in \mathbb{Z}_+$. For this setting, we have the following simple result.
	
	\begin{restatable}{lemma}{toffset}\label{lem:eswhalf}
		Given $i\in N$ with $\tau_i=\frac{t}{t+1}$, where $t\in \mathbb{Z}_+$ is fixed and a bundle $B\subseteq M$ s.t. $|B^1|=|\{g\in B|v_i(g)=1\}|=\ell$. We have that $v_i(B)=1\Leftrightarrow |B\setminus B^1|\leq \ell t-1$. 
	\end{restatable}
	
	\begin{proof}
		
		Given agent $i$ and bundle $B$ with exactly $\ell$ goods of value $1$ for $i$. Observe that it is sufficient to compare the case when there are either exactly $\ell  t-1$ items of value $0$ or $\ell t$ items of value $0.$
		
		Suppose the number of items of value $0$ is $\ell t$. The value of agent $i$ for $B$ would be from  the $\lceil (\ell t+\ell)\frac{t}{t+1}\rceil =\ell t$'th lowest valued item, which would have value $0$. Consequently $v_i(B)=0$.
		
		On the other hand, if $B$ contained  $\ell t-1$ items of value $0$, then $i$'s value would come from the item which has the $p$th lowest value where  
		
		\begin{align*}
			p &= \Bigl\lceil (\ell t-1+\ell)\frac{t}{t+1}\Bigr\rceil \\
			&= \Bigl\lceil\frac{\ell t^2+\ell t-t}{t+1}\Bigr\rceil \\
			&=\Bigl\lceil (\ell t- \frac{t}{t+1})\Bigr\rceil \\
			&= \ell t \tag{As $\frac{t}{t+1}<1$.}
		\end{align*}
		
		As a result, when there are at most $\ell t-1$ items of value $0$, $v_i(B)=1$.
	\end{proof}
	
	\begin{figure*}[t]
		\centering
		\tikzset{every picture/.style={line width=1pt}} 
		
		\begin{tikzpicture}[x=1pt,y=1pt,yscale=-1,xscale=1]
			
			\draw  [fill={rgb, 255:red, 82; green, 170; blue, 140 }  ,fill opacity=0.5 ]  (9,19) -- (309,19) -- (309,32) -- (9,32) -- cycle ;
			\draw    (109,7) -- (109,41) ;
			\draw    (159,7) -- (159,41) ;
			\draw    (209,7) -- (209,41) ;
			\draw    (234,7) -- (234,41) ;
			\draw    (249,7) -- (249,41) ;
			\draw    (279,7) -- (279,41) ;
			\draw [color={rgb, 255:red, 208; green, 2; blue, 27 }  ,draw opacity=1 ] [dash pattern={on 2pt off 2pt}]  (10,10) -- (10,40) ;
			\draw [color={rgb, 255:red, 208; green, 2; blue, 27 }  ,draw opacity=1 ] [dash pattern={on 2pt off 2pt}]  (11,11) -- (11,39) ;
			\draw [color={rgb, 255:red, 208; green, 2; blue, 27 }  ,draw opacity=1 ] [dash pattern={on 2pt off 2pt}]  (12,10) -- (12,40) ;
			\draw [color={rgb, 255:red, 208; green, 2; blue, 27 }  ,draw opacity=1 ] [dash pattern={on 2pt off 2pt}]  (13,11) -- (13,39) ;
			\draw [color={rgb, 255:red, 208; green, 2; blue, 27 }  ,draw opacity=1 ] [dash pattern={on 2pt off 2pt}]  (14,10) -- (14,40) ;
			\draw [color={rgb, 255:red, 208; green, 2; blue, 27 }  ,draw opacity=1 ] [dash pattern={on 2pt off 2pt}]  (15,11) -- (15,39) ;
			\draw [color={rgb, 255:red, 208; green, 2; blue, 27 }  ,draw opacity=1 ] [dash pattern={on 2pt off 2pt}]  (16,10) -- (16,40) ;
			\draw [color={rgb, 255:red, 208; green, 2; blue, 27 }  ,draw opacity=1 ] [dash pattern={on 2pt off 2pt}]  (17,11) -- (17,39) ;
			\draw [color={rgb, 255:red, 208; green, 2; blue, 27 }  ,draw opacity=1 ] [dash pattern={on 2pt off 2pt}]  (18,10) -- (18,40) ;
			\draw [color={rgb, 255:red, 208; green, 2; blue, 27 }  ,draw opacity=1 ] [dash pattern={on 2pt off 2pt}]  (19,11) -- (19,39) ;
			\draw [color={rgb, 255:red, 208; green, 2; blue, 27 }  ,draw opacity=1 ] [dash pattern={on 2pt off 2pt}]  (20,10) -- (20,40) ;
			\draw [color={rgb, 255:red, 208; green, 2; blue, 27 }  ,draw opacity=1 ] [dash pattern={on 2pt off 2pt}]  (21,11) -- (21,39) ;
			\draw [color={rgb, 255:red, 208; green, 2; blue, 27 }  ,draw opacity=1 ] [dash pattern={on 2pt off 2pt}]  (22,10) -- (22,40) ;
			\draw [color={rgb, 255:red, 208; green, 2; blue, 27 }  ,draw opacity=1 ] [dash pattern={on 2pt off 2pt}]  (23,11) -- (23,39) ;
			\draw [color={rgb, 255:red, 208; green, 2; blue, 27 }  ,draw opacity=1 ] [dash pattern={on 2pt off 2pt}]  (24,10) -- (24,40) ;
			\draw [color={rgb, 255:red, 208; green, 2; blue, 27 }  ,draw opacity=1 ] [dash pattern={on 2pt off 2pt}]  (25,11) -- (25,39) ;
			\draw [color={rgb, 255:red, 208; green, 2; blue, 27 }  ,draw opacity=1 ] [dash pattern={on 2pt off 2pt}]  (26,10) -- (26,40) ;
			\draw [color={rgb, 255:red, 208; green, 2; blue, 27 }  ,draw opacity=1 ] [dash pattern={on 2pt off 2pt}]  (27,11) -- (27,39) ;
			\draw [color={rgb, 255:red, 208; green, 2; blue, 27 }  ,draw opacity=1 ] [dash pattern={on 2pt off 2pt}]  (28,10) -- (28,40) ;
			\draw [color={rgb, 255:red, 208; green, 2; blue, 27 }  ,draw opacity=1 ] [dash pattern={on 2pt off 2pt}]  (29,11) -- (29,39) ;
			\draw [color={rgb, 255:red, 208; green, 2; blue, 27 }  ,draw opacity=1 ] [dash pattern={on 2pt off 2pt}]  (30,10) -- (30,40) ;
			\draw [color={rgb, 255:red, 208; green, 2; blue, 27 }  ,draw opacity=1 ] [dash pattern={on 2pt off 2pt}]  (31,11) -- (31,39) ;
			\draw [color={rgb, 255:red, 208; green, 2; blue, 27 }  ,draw opacity=1 ] [dash pattern={on 2pt off 2pt}]  (32,10) -- (32,40) ;
			
			\draw [color={rgb, 255:red, 208; green, 2; blue, 27 }  ,draw opacity=1 ] [dash pattern={on 2pt off 2pt}]  (33,11) -- (33,39) ;
			\draw [color={rgb, 255:red, 208; green, 2; blue, 27 }  ,draw opacity=1 ] [dash pattern={on 2pt off 2pt}]  (34,10) -- (34,40) ;
			\draw [color={rgb, 255:red, 208; green, 2; blue, 27 }  ,draw opacity=1 ] [dash pattern={on 2pt off 2pt}]  (35,11) -- (35,39) ;
			\draw [color={rgb, 255:red, 208; green, 2; blue, 27 }  ,draw opacity=1 ] [dash pattern={on 2pt off 2pt}]  (36,10) -- (36,40) ;
			\draw [color={rgb, 255:red, 208; green, 2; blue, 27 }  ,draw opacity=1 ] [dash pattern={on 2pt off 2pt}]  (37,11) -- (37,39) ;
			\draw [color={rgb, 255:red, 208; green, 2; blue, 27 }  ,draw opacity=1 ] [dash pattern={on 2pt off 2pt}]  (38,10) -- (38,40) ;
			\draw [color={rgb, 255:red, 208; green, 2; blue, 27 }  ,draw opacity=1 ] [dash pattern={on 2pt off 2pt}]  (39,11) -- (39,39) ;
			\draw [color={rgb, 255:red, 208; green, 2; blue, 27 }  ,draw opacity=1 ] [dash pattern={on 2pt off 2pt}]  (40,10) -- (40,40) ;
			\draw [color={rgb, 255:red, 208; green, 2; blue, 27 }  ,draw opacity=1 ] [dash pattern={on 2pt off 2pt}]  (41,11) -- (41,39) ;
			\draw [color={rgb, 255:red, 208; green, 2; blue, 27 }  ,draw opacity=1 ] [dash pattern={on 2pt off 2pt}]  (42,10) -- (42,40) ;
			\draw [color={rgb, 255:red, 208; green, 2; blue, 27 }  ,draw opacity=1 ] [dash pattern={on 2pt off 2pt}]  (43,11) -- (43,39) ;
			\draw [color={rgb, 255:red, 208; green, 2; blue, 27 }  ,draw opacity=1 ] [dash pattern={on 2pt off 2pt}]  (44,10) -- (44,40) ;
			\draw [color={rgb, 255:red, 208; green, 2; blue, 27 }  ,draw opacity=1 ] [dash pattern={on 2pt off 2pt}]  (45,11) -- (45,39) ;
			\draw [color={rgb, 255:red, 208; green, 2; blue, 27 }  ,draw opacity=1 ] [dash pattern={on 2pt off 2pt}]  (46,10) -- (46,40) ;
			\draw [color={rgb, 255:red, 208; green, 2; blue, 27 }  ,draw opacity=1 ] [dash pattern={on 2pt off 2pt}]  (47,11) -- (47,39) ;
			\draw [color={rgb, 255:red, 208; green, 2; blue, 27 }  ,draw opacity=1 ] [dash pattern={on 2pt off 2pt}]  (48,10) -- (48,40) ;
			\draw [color={rgb, 255:red, 208; green, 2; blue, 27 }  ,draw opacity=1 ] [dash pattern={on 2pt off 2pt}]  (49,11) -- (49,39) ;
			\draw [color={rgb, 255:red, 208; green, 2; blue, 27 }  ,draw opacity=1 ] [dash pattern={on 2pt off 2pt}]  (50,10) -- (50,40) ;
			\draw [color={rgb, 255:red, 208; green, 2; blue, 27 }  ,draw opacity=1 ] [dash pattern={on 2pt off 2pt}]  (51,11) -- (51,39) ;
			\draw [color={rgb, 255:red, 208; green, 2; blue, 27 }  ,draw opacity=1 ] [dash pattern={on 2pt off 2pt}]  (52,10) -- (52,40) ;
			\draw [color={rgb, 255:red, 208; green, 2; blue, 27 }  ,draw opacity=1 ] [dash pattern={on 2pt off 2pt}]  (53,11) -- (53,39) ;
			\draw [color={rgb, 255:red, 208; green, 2; blue, 27 }  ,draw opacity=1 ] [dash pattern={on 2pt off 2pt}]  (54,10) -- (54,40) ;
			\draw [color={rgb, 255:red, 208; green, 2; blue, 27 }  ,draw opacity=1 ] [dash pattern={on 2pt off 2pt}]  (55,11) -- (55,39) ;
			\draw [color={rgb, 255:red, 208; green, 2; blue, 27 }  ,draw opacity=1 ] [dash pattern={on 2pt off 2pt}]  (56,10) -- (56,40) ;
			\draw [color={rgb, 255:red, 208; green, 2; blue, 27 }  ,draw opacity=1 ] [dash pattern={on 2pt off 2pt}]  (57,11) -- (57,39) ;
			\draw [color={rgb, 255:red, 208; green, 2; blue, 27 }  ,draw opacity=1 ] [dash pattern={on 2pt off 2pt}]  (58,10) -- (58,40) ;
			\draw [color={rgb, 255:red, 208; green, 2; blue, 27 }  ,draw opacity=1 ] [dash pattern={on 2pt off 2pt}]  (59,11) -- (59,39) ;
			\draw [color={rgb, 255:red, 208; green, 2; blue, 27 }  ,draw opacity=1 ] [dash pattern={on 2pt off 2pt}]  (60,10) -- (60,40) ;
			\draw [color={rgb, 255:red, 208; green, 2; blue, 27 }  ,draw opacity=1 ] [dash pattern={on 2pt off 2pt}]  (61,11) -- (61,39) ;
			\draw [color={rgb, 255:red, 208; green, 2; blue, 27 }  ,draw opacity=1 ] [dash pattern={on 2pt off 2pt}]  (62,10) -- (62,40) ;
			\draw [color={rgb, 255:red, 208; green, 2; blue, 27 }  ,draw opacity=1 ] [dash pattern={on 2pt off 2pt}]  (63,11) -- (63,39) ;
			\draw [color={rgb, 255:red, 208; green, 2; blue, 27 }  ,draw opacity=1 ] [dash pattern={on 2pt off 2pt}]  (64,10) -- (64,40) ;
			\draw [color={rgb, 255:red, 208; green, 2; blue, 27 }  ,draw opacity=1 ] [dash pattern={on 2pt off 2pt}]  (65,11) -- (65,39) ;
			\draw [color={rgb, 255:red, 208; green, 2; blue, 27 }  ,draw opacity=1 ] [dash pattern={on 2pt off 2pt}]  (66,10) -- (66,40) ;
			\draw [color={rgb, 255:red, 208; green, 2; blue, 27 }  ,draw opacity=1 ] [dash pattern={on 2pt off 2pt}]  (67,11) -- (67,39) ;
			\draw [color={rgb, 255:red, 208; green, 2; blue, 27 }  ,draw opacity=1 ] [dash pattern={on 2pt off 2pt}]  (68,10) -- (68,40) ;
			\draw [color={rgb, 255:red, 208; green, 2; blue, 27 }  ,draw opacity=1 ] [dash pattern={on 2pt off 2pt}]  (69,11) -- (69,39) ;
			\draw [color={rgb, 255:red, 208; green, 2; blue, 27 }  ,draw opacity=1 ] [dash pattern={on 2pt off 2pt}]  (70,10) -- (70,40) ;
			\draw [color={rgb, 255:red, 208; green, 2; blue, 27 }  ,draw opacity=1 ] [dash pattern={on 2pt off 2pt}]  (71,11) -- (71,39) ;
			\draw [color={rgb, 255:red, 208; green, 2; blue, 27 }  ,draw opacity=1 ] [dash pattern={on 2pt off 2pt}]  (72,10) -- (72,40) ;
			\draw [color={rgb, 255:red, 208; green, 2; blue, 27 }  ,draw opacity=1 ] [dash pattern={on 2pt off 2pt}]  (73,11) -- (73,39) ;
			\draw [color={rgb, 255:red, 208; green, 2; blue, 27 }  ,draw opacity=1 ] [dash pattern={on 2pt off 2pt}]  (74,10) -- (74,40) ;
			\draw [color={rgb, 255:red, 208; green, 2; blue, 27 }  ,draw opacity=1 ] [dash pattern={on 2pt off 2pt}]  (75,11) -- (75,39) ;
			\draw [color={rgb, 255:red, 208; green, 2; blue, 27 }  ,draw opacity=1 ] [dash pattern={on 2pt off 2pt}]  (76,10) -- (76,40) ;
			\draw [color={rgb, 255:red, 208; green, 2; blue, 27 }  ,draw opacity=1 ] [dash pattern={on 2pt off 2pt}]  (77,11) -- (77,39) ;
			\draw [color={rgb, 255:red, 208; green, 2; blue, 27 }  ,draw opacity=1 ] [dash pattern={on 2pt off 2pt}]  (78,10) -- (78,40) ;
			\draw [color={rgb, 255:red, 208; green, 2; blue, 27 }  ,draw opacity=1 ] [dash pattern={on 2pt off 2pt}]  (79,11) -- (79,39) ;
			\draw [color={rgb, 255:red, 208; green, 2; blue, 27 }  ,draw opacity=1 ] [dash pattern={on 2pt off 2pt}]  (80,10) -- (80,40) ;
			\draw [color={rgb, 255:red, 208; green, 2; blue, 27 }  ,draw opacity=1 ] [dash pattern={on 2pt off 2pt}]  (81,11) -- (81,39) ;
			\draw [color={rgb, 255:red, 208; green, 2; blue, 27 }  ,draw opacity=1 ] [dash pattern={on 2pt off 2pt}]  (82,5) -- (82,44) ;
			\draw [color={rgb, 255:red, 208; green, 2; blue, 27 }  ,draw opacity=1 ] [dash pattern={on 2pt off 2pt}]  (176,11) -- (176,39) ;
			\draw [color={rgb, 255:red, 208; green, 2; blue, 27 }  ,draw opacity=1 ] [dash pattern={on 2pt off 2pt}]  (177,10) -- (177,40) ;
			\draw [color={rgb, 255:red, 208; green, 2; blue, 27 }  ,draw opacity=1 ] [dash pattern={on 2pt off 2pt}]  (178,11) -- (178,39) ;
			\draw [color={rgb, 255:red, 208; green, 2; blue, 27 }  ,draw opacity=1 ] [dash pattern={on 2pt off 2pt}]  (179,10) -- (179,40) ;
			\draw [color={rgb, 255:red, 208; green, 2; blue, 27 }  ,draw opacity=1 ] [dash pattern={on 2pt off 2pt}]  (180,11) -- (180,39) ;
			\draw [color={rgb, 255:red, 208; green, 2; blue, 27 }  ,draw opacity=1 ] [dash pattern={on 2pt off 2pt}]  (181,10) -- (181,40) ;
			\draw [color={rgb, 255:red, 208; green, 2; blue, 27 }  ,draw opacity=1 ] [dash pattern={on 2pt off 2pt}]  (182,11) -- (182,39) ;
			\draw [color={rgb, 255:red, 208; green, 2; blue, 27 }  ,draw opacity=1 ] [dash pattern={on 2pt off 2pt}]  (183,10) -- (183,40) ;
			\draw [color={rgb, 255:red, 208; green, 2; blue, 27 }  ,draw opacity=1 ] [dash pattern={on 2pt off 2pt}]  (184,11) -- (184,39) ;
			\draw [color={rgb, 255:red, 208; green, 2; blue, 27 }  ,draw opacity=1 ] [dash pattern={on 2pt off 2pt}]  (185,10) -- (185,40) ;
			\draw [color={rgb, 255:red, 208; green, 2; blue, 27 }  ,draw opacity=1 ] [dash pattern={on 2pt off 2pt}]  (186,11) -- (186,39) ;
			\draw [color={rgb, 255:red, 208; green, 2; blue, 27 }  ,draw opacity=1 ] [dash pattern={on 2pt off 2pt}]  (187,10) -- (187,40) ;
			\draw [color={rgb, 255:red, 208; green, 2; blue, 27 }  ,draw opacity=1 ] [dash pattern={on 2pt off 2pt}]  (188,11) -- (188,39) ;
			\draw [color={rgb, 255:red, 208; green, 2; blue, 27 }  ,draw opacity=1 ] [dash pattern={on 2pt off 2pt}]  (121,10) -- (121,40) ;
			\draw [color={rgb, 255:red, 208; green, 2; blue, 27 }  ,draw opacity=1 ] [dash pattern={on 2pt off 2pt}]  (122,11) -- (122,39) ;
			\draw [color={rgb, 255:red, 208; green, 2; blue, 27 }  ,draw opacity=1 ] [dash pattern={on 2pt off 2pt}]  (123,10) -- (123,40) ;
			\draw [color={rgb, 255:red, 208; green, 2; blue, 27 }  ,draw opacity=1 ] [dash pattern={on 2pt off 2pt}]  (124,11) -- (124,39) ;
			\draw [color={rgb, 255:red, 208; green, 2; blue, 27 }  ,draw opacity=1 ] [dash pattern={on 2pt off 2pt}]  (125,10) -- (125,40) ;
			\draw [color={rgb, 255:red, 208; green, 2; blue, 27 }  ,draw opacity=1 ] [dash pattern={on 2pt off 2pt}]  (126,11) -- (126,39) ;
			\draw [color={rgb, 255:red, 208; green, 2; blue, 27 }  ,draw opacity=1 ] [dash pattern={on 2pt off 2pt}]  (127,10) -- (127,40) ;
			\draw [color={rgb, 255:red, 208; green, 2; blue, 27 }  ,draw opacity=1 ] [dash pattern={on 2pt off 2pt}]  (128,11) -- (128,39) ;
			\draw [color={rgb, 255:red, 208; green, 2; blue, 27 }  ,draw opacity=1 ] [dash pattern={on 2pt off 2pt}]  (129,5) -- (129,44) ;
			\draw [color={rgb, 255:red, 208; green, 2; blue, 27 }  ,draw opacity=1 ] [dash pattern={on 2pt off 2pt}]  (130,10) -- (130,40) ;
			\draw [color={rgb, 255:red, 208; green, 2; blue, 27 }  ,draw opacity=1 ] [dash pattern={on 2pt off 2pt}]  (131,11) -- (131,39) ;
			\draw [color={rgb, 255:red, 208; green, 2; blue, 27 }  ,draw opacity=1 ] [dash pattern={on 2pt off 2pt}]  (132,10) -- (132,40) ;
			\draw [color={rgb, 255:red, 208; green, 2; blue, 27 }  ,draw opacity=1 ] [dash pattern={on 2pt off 2pt}]  (133,11) -- (133,39) ;
			\draw [color={rgb, 255:red, 208; green, 2; blue, 27 }  ,draw opacity=1 ] [dash pattern={on 2pt off 2pt}]  (134,10) -- (134,40) ;
			\draw [color={rgb, 255:red, 208; green, 2; blue, 27 }  ,draw opacity=1 ] [dash pattern={on 2pt off 2pt}]  (135,11) -- (135,39) ;
			\draw [color={rgb, 255:red, 208; green, 2; blue, 27 }  ,draw opacity=1 ] [dash pattern={on 2pt off 2pt}]  (136,10) -- (136,40) ;
			\draw [color={rgb, 255:red, 208; green, 2; blue, 27 }  ,draw opacity=1 ] [dash pattern={on 2pt off 2pt}]  (137,11) -- (137,39) ;
			\draw [color={rgb, 255:red, 208; green, 2; blue, 27 }  ,draw opacity=1 ] [dash pattern={on 2pt off 2pt}]  (138,10) -- (138,40) ;
			\draw [color={rgb, 255:red, 208; green, 2; blue, 27 }  ,draw opacity=1 ] [dash pattern={on 2pt off 1pt}]  (139,11) -- (139,39) ;
			\draw [color={rgb, 255:red, 208; green, 2; blue, 27 }  ,draw opacity=1 ] [dash pattern={on 2pt off 2pt}]  (140,10) -- (140,40) ;
			\draw [color={rgb, 255:red, 208; green, 2; blue, 27 }  ,draw opacity=1 ] [dash pattern={on 2pt off 2pt}]  (141,11) -- (141,39) ;
			\draw [color={rgb, 255:red, 208; green, 2; blue, 27 }  ,draw opacity=1 ] [dash pattern={on 2pt off 2pt}]  (142,10) -- (142,40) ;
			\draw [color={rgb, 255:red, 208; green, 2; blue, 27 }  ,draw opacity=1 ] [dash pattern={on 2pt off 2pt}]  (143,11) -- (143,39) ;
			\draw [color={rgb, 255:red, 208; green, 2; blue, 27 }  ,draw opacity=1 ] [dash pattern={on 2pt off 2pt}]  (189,5) -- (189,44) ;
			\draw    (9,7) -- (9,41) ;
			\draw    (309,7) -- (309,41) ;
			\draw    (294,14) -- (294,37) ;
			\draw    (266.14,14) -- (266.14,37) ;
			\draw    (271.5,14) -- (271.5,37) ;
			\draw    (275.67,14) -- (275.67,37) ;
			\draw    (281.727,14) -- (281.727,37) ;
			\draw    (284,14) -- (284,37) ;
			\draw    (286,14) -- (286,37) ;
			\draw    (287.57,14) -- (287.57,37) ;
			\draw    (291.353,14) -- (291.353,37) ;
			\draw    (292.333,14) -- (292.333,37) ;
			\draw    (290.25,14) -- (290.25,37) ;
			\draw    (259,14) -- (259,37) ;
			\draw    (271.5,14) -- (271.5,37) ;
			\draw    (289,14) -- (289,37) ;
			\draw    (304,14) -- (304,37) ;
			\draw    (300,14) -- (300,37) ;
			\draw    (293,14) -- (293,37) ;
			\draw    (294,14) -- (294,37) ;
			\draw    (295,14) -- (295,37) ;
			\draw    (296,14) -- (296,37) ;
			\draw    (297,14) -- (297,37) ;
			\draw    (298,14) -- (298,37) ;
			\draw    (299,14) -- (299,37) ;
			\draw    (300,14) -- (300,37) ;
			\draw    (301,14) -- (301,37) ;
			\draw    (302,14) -- (302,37) ;
			\draw    (303,14) -- (303,37) ;
			\draw    (304,14) -- (304,37) ;
			\draw    (305,14) -- (305,37) ;
			\draw    (306,14) -- (306,37) ;
			\draw    (307,14) -- (307,37) ;
			\draw    (308,14) -- (308,37) ;
			
			\draw (78,45) node [anchor=north west][inner sep=0.75pt]  [font=\footnotesize]  {$\frac{1}{4}$};
			\draw (105,45) node [anchor=north west][inner sep=0.75pt]  [font=\footnotesize]  {$\frac{1}{3}$};
			\draw (155,45) node [anchor=north west][inner sep=0.75pt]  [font=\footnotesize]  {$\frac{1}{2}$};
			\draw (125,45) node [anchor=north west][inner sep=0.75pt]  [font=\footnotesize]  {$\frac{2}{5}$};
			\draw (185,45) node [anchor=north west][inner sep=0.75pt]  [font=\footnotesize]  {$\frac{3}{5}$};
			\draw (205,45) node [anchor=north west][inner sep=0.75pt]  [font=\footnotesize]  {$\frac{2}{3}$};
			\draw (229,45) node [anchor=north west][inner sep=0.75pt]  [font=\footnotesize]  {$\frac{3}{4}$};
			\draw (245,45) node [anchor=north west][inner sep=0.75pt]  [font=\footnotesize]  {$\frac{4}{5}$};
			\draw (273,45) node [anchor=north west][inner sep=0.75pt]  [font=\footnotesize]  {$\frac{9}{10}$};
			\draw (305,47) node [anchor=north west][inner sep=0.75pt]  [font=\footnotesize]  {$1$};
			\draw (5,47) node [anchor=north west][inner sep=0.75pt]  [font=\footnotesize]  {$0$};

		\end{tikzpicture}

		\caption{Quantile-wise tractability or intractability of max \ESW{}. Red dashed lines show values of $\tau$ for which maximizing \ESW{} is NP-hard, black solid lines show a value for $\tau$ for which we have polytime algorithms.}
		\label{fig:ESWpvsnp}
	\end{figure*}
	Based on \cref{obs:egalbasic} and \cref{lem:eswhalf} we develop an algorithm for maximizing \ESW{} over unconstrained allocations when there is a $t\in \mathbb{Z}_+$ s.t. $\tau_i=\frac{t}{t+1}$ for each $i\in N$. 
	\cref{alg:ESW-unbal} divides the items into sets $M_0$ and $M_1$. $M_0$ contains items that give all agents a value of $0$ and $M_1$ contains the remaining item. The algorithm proceeds by first checking if i) all agents can simultaneously receive one item of value $1$, and then if for the items in $M_0$, if there are enough items in $M_1$ to ensure the condition in \cref{lem:eswhalf} is met. If so, the remaining items are greedily allocated while ensuring that the intermediate partial allocations all have an \ESW{} of 1. 
	
	{
		\begin{algorithm}[t]
			\KwIn{$I=\langle N,M, v,\tau \rangle$ with binary goods and $\tau=\sfrac{t}{t+1}$ }
			\KwOut{An allocation $A$}
			Create bipartite graph $G=(X,Y,E)$ where $X$ contains $x_g$ for each $g\in M$,
			$Y$ contains $y_i$ for each $i\in N$ and 
			$(x_g,y_i)\in E$ only if $v_i(g)=1$, for each $i\in N,\, g\in M$\;
			Let $\mu$ be a maximum cardinality matching in $G$\; 
			Let $M_0=\{g\in M|v_i(g)=0$ for all $i\in N\}$\;
			Let $M_1=M\setminus M_0$\;\label{step:defn} 
			\eIf{$|M_0|> t|M_1|-n$ OR $|\mu|<n$}{
				Let $A$ be an arbitrary allocation\;
			}
			{
				Let $A=(A_1,\cdots, A_n)$ be s.t. $A_i\gets\{g|(x_g,y_i)\in \mu\}$\;\label{step:else}
				$M_1\gets M_1\setminus \cup_i A_i$\;
				\While{$M_1\neq \emptyset$ AND $M_0\neq \emptyset$}{
					Arbitrarily pick $g\in M_1$ and $i\in N$ s.t. $v_i(g)=1$\;
					\eIf{$|M_0|\geq t$}{
						Pick an arbitrary subset $S\subseteq M_0$ s.t. $|S|=t$\;
					}{
						Let $S\gets M_0$\;
					}
					$A_i\gets A_i\cup \{g\}\cup S$\;
					$M_1\gets M_1\setminus \{g\}$ and 
					$M_0\gets M_0 \setminus S$\;
				}
				\If{$M_0\neq \emptyset$}{
					Let $B_1\cdots B_n$ be an arbitrary partition of $M_0$ s.t. $|B_i|\leq t-1$ for all $i\in N$\;
					
				}
				\If{$M_1\neq \emptyset$}{
					Let $B_1\cdots B_n$ be an arbitrary partition of $M_1$ s.t. $g\in|B_i|$ only if $v_i(g)=1$\;
				}
				For each $i\in N$, set $A_i\gets A_i \cup B_i$\;
				
			}
			\textbf{Return} $A$\;    
			
			\caption{Max \ESW{} for binary goods and $\tau=\sfrac{t}{t+1}$}\label{alg:ESW-unbal}
		\end{algorithm}
	}
	
	\begin{restatable}{proposition}{ESWunbal}\label{prop:ESWunbal}
		Given instance $I=\langle N,M,v,\tau=\frac{t}{t+1}\rangle$ where $t\in \mathbb{Z}_+$, \cref{alg:ESW-unbal} returns a maximum \ESW{} allocation in polynomial time.
	\end{restatable}
	
	\begin{proof}
		We now show that given an instance with binary goods $I=\langle N,M,v, \tau=\frac{t}{t+1}\rangle$, \cref{alg:ESW-unbal} finds an allocation with \ESW{} of $1$ whenever it exists.  
		
		Let $\mu$, $M_0$ and $M_1$ be as defined in \cref{alg:ESW-unbal} by step \ref{step:defn}.  We shall now show that whenever an allocation of \ESW{} $1$ exists, \cref{alg:ESW-unbal} will return an allocation with \ESW{} $1$. We first show that when $|\mu|=n$ and $|M_0|\leq t|M_1|-n$, \cref{alg:ESW-unbal} creates an allocation where if agent $i\in N$ receives $\ell_i>1$ items of value $1$ then they receive at most $t\ell_i-1$ items of value $0$. We have that $\ell _i>1$ as $\mu$ matches each agent to an item of value $1$. 
		
		Further, in the while loop, whenever $i$ receives at most $t$ items of value $0$ from $M_0$, they are accompanied with one item of value $1$. After the while loop, an additional $t-1$ items of value $0$ may be  allocated to $i$. As a result, from \cref{lem:eswhalf} $v_i(A_i)=1$ in this case.
		
		Consequently, when $|\mu|=n$ and $|M_0|\leq t|M_1|-n$, we have that \cref{alg:ESW-unbal} finds an allocation with $\ESW{}(A)=1$.
		
		Conversely, assume that an allocation $A^*$ exists s.t. $\ESW{}(A^*)=1$. We show that it must hold that $|\mu|=n$ and $|M_0|\leq t|M_1|-n$. Now, clearly each agent $i$ must receive at least one good of value $1$, thus we have that $|\mu|=n$. 
		
		Now let $A^*_{i,0}$ and $A^*_{i,1}$ respectively denote the $0$ and $1$ valued items $i$ is allocated under $A^*$. We have that $M_0\subseteq \cup_{i\in N} A^*_{i,0}$ and $\cup_{i\in N} A^*_{i,1}\leq M_1$. 
		
		As $v_i(A^*_i)=1$, from \cref{lem:eswhalf}, we have that $|A_{i,0}^*|\leq t|A^*_{i,1}|-1$. Consequently, we have that 
		\[|M_0|\leq \sum_i |A_{i,0}^*|\leq \sum_i t|A^*_{i,1}|-1\leq t|M_1| - n  \]
		
		Hence, we have that the necessary conditions will be satisfied and \cref{alg:ESW-unbal} will return an allocation of \ESW{} $1$. As a result, \cref{alg:ESW-unbal} will return an allocation with \ESW{} $1$ if and only if one exists.
	\end{proof}

	We can extend this idea to the setting of $\tau_i=\sfrac{1}{3}$ for all agents. Here, we need to pick agents and items a lot more carefully. 
	
	\paragraph{$\sfrac{1}{3}$ quantile.} We now consider the case where $\tau=\sfrac{1}{3}$. To this end, we begin with the following simple observation, analogous to \cref{lem:eswhalf}.
	
	\begin{observation}\label{obs:third}
		For an agent $i\in N$ with $\tau_i=\sfrac{1}{3}$, a bundle $B\subseteq M$ with exactly $\ell $ items of value $0$ for $i$, we have that $v_i(B)=1$ if and only if the number of $1$ valued items in $B$ is at least $2\ell +1$.
	\end{observation}
	Thus, when $M_0$ and $M_1$ are as in \cref{alg:ESW-unbal}, we need two items from $M_1$ to offset one from $M_0$. We can now build \cref{alg:ESW-unbalthird} where we need to check if we can satisfy both \cref{obs:egalbasic,obs:third}.

	\begin{algorithm}[t]
		\KwIn{$I=\langle N,M, v,\tau \rangle$ with binary goods and $\tau=\sfrac{1}{3}$ }
		\KwOut{An allocation $A$}
		Let $M_0=\{g\in M|v_i(g)=0$ for all $i\in N\}$\;
		Let $M_1=M\setminus M_0$\;
		Create graph $G=(X \cup Y, E)$ where $X$ contains $x_g$ for each $g\in M_1$, 
		$Y$ contains $y_i$ for each $i\in N$ and $(x_g,y_i)\in E$ only if $v_i(g)=1$ and 
		$(x_g,x_{g'})\in E$ only if there exists $i\in N$ s.t. $v_i(g)=v_i(g')=1$\;
		Define edge weight function $w$ where $w(x,y)=|X\cup Y|+1$ and $w(y,y')=1$\;
		Let $\mu$ be a maximum weight matching in $G_2$\; 
		
		\eIf{$w(\mu)<|M_0|+n(|X\cup Y|+1)$}{
			Let $A$ be an arbitrary allocation\;
		}
		{
			Let $A=(A_1,\cdots, A_n)$ be such that $A_i=\{g|(x_g,y_i)\in \mu\}$\;
			$M_1\gets M_1\setminus \cup_i A_i$\;
			\While{$M_0\neq \emptyset$}{
				Arbitrarily pick $g_0\in M_0$ and $g,g'$ s.t. $(x_{g},x_{g'})\in \mu$\;
				Pick $i\in N$ s.t. $v_i(g)=v_i(g')=1$\; 
				$A_i\gets A_i\cup \{g_0,g,g'\}$\;
				$M_1\gets M_1\setminus \{g,g'\}$ and 
				$M_0\gets M_0 \setminus \{g_0\}$\;
			}
			\If{$M_1\neq \emptyset$}{
				Let $B_1\cdots B_n$ be an arbitrary partition of $M_1$ s.t. $g\in|B_i|$ only if $v_i(g)=1$\;
			}
			For each $i\in N$, set $A_i\gets A_i \cup B_i$\;
			
		}
		\textbf{Return} $A$\;    
		
		\caption{Max \ESW{} binary goods for $\tau=\sfrac{1}{3}$}\label{alg:ESW-unbalthird}
	\end{algorithm}

	\begin{restatable}{proposition}{ESWunbalthird}\label{prop:ESWunbalthird}
		Given $I=\langle N,M,v,\tau=\sfrac{1}{3}\rangle$, \cref{alg:ESW-unbalthird} returns a maximum \ESW{} allocation in polynomial time.
	\end{restatable}
	
	\begin{proof}
		We now show that given an instance with binary goods $I=\langle N,M,v, \tau=\sfrac{1}{3}\rangle$, \cref{alg:ESW-unbalthird} finds an allocation with \ESW{} of $1$ whenever it exists.  
		
		Let $\mu$, $M_0$ and $M_1$ be as initially defined in \cref{alg:ESW-unbalthird}.  We shall now show that whenever an allocation of \ESW{} $1$ exists, \cref{alg:ESW-unbal} will return an allocation with \ESW{} $1$. We first show that when $w(\mu)\geq|M_0|+n(|X\cup Y|+1)$, \cref{alg:ESW-unbalthird} creates an allocation where if agent $i\in N$ receives $\ell_i$ items of value $0$, then they receive at least $2\ell_i+1$ items of value $1$. 
		
		First, $\mu$ matches each agent to one item of value $1$, so for agents with $\ell_i=0$, the requirement is satisfied.     
		Further, in the while loop, whenever $i$ receives two items of value $1$ for every item from $M_0$. After the while loop,  only items of value $1$ may be  allocated to $i$. As a result, $v_i(A_i)=1$ in this case.
		
		Consequently, when $w(\mu)\geq|M_0|+n(|X\cup Y|+1)$, we have that \cref{alg:ESW-unbalthird} finds an allocation with $\ESW{}(A)=1$.
		
		Conversely, assume that an allocation $A^*$ exists s.t. $\ESW{}(A^*)=1$. 
		Now let $A^*_{i,0}$ and $A^*_{i,1}$ respectively denote the $0$ and $1$ valued items $i$ is allocated under $A^*$. We have that $M_0\subseteq \cup_{i\in N} A^*_{i,0}$ and $\cup_{i\in N} A^*_{i,1}\leq M_1$. 
		
		As $v_i(A^*_i)=1$, from \cref{lem:eswhalf}, we have that $|A_{i,1}^*|\geq 2|A^*_{i,0}|+1$. Consequently, we have can build a matching $\mu'$  in $G_2$ matching $|A^*_{i,0}|$ pairs of items from $A^*_{i,1}$ to each other and one additional item to $i$. Now as $\mu_2$ is a maximum weight matching in $G_2$, it must have weight at least  
		\begin{align*}
			w(\mu_2) &\geq w(\mu')\\
			&\geq \sum_i (|A^*_{i,0}| +  |X\cup Y|+1)\\
			&= n(|X\cup Y|+1)+\sum_i |A^*_{i,0}| \\
			&\geq  n(|X\cup Y|+1)+|M_0|.
		\end{align*}
		
		Hence, we have that the necessary condition will be satisfied and \cref{alg:ESW-unbal} will return an allocation of \ESW{} $1$. As a result, \cref{alg:ESW-unbal} will return an allocation with \ESW{} $1$ if and only if one exists.
	\end{proof}

	We can now summarize our tractability results for maximum \ESW{} over all allocations as follows. 
	
	\begin{theorem}\label{thm:ESWunbal}
		A maximum \ESW{} allocation can be found in polynomial time for $\tau=\{0,\sfrac{1}{3},1\}\cup \{\frac{t}{t+1}|t\in \mathbb{Z}_+\}$.
	\end{theorem}

	\subsubsection{Intractability.} We now show that there are several quantile values for which maximizing \ESW{} is APX-hard. Intriguingly, these values interweave between quantile values for which maximizing \ESW{} can be done in polynomial time. We find three ranges of intractability. Namely, for $\tau\in (0,\sfrac{1}{4}]\cup (\sfrac{3}{8},\sfrac{2}{5}]\cup (\sfrac{5}{9},\sfrac{3}{5}]$. 
	For $\tau\in (\sfrac{3}{8},\sfrac{2}{5}]$ or $\tau\in (\sfrac{5}{9},\sfrac{3}{5})$ the ratio of additional items of value $1$ for a new item of value $0$ can vary. We find that deciding between these cases proves to be intractable. 
	We show intractability for binary valuations for each range. Under binary valuations, any $\alpha>0$ approximation on \ESW{} would be an exact algorithm. Consequently, we get that it is NP-hard to have any $\ESW{\alpha}$ algorithm for all $\alpha>0$.
	
	\begin{theorem}\label{thm:ESWunbalNP}
		Given $I=\langle N,M,v,\tau\rangle$, maximizing \ESW{} is APX-hard for $\tau\in (0,\frac{1}{4}]\cup (\sfrac{3}{8},\sfrac{2}{5}]\cup (\sfrac{5}{9},\sfrac{3}{5}]$.
	\end{theorem}
	
	For each interval, we provide separate reduction from X3C, in the same vein as \cref{thm:NP-USW}. Different ranges require different number of items of low and high values, consequently, separate reductions.

	
	\section{Identical Valuations}\label{sec:identical}
	
	We now consider identical valuations, that is all agents have the same quantile $\tau$ and the same valuation function $v$. We defer all omitted proofs to \cref{app:identical}.
	
	
	\subsection{Utilitarian Welfare}

	Maximizing \USW{} remains open for the unconstrained case. A maximum \USW{} balanced allocation can be found in polynomial time for any $\tau\in [0,1]$.  In fact, we have that the same greedy algorithm (\cref{alg:goods-USWbalanced}) that was $\USW{\min (\frac{m}{n}+1,n)}$ that proves to be an exact algorithm in this case.
	
	\begin{restatable}{theorem}{USWbalIdent}\label{thm:USWbalIdent}
		Given an instance with identical valuations $I=\langle N,M,v,\tau\rangle$, \cref{alg:goods-USWbalanced} returns a balanced allocation with maximum \USW{}.
	\end{restatable}
	
	\begin{proof}
		Given $I$ with identical valuation $v$, let the items be such that $v(g_1)\geq v(g_2)\geq\cdots \geq v(g_m)$. Let $A$ be the allocation returned by \cref{alg:goods-USWbalanced}. 
		Recall that the quantiles are identical, and consequently we have that in each round, each agent demands the same sized set $k' =\min(k,k-\lceil \tau k\rceil+1)$  where $k$ is as in \cref{alg:goods-USWbalanced}.  Let $k'= \min(k,k-\lceil \tau k\rceil+1)$.
		
		Let the order in which agents are first assigned their demanded set under \cref{alg:goods-USWbalanced} be $i_1,i_2,\cdots i_n$. Without loss of generality, we may assume that the set demanded by agent $i_t$ is $S_{i_t}=\{g_{(t-1)k'+1},\cdots,g_{tk'}\}$.  Thus, $v(A_{i_t}) =  v(g_{tk'})$ for each $t\in [n]$. 
		
		Let $A^*$  be  a balanced allocation with maximum \USW{}. As agents have identical valuations, we assume without loss of generality that $v(A^*_{1})\geq v(A^*_{2})\geq \cdots \geq v(A^*_{n})$. We now show that we can also assume without loss of generality that $|A^*_{1}|\leq |A_{2}^*|\leq \cdots \leq |A_{n}^*|$. \\
		
		\noindent \textbf{Sorted by size.} Assume for contradiction that there exist $j_1,j_2\in [n]$, s.t. $j_1<j_2$ and $|A_{{j_1}}^*|>|A_{{j_2}}^*|$. As $A^*$ is a balanced allocation, it must hold that $|A_{{j_1}}^*|=\lfloor \frac{m}{n}\rfloor<\lceil \frac{m}{n}\rceil=|A_{{j_2}}^*|$. First, consider the case where $\lceil\tau \lfloor\frac{m}{n}\rfloor\rceil<\lceil \tau \lceil \frac{m}{n}\rceil\rceil$. Choose $g\in \argmin \{v(g)|g\in A_{{j_1}}^*\}$. Observe that $v(A_{{j_1}}^*)\leq v(A_{{j_1}}^*\setminus \{g\})$. Further for the set $A_{{j_{2}}}^*\cup \{g\}$, we have that if $v(g)<v(A_{{j_{2}}}^*)$, we have that $v(A_{{j_{2}}}^*\cup \{g\})=v(A_{{j_{2}}}^*)$. Otherwise, if $v(g)\geq A_{{j_{2}}}^*$, the value of the new set $$v(A^*_{{j_2}}\cup \{g\})\geq \min (v(g),v(g_{\lceil \tau\lfloor\frac{m}{n}\rfloor\rceil+1}))\geq v((g_{\lceil \tau\lfloor\frac{m}{n}\rfloor\rceil}))=v(A_{{j_2}}^*)$$ where $g_1,\cdots, g_{\lfloor\frac{m}{n}\rfloor }$ are the items in $A_{i_{j_2}}$ s.t. $v(g_1)\leq \cdots v(g_{\lfloor\frac{m}{n}\rfloor })$.

		The remaining case is when $\lceil\tau \lceil\frac{m}{n}\rceil\rceil=\lceil\tau \lfloor\frac{m}{n}\rfloor\lceil$. In this case, it holds that $\lceil\tau \lceil\frac{m}{n}\rceil\rceil=\lceil\tau \lfloor\frac{m}{n}\rfloor\lceil<\lceil\frac{m}{n}\rceil$. Here, choose $g\in \argmax \{v(g)|g\in A_{{j_1}}^*\}$. Observe that $v(A_{{j_1}}^*)=v(A_{{j_1}}^*\setminus\{g\})$ in this case. Further, as $v(g)\geq v(A_{i_{j_1}}^*)\geq v(A_{{j_2}}^*)$, we have that $v(A_{{j_2}}^*\cup \{g\})=v(A_{{j_2}}^*)$ as the $\lceil \tau \lfloor \frac{m}{n}\rfloor\rceil^{\text{th}}$ smallest item is the same in both sets. As a result, we can also assume that  $|A^*_{1}|\leq |A_{2}^*|\leq \cdots \leq |A_{n}^*|$. Consequently, as \cref{alg:goods-USWbalanced} first allocates sets of size $\lfloor \frac{m}{n}\rfloor$, we have that for each $t\in [n]$, $|A_{i_t}|=|A^*_t|$.  We now prove correctness of our algorithm.\\
		
		\noindent \textbf{Correctness of algorithm. } We now show that  $v(A_{i_t})\geq v(A^*_{t})$ for all  $t\in [n]$. For each $t\in [n]$, choose $k_t$ to be the value of $k'$ in the $t^{\text{th}}$ iteration of the while-loop. That is, $k_t$ is the size of the set demanded by agent $i_t$ when it is allocated an item for the first time. Further, relabel the items s.t. $v(g_1)\geq v(g_2)\geq \cdots \geq v(g_m)$. Consequently, we have that $v(A_{i_t})=v(g_{\sum_{t'=1}^t k_{t'}})$. 
		
		Suppose, for the sake of contradiction, that there exists $\ell\in [n]$ such that $v(A_{i_\ell})< v(A^*_{\ell})$. Specifically, let $\ell$ be the smallest such index. It follows that, we have  $v(A^*_{\ell})> v(g_{\sum_{t=1}^\ell k_t})$.   Since $v(A^*_{i_1})\geq \cdots \geq v(A^*_{i_{\ell-1}})\geq v(A^*_{i_\ell})$, we see that  in $A^*$ the number of agents who get  value strictly higher than $v(g_{\sum_{t=1}^\ell k_t})$ is at least $\ell$.  
		
		However, as  $|A^*_t | =|A_{i_t}|$ for all $t\in [n]$, for each $t\leq \ell$, $A^*_t$ contains at least $k_t$ items of value at least $v(g_{\sum_{t=1}^\ell k_t})$. This implies  that  the number of items value strictly more that $v(g_{\sum_{t=1}^\ell k_t})$ is at most $ ({\sum_{t=1}^\ell k_t})-1 $. Consequently, it must hold that $v(A^*_{\ell})\leq v(g_{\sum_{t=1}^\ell k_t})=v(A_{i_t})$. This contradicts the assumption that $v(A^*_{\ell})>v(A_{i_{\ell}})$.  
		
		Consequently, for each $t\in [n]$, we must have that $v(A_{i_t})\geq v(A^*_{i_t})$. Hence, $\USW{}(A)\geq \USW{}(A^*)$.
	\end{proof}
	
	\subsection{Egalitarian Welfare}
	We again focus our attention to binary goods, as a consequence of \cref{lem:binRednESW}. The balanced case is already covered by \cref{thm:balESW}. We find that a maximum \ESW{} allocation in the unconstrained case can be found in polynomial time. To this end, we make the following observation.
	
	\begin{observation}\label{obs:stepfn}
		Given $\tau\in (0,1]$ and a $\tau$-quantile binary valuation function $v$ and bundle $B\subset M$, let $B_0=\{g\in B \ : \ v(B)=0\}$. We have that $v(g)=1 \Leftrightarrow |B|> \frac{|B_0|}{\tau}$.
	\end{observation}
	
	We use this to show in order to maximize \ESW{} under identical valuations, it suffices to consider allocations where the number of items of value $0$ is balanced across agents.
	
	\begin{lemma}\label{lem:ESWidentical}
		Given instance $I=\langle N,M,v, \tau\rangle$ with identical binary valuations, let $B_1=|\{g\in M|v(g)=1\}|$. There exists $A'$ s.t. $\ESW{}(A')=1$ only if there exists an $A$ where  $|A\setminus B_1|\leq\lceil (m-|B_1|)/n \rceil$ and $v(A_i)=1$ for each $i\in N$. 
	\end{lemma}
	
	\begin{proof}
		Let there exist an allocation $A=(A_1,\cdots, A_n)$ s.t. $v(A_i)=1$ for all $i\in N$. Let $t_i$ denote the number of goods of value $0$ in $A_i$. 
		If $\tau=0$, then if $r\neq m$, no allocation can exist where all agents get value $1$. It follows that any allocation with \ESW{} of 1, satisfies the required property.  
		
		Now consider the case where $\tau>0$. Let there exist an agent $j$ s.t $t_j>\lceil \frac{m-r}{n} \rceil$. Then there must exist an agent $j'$ s.t. $t_j -t_{j'}\geq 2$. Thus, $t_{j'}<\lfloor \frac{m-r}{n} \rfloor$.  
		We shall show that there exists an allocation $A'$ where $j$ gets $t_j-1$ items of value $0$, and $j'$ gets $t_{j'}+1$ items of value $0$ and $\ESW{}(A)=1$.
		
		Consider $A'$ where $A'_i=A_i$ for all $i\neq j, j'$. We shall now transfer one item of value $0$ and just enough items of value $1$ from $j$ to $j'$ to get the required allocation.  From \cref{obs:stepfn}, for $t_{j'}+1$ items of value $0$, in order for $v(A'_{j'})=1$, it must be the case that  $|A'_{j'}|> \frac{t_{j'}+1}{\tau}$.
		If $|A_{j'}|+1>\frac{t_{j'}+1}{\tau}$, we let $A'_{j'}=A_{j'}\cup g_0$ such that $v(g_0)=0$ and $g_0\in A_j$. Finally, let  $A'_j=A_j \setminus g_0$, observe that $v(A'_j)=v(A'_{j'})=1$.
		
		Suppose $|A_{j'}|+1\leq \frac{t_{j'}+1}{\tau}$.
		Now as $v(A_{j'})=1$, by \cref{obs:stepfn}, it must be that $|A_{j'}|>\frac{t_{j'}}{\tau}$.  We shall show that there are enough goods of value $1$ in $A_j$ that can be transferred while maintaining the values of both bundles. 
		
		Choose $\ell=\lceil\frac{1}{\tau}\rceil-1$ goods $\{g_1,..., g_\ell\}$ of value $1$ and $g_0$ of $v(g_0)=0$ from  $A_j$.   Let $A'_{j'}= A_{j'} \cup \{g_0,g_1,..., g_\ell\} $, and $A'_j=A_j \setminus \{g_0,g_1,..., g_\ell\} $. 
		Observe that as $v(A_{j'})=1$, we have that  $|A_{j'}|> \frac{t_{j'}}{\tau}$, it follows that
		
		$$|A'_{j'}|=|A_{j'}|+1+\ell \geq |A_{j'}|+\frac{1}{\tau}> \frac{t_j'+1}{\tau}. $$
		Thus, we have $v(A'_{j'})=1 $. We now show that $v(A'_j)=1$.
		By assumption, $v(A_j)=1$. Consequently, we have that $\lceil \tau |A_j|\rceil \geq t_j+1$.  Now consider $A'_j$, we need $\lceil \tau (|A_j|-(1+\ell))\rceil\geq t_j $. Consider
		
		\begin{align*}
			\lceil \tau |A'_j|\rceil&=\lceil\tau (|A_j|-(1+\ell))\rceil\\   
			&\geq \lceil \tau (|A_j| - \frac{1}{\tau})\rceil = \lceil \tau |A_j|  -   1            \rceil =\lceil \tau|A_j|\rceil -1 \\
			&\geq t_j+1 -1=t_j.
		\end{align*}
		
		Hence, we have that $\ESW{}(A')=1$. Consequently, we can repeat this procedure till each agent has at most $\lceil \frac{m-r}{n}\rceil$ goods of value $0$. As a result, whenever an allocation exists s.t. $\ESW{}(A)=1$, there must exist an allocation $A'$, s.t. $\ESW{}(A')=1$ and each agent receives either $\lfloor (m-r)/n\rfloor$ or $\lceil (m-n)/r\rceil$ items of value $0$.
	\end{proof}
	This result proves useful in maximizing \ESW{} in the unconstrained setting and even leads to an algorithm for maximum \USW{} for binary goods.
	\begin{restatable}{theorem}{ESWident}\label{thm:ESWident}
		Given instance $I=\langle N,M,v, \tau\rangle$ with identical valuations, an allocation with maximum \ESW{} can be found in polynomial time. 
		
	\end{restatable}
	
	\begin{proof}
		
		Given $I$, as a consequence of \cref{lem:binRednESW}, we assume $I$ is an instance with binary goods.  
		
		First, consider the set of items of value $0$, that is, $M_0=\{g\in M|v(g)=0\}$. From \cref{lem:ESWidentical}, we know that it is sufficient to consider only allocations where items in $M_0$ are distributed uniformly. Let $t=|M_0|-n\lfloor\frac{|M_0|}{n}\rfloor$. That is, $t$ is the number of agents who need to receive more than $\lfloor\frac{|M_0|}{n}\rfloor$ items from $M_0$.
		
		Now, consider the set of items of items of value $1$, that is $M_1=M\setminus M_0$. If an agent receives $\ell$ items from $M_0$, by \cref{obs:stepfn}, we can easily calculate the minimum bundle size $k$ s.t. $\lceil \tau k \rceil >\ell$. In particular, $k=\min \{k'\in \mathbb{Z}_+|k>\frac{\ell}{\tau}\}$. Consequently, let $k_1=\min \{k'\in \mathbb{Z}_+|k>\frac{\lfloor |M_0|/n\rfloor}{\tau}\}-\lfloor |M_0|/n\rfloor$ and $k_2=\min \{k'\in \mathbb{Z}_+|k>\frac{\lceil |M_0|/n\rceil}{\tau}\}-\lceil |M_0|/n\rceil$.
		
		As a result, an agent receiving $\lfloor |M_0|/n\rfloor$ items from $M_0$ requires at least $k_1$ items from $M_1$ to have value $1$. Analogously, an agent receiving $\lceil |M_0|/n\rceil$ items from $M_0$ requires at least $k_2$ items from $M_1$ to have value $1$.  Thus, an allocation of \ESW{} $1$ exists, if and only if $|M_1|\geq tk_2+(n-t)k_1$. 
		
		This can easily be checked and an appropriate allocation can accordingly be built. Thus, a maximum \ESW{} allocation can be found in polynomial time. 
	\end{proof}
	
	\begin{restatable}{proposition}{binUSWIdent}\label{prop:binUSWIdent}
		Given instance $I=\langle N,M,v, \tau\rangle$ with identical valuations over binary goods, an allocation with maximum \USW{} can be found in polynomial time.
	\end{restatable}
	
	\begin{proof}
		Let $M_0$ and $M_1$ be the sets of items of value $0$ and $1$, respectively.
		
		Observe that with binary goods, an allocation with \USW{} $n$ exists, if and only if an allocation with \ESW{} $1$ exists. Thus, we first find a maximum \ESW{} allocation. If it has \ESW{} $1$, it must have \USW{} $n$.  
		
		If an \ESW{} $1$ allocation does not exist, we can check if $|M_1|>n-1$. If so we can give $n-1$ agents exactly one item each from $M_1$ and give all remaining items in $M_1$ and $M_0$ to the remaining agent.
		
		Finally, if $|M_1|\leq n-1$, we can give $|M_1|$ agents one item each from $M_1$ and the items in $M_0$ are distributed arbitrarily among the remaining agents. In this case, no agent can get an item from both $M_1$ and $M_0$.
		
		It is easy to see that this can be done in polynomial time. 
	\end{proof}

	
	\section{Conclusions and Future Work}
	In this work,  we proposed a novel quantile-based preference model  in the context of indivisible item allocation.
	We studied {\em Utilitarian} and {\em Egalitarian Welfare}, both with and without the balanced allocation requirement, and provided comprehensive algorithmic and complexity-theoretic results.
	Interestingly, our results reveal that the complexity of the problems changes significantly depending on whether the balancedness requirement is imposed. For instance, for  balanced allocations there is a strong hardness of approximation bound for maximizing \USW{}, whereas for unconstrained allocations, a near-exact approximation algorithm exists.  A similar phenomenon occurs with  \ESW{} but in reverse: for balanced allocations, maximizing \ESW{} can be solved efficiently, while for unconstrained allocations, maximizing \ESW{} is APX-hard  for many quantile values. Further, we observed a simple intractability versus tractability threshold for utilitarian welfare under both balanced and unconstrained allocations. Meanwhile, for Egalitarian welfare, there is no simple threshold separating tractability and intractability under the unconstrained case.  
	
	Our work opens up several promising directions for future research. Firstly, while we focused on the two extremes of the p-means (Utilitarian and Egalitarian welfare), exploring other welfare functions, such as Nash welfare, presents an intriguing avenue for study. There are a variety of open questions here such as whether Nash welfare continues to satisfy the strong fairness properties it does under monotone valuations. Further, achieving simultaneous guarantees across the family of p-means welfare functions is also a well-motivated direction. Another relevant question here is to understand the welfare guarantees provided by strategyproof mechanisms under quantile-valuations.   
	
	Secondly, while we focus on Rawlsian fairness, other envy-based fairness notions like envy-freeness, proportionality, maximin share remain open. Specifically, investigating the compatibility between fairness notions, such as EF1 or EFx, and Pareto efficiency within the framework of our valuation class is another interesting direction of further research. 
	
	Another complementary model would be studying quantile valuations with chores or mixtures of goods and chores. In fact, quantile-based preferences can also be defined for divisible goods, and thus the investigation carried out in this paper as well as the other proposed directions of further work can also be carried out for the divisible setting, for settings with mixtures of divisible and indivisible items as well as for the more recently studied setting of items with subjective divisibility. Similarly, new valuation models can also be built combining quantiles with the more typically studied valuation classes. Therefore, we believe that our paper can lead to many worthwhile avenues of future work.
	
	\section*{Acknowledgements}
	
	Haris Aziz and Shivika Narang are supported by the NSF-CSIRO project on “Fair Sequential Collective Decision-Making”. Mashbat Suzuki is supported by the ARC Laureate Project FL200100204 on “Trustworthy AI”.
	
	\bibliographystyle{plainnat}
	\bibliography{references.bib}
	
	\clearpage
	\appendix
	

	\section{Omitted Proofs from Section 4}\label{app:unbalanced}
	
	\subsubsection{Intractability.}
	
	\begin{lemma}
		Given an instance with binary goods $I=\langle N,M,v,\tau\in (\sfrac{1}{5},\sfrac{1}{4}]\rangle$, maximizing \ESW{} is NP-hard.
	\end{lemma}
	\begin{proof}
		We shall give a reduction from \textup{\textsc{Exact3Cover}} (X3C) \footnote{See the proof of \cref{thm:NP-USW} for a definition of the exact 3 cover problem}. Given an instance of X3C $\langle \mathcal{U}, \mathcal{S}\rangle$ where $|\mathcal{U}|=3t$ and $|\mathcal{S}|=\ell $, we shall create an instance of our problem as follows: 
		
		For each $S_j\in \mathcal{S}$, create a set agent $i_j$ and a set item $g_j$.
		
		For each element $u\in\mathcal{U}$, we create an element item $g_u$.
		
		Create $t$ dummy items $g'_1,\cdots, g'_t$.
		
		As a result, we have created $\ell $ agents and $\ell +4t$ items. We define agent preferences as follows. For any $i_j\in N$, $v_{i_j}(g')=0$ for any dummy item $g'$. Further, for any $u\in \mathcal{U}$, if $u\in S_j$, set $v_{i_j}(g_u)=1$, otherwise set $v_{i_j}(g_u)=0$. Finally, for any $j'\in [\ell ]$, set $v_{i_j}(g_{j'})=1$ if $j=j'$ otherwise, set $v_{i_j}(g_{j'})=0$. Lastly, choose $\tau_i$ arbitrarily from the range $\in (\sfrac{1}{5},\sfrac{1}{4}]$, for all $i\in N$.
		
		We shall now show that an allocation with \ESW{} of 1 exists if and only if the given X3C instance has an exact 3-cover.
		
		First, assume that $S_{j_1},\cdots, S_{j_t}$ form an exact 3 cover of $\langle \mathcal{U},\mathcal{S}\rangle$. Consider the following allocation $A$
		
		\[A_{i_j}=\begin{cases}
			\{g_j\} & \text{ if } j\notin \{j_1\cdots,j_t\}\\
			\{g_j,g'_p\}\cup \{g_u|u\in S_j\} & \text{ if } j=j_p \text{ for some } p\in [t]
		\end{cases}\]
		
		That is, agents corresponding to sets in the exact 3 cover receive their set item, one dummy item and constituent items. Agents corresponding to sets not in the exact 3 cover only receive their corresponding set item.      
		Firstly, observe that as $S_{j_1},\cdots, S_{j_t}$ form an exact 3 cover of $\langle \mathcal{U},\mathcal{S}\rangle$, $A$ must be a valid allocation, where all items are allocated, and the bundles of agents are disjoint. Now, for any agent $i_j$, where $j\notin \{j_1\cdots,j_t\}$, we have that $A_{i_j}=\{g_j\}$, thus, $v_{i_j}(A_{i_j})=1$.
		
		Further, for $j\in \{j_1\cdots,j_t\}$, $A_{i_j}$ contains one item of value $0$, the dummy item $g'$ and $4$ items of value $1$, the set item and element items. As, $\tau\in (\sfrac{1}{5},\sfrac{1}{4}]$, we have that $v_{i_j}(A_{i_j})=1$. Consequently, $\ESW{}(A)=1$.
		
		Now, conversely, assume that an allocation $A^*$ exists s.t. $\ESW{}(A^*)=1$. That is, for each $i\in N$, $v_i(A^*_i)=1$. 
		
		Observe that for each agent there are exactly four items of value $1$: the corresponding set item and constituent element items. As $\tau_i\in (\sfrac{1}{5},\sfrac{1}{4}]$, $A_i^*$ can contain at most $1$ item of value $0$ for $i$. If it does contain one item of value $0$, all four of the value $1$ items must also be contained in order to ensure $v_i(A_i^*)=1$.  
		
		In particular, as there are $t$ dummy items for which each agent has value $0$, each agent can be allocated at most one dummy item under $A^*$. Let the set of agents who receive one dummy item be $i_{j_1},\cdots,i_{j_t}$. Further, each for $j\in \{ {j_1},\cdots,{j_t}\}$,  we have that $A_{i_j}$ must contain $g_j$ and all three items in $\{g_u|u\in S_j\}$. As a result, $S_{j_1},\cdots, S_{j_t}$ must be mutually disjoint. Consequently, $S_{j_1},\cdots, S_{j_t}$ form an exact 3 cover of $\langle \mathcal{U},\mathcal{S}\rangle$.  
		
		Hence, the problem of finding a maximum \ESW{} allocation is NP-hard for binary goods and $\tau\in (\sfrac{1}{5},\sfrac{1}{4}]$.
	\end{proof}
	
	We can do an analogous reduction from the $t$-dimensional matching problem for $t\geq 3$, to an instance with binary goods and $\tau\in (\frac{1}{t+2},\frac{1}{t+1})$ where 1 item of value $0$ needs to be offset by $t$ items of value $1$ to ensure that the bundle has value $1$ for the corresponding agent.
	\begin{corollary}
		Given $I=\langle N,M,v,\tau\in (0,\sfrac{1}{4}]\rangle$ with binary goods, maximizing \ESW{} is NP-hard.
	\end{corollary}

	\paragraph{Intractability with $\tau \in  (\sfrac{3}{8},\sfrac{2}{5}]$.} The main source of intractability in this range of quantiles comes from differing number of value $1$ items that can are needed to offset an additional item of value $0$. Considering only bundles that give value $1$ to an agent,  with four items of value $1$, there can be at most two items of value $0$. However with five items of value $1$ there can be at most three items of value $0$. Thus when the number items which give value $0$ is strictly greater than the number of items that give value $1$ to at least one agent, deciding if an \ESW{} 1 allocation may not be possible with polynomially many greedy decisions. 
	
	\begin{lemma}
		Finding a maximum \ESW{} allocation is NP-hard for $\tau\in (\sfrac{3}{8},\sfrac{2}{5}]$.
	\end{lemma}

	\begin{proof} 
		
		We shall give a reduction from \textup{\textsc{Exact3Cover}} (X3C).  Given an instance of X3C $\langle \mathcal{U}, \mathcal{S}\rangle$ where $|\mathcal{U}|=3t$ and $|\mathcal{S}|=\ell$, we shall create an instance of our problem as follows: 
		\begin{itemize} 
			\item For each $S_j\in \mathcal{S}$, we create a set agent $i_j$ and a set items $g_j^1$ and $g_j^2$.
			\item For each element $u\in\mathcal{U}$, we create an element item $g_u$.
			\item Create $\ell + 2t$ dummy items $g'_1,\cdots, g'_{\ell+2t}$.
		\end{itemize}
		
		As a result, we have created $n=\ell$ agents and $m=3\ell+5t$ items. We define agent preferences as follows. For any $i_j\in N$, $v_{i_j}(g')=0$ for any dummy item $g'$. Further, for any $u\in \mathcal{U}$, if $u\in S_j$, set $v_{i_j}(g_u)=1$, otherwise set $v_{i_j}(g_u)=0$. Finally, for any $j'\in [\ell ]$, set $v_{i_j}(g_{j'}^1)=v_{i_j}(g_{j'}^2)=1$ if $j=j'$ otherwise, set $v_{i_j}(g_{j'}^1)=v_{i_j}(g_{j'}^2)=0$. Lastly, set $\tau_i=\sfrac{2}{5}$, for all $i\in N$.
		
		We shall now show that an allocation with \ESW{} of 1 exists if and only if the given X3C instance has an exact 3-cover.
		
		First, assume that $S_{j_1},\cdots, S_{j_t}$ form an exact 3 cover of $\langle \mathcal{U},\mathcal{S}\rangle$. Consider the following allocation $A$ where $A_{i_j}= \{g_j^1,g_j^2, g'_{j}\}$ if $j\notin \{j_1\cdots,j_t\}$ and $A_{i_j}=\{g_j^1,g_j^2,$ $g'_{j},g'_{\ell +2p-1},g'_{\ell +2p}\}\cup \{g_u|u\in S_j\}$ if $j=j_p$ for some $p\in [t]$

		That is, agents corresponding to sets in the exact 3 cover receive their set items, three dummy items and their constituent items. Agents corresponding to sets not in the exact 3 cover only receive their corresponding set items and one dummy item.      
		Firstly, observe that as $S_{j_1},\cdots, S_{j_t}$ form an exact 3 cover of $\langle \mathcal{U},\mathcal{S}\rangle$, $A$ must be a valid allocation, where all items are allocated, and the bundles of agents are disjoint. Now, for any agent $i_j$, where $j\notin \{j_1\cdots,j_t\}$, we have that $A_{i_j}=\{g_j^1,g_j^2,g'_{j}\}$. As $\tau_i>\sfrac{1}{3}$, we have that, $v_{i_j}(A_{i_j})=1$.
		
		Further, for $j\in \{j_1\cdots,j_t\}$, $A_{i_j}$ contains three items of value $0$, the dummy items $g'_{j}$,  $g'_{\ell +2p-1},g'_{\ell +2p}$ and five items of value $1$, the set items and element items. As, $\tau_i>\sfrac{3}{8}$, we have that $v_{i_j}(A_{i_j})=1$. Consequently, $\ESW{}(A)=1$.
		
		Now, conversely, assume that an Exact 3 Cover does not exist. 
		
		Observe that for each agent there are exactly five items of value $1$: the corresponding set items and constituent element items. This along with the fact that $\tau_i\leq\sfrac{2}{5}$ implies that any bundle of value $1$ for $i$ can contain at most three items of value $0$ for $i$. If it does contain three items of value $0$, all five of the value $1$ items must also be contained to ensure the bundle has value $1$. 
		
		Recall that there are $\ell+2t$ dummy items for which each agent has value $0$. These items can only be offset by $2\ell$ set items and $3t$ element items. 
		Further, as for all agents $\tau_i\in (\sfrac{3}{8},\sfrac{2}{5}]$, we have that if an agent received $p<5$ items of value $1$, they must have at most $\lfloor p/2\rfloor$ items of value $0$. 
		
		Consequently, in order to offset all the dummy items, we must have at least $t$ agents who each receive three dummy items and their corresponding set item and constituent items.  
		Now as no exact 3 cover exists, at most $t-1$ (set) agents can receive all three constituent element items. Thus, no allocation exists with an \ESW{} of $1$.
	\end{proof}

	\paragraph{Intractability with $\tau \in  (\sfrac{5}{9},\sfrac{3}{5}]$.} The main source of intractability in this range of quantiles comes from differing number of value $0$ items that can be added with an additional item of value $1$. Considering only bundles that give value $1$ to an agent,  with three items of value $1$, there can be at most three items of value $0$. However with four items of value $1$ there can be at most five items of value $0$. Thus when the number items which give value $0$ is strictly greater than the number of items that give value $1$ to at least one agent, deciding if an \ESW{} 1 allocation may not be possible with polynomially many greedy decisions. 
	
	\begin{lemma}
		Finding a maximum \ESW{} allocation is NP-hard for $\tau\in (\sfrac{5}{9},\sfrac{3}{5}]$.
	\end{lemma}
	
	\begin{proof} 
		
		We shall give a reduction from \textup{\textsc{Exact3Cover}} (X3C).  Given an instance of X3C $\langle \mathcal{U}, \mathcal{S}\rangle$ where $|\mathcal{U}|=3t$ and $|\mathcal{S}|=\ell $, we shall create an instance of our problem as follows: 
		
		For each $S_j\in \mathcal{S}$, we create a set agent $i_j$ and a set item $g_j$.
		
		For each element $u\in\mathcal{U}$, we create an element item $g_u$.
		
		Create $\ell +4t$ dummy items $g'_1,\cdots, g'_{\ell +4t}$.
		
		Thus, we have created $\ell $ agents and $2\ell +7t$ items. We define agent preferences as follows. For any $i_j\in N$, $v_{i_j}(g')=0$ for any dummy item $g'$. Further, for any element $u\in \mathcal{U}$, if $u\in S_j$, set $v_{i_j}(g_u)=1$, otherwise set $v_{i_j}(g_u)=0$. Finally, for any $j'\in [\ell ]$, set value for set item $g_{j'}$ $v_{i_j}(g_{j'})=1$ if $j=j'$ otherwise, set $v_{i_j}(g_{j'})=0$. Lastly, arbitrarily set $\tau_i\in (\sfrac{5}{9},\sfrac{3}{5}]$, for all $i\in N$.
		
		We shall now show that an allocation with \ESW{} of 1 exists if and only if the given X3C instance has an exact 3-cover.
		
		First, assume that $S_{j_1},\cdots, S_{j_t}$ form an exact 3 cover of $\langle \mathcal{U},\mathcal{S}\rangle$. Consider the following allocation $A$ where $A_{i_j}= \{g_j, g'_{j}\}$ if $j\notin \{j_1\cdots,j_t\}$ and $A_{i_j}=\{g_j,g'_{j},g'_{\ell +4p-3},$ $g'_{\ell +4p-2},g'_{\ell +4p-1},g'_{\ell +4p}\}\cup \{g_u|u\in S_j\}$ if $j=j_p$ for some $p\in [t]$

		That is, agents corresponding to sets in the exact 3 cover receive their set item, five dummy items and their constituent items. Agents corresponding to sets not in the exact 3 cover only receive their corresponding set item and one dummy item.      
		Firstly, observe that as $S_{j_1},\cdots, S_{j_t}$ form an exact 3 cover of $\langle \mathcal{U},\mathcal{S}\rangle$, $A$ must be a valid allocation, where all items are allocated, and the bundles of agents are disjoint. Now, for any agent $i_j$, where $j\notin \{j_1\cdots,j_t\}$, we have that $A_{i_j}=\{g_j,g'_{j}\}$. As $\tau_i>0.5$, we have that, $v_{i_j}(A_{i_j})=1$.
		
		Further, for $j\in \{j_1\cdots,j_t\}$, $A_{i_j}$ contains five items of value $0$, the dummy items $g'_{j}$, $g'_{\ell +4p-3}$, $g'_{\ell +4p-2}$, $g'_{\ell +4p-1},g'_{\ell +4p}$ and $4$ items of value $1$, the set item and element items. As, $\tau_i>\sfrac{5}{9}$, we have that $v_{i_j}(A_{i_j})=1$. Consequently, $\ESW{}(A)=1$.
		
		Now, conversely, assume that an Exact 3 Cover does not exist. 
		
		Observe that for each agent there are exactly four items of value $1$: the corresponding set item and constituent element items. This along with the fact that $\tau_i\leq\sfrac{3}{5}$ implies that any bundle of value $1$ for $i$ can contain at most $5$ items of value $0$ for $i$. If it does contain five items of value $0$, all four of the value $1$ items must also be contained to ensure the bundle has value $1$. 
		
		Recall that there are $\ell +4t$ dummy items for which each agent has value $0$. These items can only be offset by $d$ set items and $3t$ element items. Clearly there are fewer items that can give an agent value $1$ than the number of items that give all agents value $0$. Further, as for all agents $\tau_i\in (\sfrac{5}{9},\sfrac{3}{5}]$, we have that if an agent received $p<4$ items of value $1$, they must have at most $p$ items of value $0$. 
		
		Consequently, in order to offset all the dummy items, we must have at least $t$ agents who each receive $5$ dummy items and their corresponding set item and constituent items.  
		Now as no exact 3 cover exists, at most $t-1$ (set) agents can receive all three constituent element items. Thus, no allocation exists with an \ESW{} of $1$.
	\end{proof}

\end{document}